%% file: ForcedFoldBifurcations.tex
\documentclass[10pt,a4paper,reqno]{amsart}

\usepackage{amsmath,amssymb,amscd} \usepackage{a4wide}     
\usepackage[utf8]{inputenc} \usepackage[english]{babel}
\usepackage{graphicx,subfigure}
\usepackage[font=footnotesize]{caption}
\usepackage{mwe}
\usepackage{xcolor}

\newtheorem{intro_theorem}{Theorem}
\newtheorem{theorem}{Theorem}[section]
\newtheorem{lem}[theorem]{Lemma}
\newtheorem{cor}{Corollary}[theorem]

\theoremstyle{definition}

\newtheorem{remark}[theorem]{Remark}

\input{bibstandard}

\newcommand{\etal}{+}

\makeatletter
\renewenvironment{proof}[1][\proofname]{%
   \par\pushQED{\qed}\normalfont%
   \topsep6\p@\@plus6\p@\relax
   \trivlist\item[\hskip\labelsep\bfseries#1\@addpunct{.}]%
   \ignorespaces
}{%
   \popQED\endtrivlist\@endpefalse
}
\makeatother

\begin{document}

\title[Forcing and fold bifurcations in population dynamics]{On the effect of
  forcing on fold bifurcations and early-warning signals in population dynamics}

\author{F.~Remo}
\address{Institute of Mathematics, Friedrich Schiller University Jena, Germany}
\email{flavia.remo@uni-jena.de}

\author{G.~Fuhrmann}
\address{Department of Mathematics, Imperial College London, United Kingdom}
\email{gabriel.fuhrmann@imperial.ac.uk}

\author{T.~J\"ager}
\address{Institute of Mathematics, Friedrich Schiller University Jena, Germany}
\email{tobias.jaeger@uni-jena.de}

\subjclass[2010]{}

\begin{abstract} The classical fold bifurcation is a paradigmatic example
of a critical transition.  It has been used in a variety of contexts, including
in particular ecology and climate science, to motivate the role of slow recovery
rates and increased autocorrelations as early-warning signals of such
transitions.

We study the influence of external forcing on fold bifurcations and the
respective early-warning signals. Thereby, our prime examples are single-species
population dynamical models with Allee effect under the influence of either
quasiperiodic forcing or bounded random noise. We show that the presence of
these external factors may lead to so-called {\em non-smooth} fold bifurcations,
and thereby has a significant impact on the behaviour of the Lyapunov exponents
(and hence the recovery rates). In particular, it may lead to the absence of
critical slowing down prior to population collapse. More precisely, unlike in
the unforced case, the question whether slow recovery rates can be observed or detected
prior to the transition crucially depends on the chosen time-scales and the size
of the considered data set.


\end{abstract}

\maketitle

\section{Introduction}\label{Intro}

In recent years, the notions of {\em tipping points} and {\em critical
  transitions} have received widespread attention throughout a broad scope of
sciences. These terms usually refer to abrupt and drastic changes in a system's
behaviour upon a small and slow variation of the system parameters
\cite{van2007slow,Scheffer2009CriticalTransitions,kuehn2011mathematical,scheffer2012anticipating}. In
this context, an important issue of immediate practical interest is that of
early warning signals, that is, indicators which allow to anticipate an oncoming
transition in a systems qualitative behaviour before it actually occurs. A
concept that has led to widely recognised advances in this direction are {\em
  slow recovery rates (critical slowing down)}, which often come along with an
{\em increase in autocorrelation}
\cite{van2007slow,Schefferetal2009EWSforCT,Scheffer2009CriticalTransitions,scheffer2012anticipating,veraart2012recovery}. Both
have been described as possible early warning signals in a variety of contexts,
in theoretical as well as in experimental settings
\cite{dakos2008slowing,carpenter2011early,guttal2013robustness,van2014critical,kefi2014early,rikkert2016slowing}.

The aim of our work here is to understand which of the above-mentioned features
of critical transitions persist when the considered system is subject to the
influence of external forcing, as it happens for many real-world processes, and
also to obtain a better idea on how the above notions may best be captured in
mathematical terms. Here it should be pointed out that when it comes to
providing a sound mathematical framework for critical transitions, it is an
eminent problem that all the above notions comprise a variety of different
phenomena and still lack a precise and comprehensive mathematical definition --
an issue which is well-known in non-linear dynamics and comes up in similar form
for key concepts like `chaos', `fractals' or `strange attractors' in dynamical
systems theory.\smallskip

{\bf A forced Allee model.}\quad We will concentrate on the classical fold
bifurcation, which comprises key features of critical transitions and has
emerged as a paradigmatic example in this context
\cite{Scheffer2009CriticalTransitions}. As it is well-known, in this type of
bifurcation a stable and an unstable equilibrium point of a parameter-dependent
scalar ODE approach each other and eventually merge to form a single neutral
equilibrium point, which then vanishes. Since this leads to the disappearance of
all equilibria in a certain region, it presents the abrupt change in the
system's qualitative behaviour that is characteristic of critical transitions.
In order to fix ideas, we consider as a specific example the single-species
population model with Allee effect given by the scalar ODE
\begin{equation} \label{e.simple_Allee}
  \begin{split}
  x' & = \ rx \cdot \left(1-\frac{x}{K}\right)\cdot
  \left(\frac{x}{K}-\frac{S}{K}\right) - \beta x \\ & = \ \frac{r}{K^2}\cdot x
  \cdot (K-x)\cdot (x-S) -\beta x\ =: \ v_\beta(x)
  \end{split}
\end{equation}
Here $r>0$ denotes the intrinsic growth factor of the population, $K>0$ is the
carrying capacity and $S\in (0,K)$ is the threshold value below which
the population dies out due to an Allee effect. The term $\beta x$ represents an
external stress factor that puts additional pressure on the population. An
increase of the parameter $\beta$ leads to a fold bifurcation and the subsequent
collapse of the population at some critical value $\beta_c>0$.\foot{In fact, a
  straightforward calculation yields $\beta_c=\frac{(K-S)^2\cdot r}{4K^2}$.}  The
  bifurcation pattern is drawn in Figure~\ref{f.fold_bifurcation}(a), whereas
  Figure~\ref{f.fold_bifurcation}(b) shows the behaviour of the Lyapunov
  exponents of the attracting and repelling equilibria during the bifurcation.

\begin{figure}[h!]
  \includegraphics[scale=0.45]{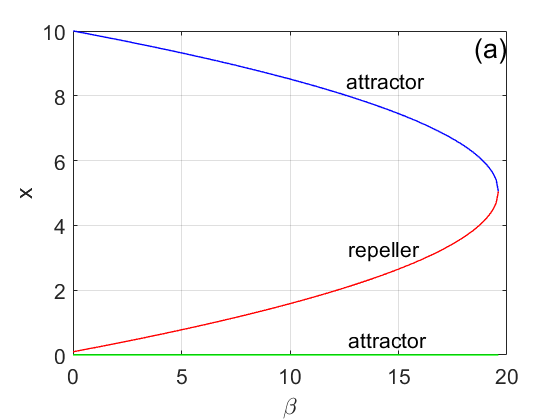}
  \includegraphics[scale=0.45]{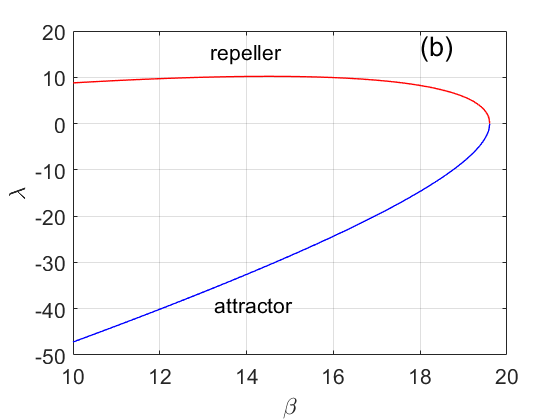}
  \caption{(a) Bifurcation diagram of a fold bifurcation in the Allee model
    (\ref{e.simple_Allee}) with parameters $r=80,\ K=10, S=0.1$ and a
    bifurcation at a critical parameter $\beta_c\simeq 19.978$. The stable
    non-zero equilibrium is shown in blue, the unstable equilibrium in red. The
    stable equilibrium at $x=0$ is shown in green. (b) Behaviour of the Lyapunov
    exponents of the upper stable and the unstable equilibrium during the fold
    bifurcation.
} \label{f.fold_bifurcation}
\end{figure}
In this setting, one obvious possible mathematical interpretation of {\em
  recovery rates} is to identify them with the Lyapunov exponents of the stable
or neutral equilibria, so that {\em slow recovery rates} and {\em critical
  slowing down} correspond exactly to the fact that the two lines in
Figure~\ref{f.fold_bifurcation}(b) meet at zero when the bifurcation parameter
is reached. \medskip

As mentioned before, our main goal is to investigate the same bifurcation
pattern in forced versions of the above model (\ref{e.simple_Allee}), which are
given by non-autonomous scalar ODEs of the form

\begin{equation} \label{e.forced_Allee}
 \boxed{ \ x'(t) \ = \ \frac{r}{K^2} \cdot x(t) \cdot (K-x(t))\cdot (x(t)-S) -
   (\beta+\kappa\cdot F(t)) \cdot x \ =: \ V_{\kappa,\beta}(t,x) \ }
\end{equation}\medskip

Here, time-dependence (or external forcing) of the system (\ref{e.forced_Allee})
is introduced via a {\em forcing term} $\kappa\cdot F(t)$, with {\em coupling
  constant } $\kappa>0$ and a forcing function $F:\R\to[0,1]$. We consider two
different types of forcing processes. \smallskip

\begin{description}
  \item[Quasiperiodic forcing] This kind of forcing corresponds to the influence of several periodic
    external factors with incommensurate frequencies $\nu_1\ld \nu_d\in\R$. As a
    specific example, we choose the forcing function as
       \begin{equation}
      \label{e.quasiperiodic_forcing_term_example}
       F(t) \ = \ \prod_{i=1}^d \left(\frac{1+ \sin(2\pi
         (\theta_i+t\cdot\nu_i))}{2}\right)^q \ ,
    \end{equation}
    with arbitrary initial conditions $\theta_1\ld \theta_d\in\R$. Let us note
    here that due to the periodicity of the sine function, these initial values
    may also be viewed as elements of the circle $\T^1=\R/(2\pi\Z)$ so that
    $\theta=(\theta_1\ld \theta_d)$ becomes an element of the $d$-torus
    $\T^d=\R^d/(2\pi\Z)^d$.  The parameter $q\in\N$ allows some additional
    control the geometry of the forcing function.
    \item[Bounded random noise] Secondly, we consider the effect of external
      random perturbations on the system, given by the forcing function
       \begin{equation} \label{e.random_forcing_term_example}
    F(t) \ = \ \frac{1+\sin(\theta_0+W_t)}{2} \ ,
       \end{equation}
       where $\theta_0\in\T^1$ is again an arbitrary initial condition and $W_t$
       denotes a one-dimensional Brownian motion (but higher-dimensional
       analogues could be considered as well).
\end{description}\smallskip

We thus arrive at the forced scalar differential equation
(\ref{e.forced_Allee}), with the forcing function $F$ given either by
(\ref{e.quasiperiodic_forcing_term_example}) or
(\ref{e.random_forcing_term_example}), as basic models to which we mainly refer
to in this introduction. Most of the statements we actually prove in the later
sections hold in greater generality, both with respect to the model
(\ref{e.forced_Allee}) and to the employed forcing processes
(\ref{e.quasiperiodic_forcing_term_example}) and
(\ref{e.random_forcing_term_example}). In some other cases, we will need to
replace the Allee model by discrete time systems with qualitatively similar
behaviour in order to obtain rigorous results. These systems should then be
considered as simplified models for the time-one-maps of the flow induced by
(\ref{e.forced_Allee}). However, we refer to the respective later sections for
details in order to avoid too many technicalities at this point.
\smallskip

{\bf Lyapunov gap for fold bifurcations in forced systems.}\quad In order to
discuss what happens with the fold bifurcation pattern and the corresponding
early-warning signals in the forced Allee model (\ref{e.forced_Allee}), we will
first concentrate on the behaviour of the Lyapunov
exponents. Figure~\ref{f.LE_smooth/nonsmooth} shows the behaviour of the
Lyapunov exponents of the attractor and the repeller of (\ref{e.forced_Allee})
throughout the bifurcation, with two different choices of the parameters
$\kappa$ and $q$ in the case of quasiperiodic forcing in (a) and (b) and
different parameters $\kappa$ in the random case in (c) and (d).

\begin{figure}[h!]
\includegraphics[scale=0.45]{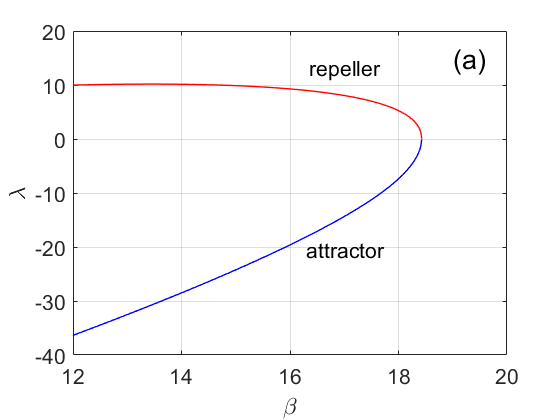} \quad
\includegraphics[scale=0.45]{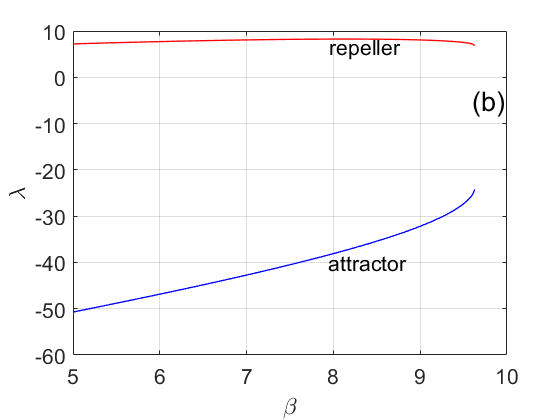} 

\includegraphics[scale=0.45]{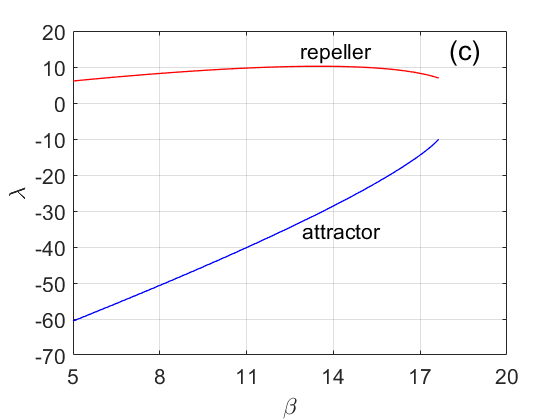}\quad
\includegraphics[scale=0.45]{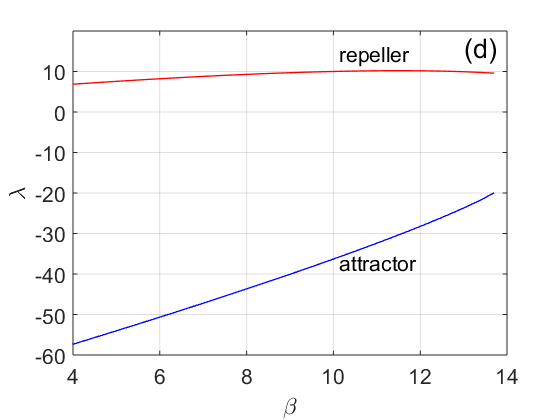}
  \caption{{\bf (a) and (b):} \ Lyapunov exponents during fold bifurcations in
    the qpf Allee model; (a) smooth bifurcation with $r=80,\ K=10, S=0.1,\ q=1,
    \nu=(2\omega,2\pi)$ (where $\omega$ is the golden mean) and $\kappa=4$
    (bifurcation at $\beta_c\simeq 18.4269$); (b) non-smooth bifurcation with
    $r=80, K=10,\ S=0.1,\ q=5, \nu=(2\omega,2\pi)$ and $\kappa=51.2$. The
    bifurcation occurs at $\beta_c\simeq 9.628$.\newline \mbox{ } {\bf (c) and
      (d):} \ Lyapunov exponents during the fold bifurcation in the randomly
    forced Allee model; (c) with $r=80, K=10, S=0.1, \kappa=2$ and bifurcation
    parameter $\beta_c=17.978$; (d) with $r=80, K=10, S=0.1, \kappa=6$ and
    bifurcation parameter $\beta_c=13.978$;}\label{f.LE_smooth/nonsmooth}
\end{figure}

While the behaviour in (a) is in perfect analogy with the unforced case in
Figure~\ref{f.fold_bifurcation}(b), the situation in (b)--(d) is clearly
different. Although the Lyapunov exponents of the attractor and the repeller do
approach each other, there remains a clear gap at the bifurcation point, and in
particular the Lyapunov exponents of the attractor (which are the `visible' or
`physically relevant' ones) stay strictly away from zero. Given the significance
of zero exponents for the observation of critical slowing down and slow recovery
rates, this is certainly noteworthy and deserves a closer examination. Moreover,
while in Figure~\ref{f.LE_smooth/nonsmooth}(b)--(d) the Lyapunov exponents do at
least move towards each other as the bifurcation is approached, this actually
turns out to depend just on the precise form of the parameter family. In the
above cases, we have only varied the bifurcation parameter $\beta$, while
leaving all other constants in (\ref{e.forced_Allee}) invariant. In contrast to
this, it seems likely that in real-world situations other system parameters such
as the intrinsic growth rate $r$ in (\ref{e.forced_Allee}) or the noise
amplitude in (\ref{e.random_forcing_term_example}), vary as well as the pressure
on the population increases. The result of such couplings is shown in
Figure~\ref{f.LE_nonsmooth_alternative_families}. It can be seen that in this
case the Lyapunov exponents may move away from each other all throughout the
bifurcation process, and hence there is no chance at all to anticipate the
oncoming transition based only on their behaviour.

\begin{remark}
We should note that the phenomena that we describe here are known by folklore in
the field of stochastic processes and stochastic differential equations, where
the presence of noise equally prevents the recovery rates from going down all
the way to zero before a transition happens. However, in this context it is much
harder to pin this observation down mathematically, since the presence of
unbounded noise immediately `destroys' the fold bifurcation and leads to the
existence of a unique stationary measure in stochastic versions of
(\ref{e.forced_Allee}) and similar models. Moreover, the forcing with bounded
noise is arguably quite relevant and intrisically motivated from the biological
viewpoint.
\end{remark}

\begin{figure}[h!]
\includegraphics[scale=0.435]{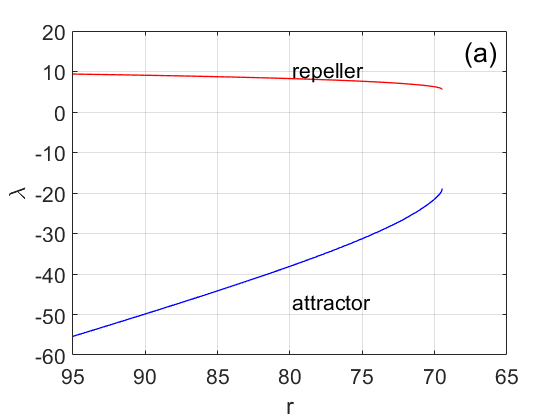} \quad \includegraphics[scale=0.435]{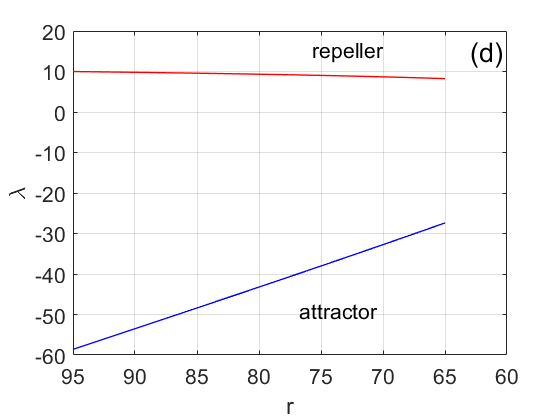} 
\includegraphics[scale=0.435]{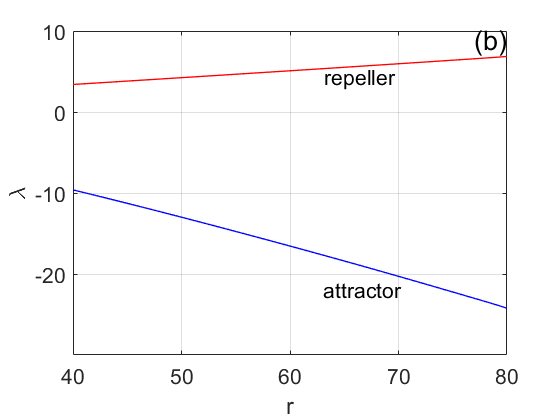}  \quad \includegraphics[scale=0.435]{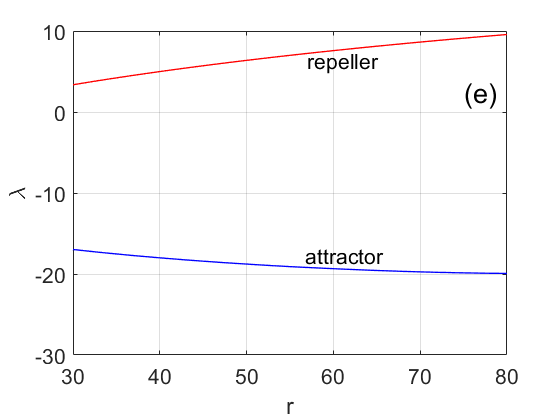} 
 \includegraphics[scale=0.35]{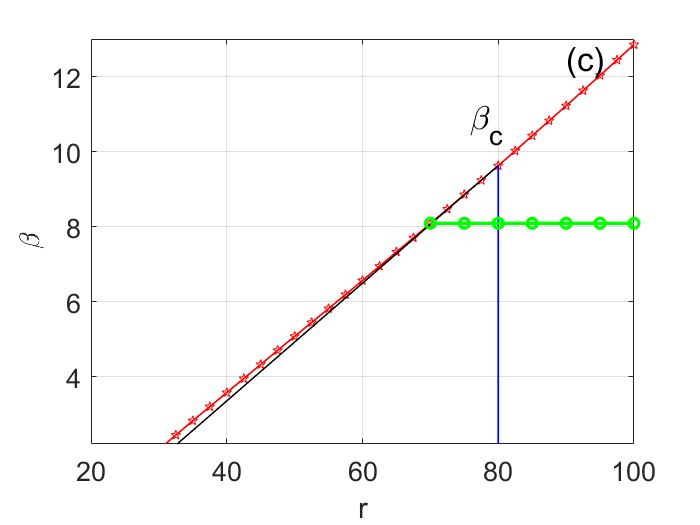} \quad\includegraphics[scale=0.35]{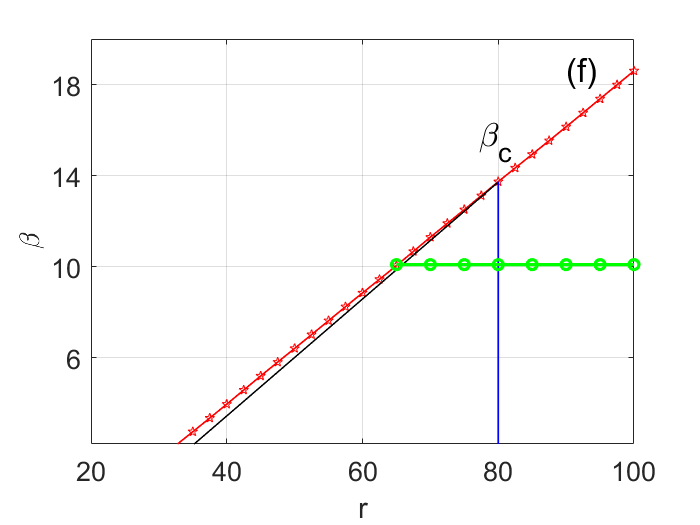} 
  \caption{Lyapunov exponents during fold bifurcations in the forced Allee model
    with different variations of parameters. (a) shows the behaviour in the qpf case
    as the parameter $r$ is decreased, corresponding to the horizontal green
    line in the two-parameter space depicted in (c). In contrast, (b) shows the behaviour when $\beta$ and
    $r$ are varied simultaneously along the black curve in (c). In this case,
    the Lyapunov exponents move apart throughout the entire bifurcation process. In
    (d) and (e), similar plots are shown for the randomly forced case. In (d),
    again only the parameter $r$ is varied and decreases along the horizontal
    green line in (f). In (e), both parameters $\beta$ and $r$ are again
    varied at the same time, along the black curve in (f). In both (c) and (f),
    the red line is an interpolation of the numerically determined critical parameters for 
    different values of $\beta$ and $r$.}\label{f.LE_nonsmooth_alternative_families}
\end{figure}
\medskip

{\bf Mathematical analysis: skew product flows and non-smooth
  bifurcations.}\quad In order to understand and explain these phenomena, it is
indispensable to have a look at the mathematical framework that is used to
describe fold bifurcations in forced systems. To that end, we first concentrate
on the case of quasiperiodic forcing. The rigorous analysis of non-autonomous
ODE's, such as the one given by (\ref{e.forced_Allee}) and
(\ref{e.quasiperiodic_forcing_term_example}), hinges on the fact that the family
of equations (\ref{e.forced_Allee}), with all possible initial conditions
$\theta=(\theta_1\ld \theta_d)\in\T^d$, defines a skew product flow
\begin{equation}\label{e.skew_product_flow}
\Xi:\R\times \Theta\times\R\to \Theta\times \R\quad ,\quad (t,\theta,x)\mapsto
\Xi^t(\theta,x)=(\omega_t(\theta),\xi^t(\theta,x)) \ .
\end{equation}
Here the {\em driving space} $\Theta$ is the $d$-torus, $\Theta=\T^d=\R^d/\Z^d$
(corresponding to the set of possible initial conditions). The {\em driving
  flow} $\omega :\R\times\Theta\to\Theta$ is given by the irrational Kronecker
flow $\omega_t(\theta)=\theta+t\cdot v$ with translation vector\foot{Composed of
  the $d$ incommensurate frequencies $\nu_i$ in
  (\ref{e.quasiperiodic_forcing_term_example}).} $v=(\nu_1\ld \nu_d)$, and
models the quasiperiodic dynamics of the external driving factors. The flow
$\Xi$ is uniquely determined by the fact that the mapping $t\mapsto
\xi^t(\theta,x)$ is the solution to (\ref{e.forced_Allee}) with forcing function
(\ref{e.quasiperiodic_forcing_term_example}). A similar flow representation can
be given in the case of random forcing. We will describe this passage from
non-autonomous equations to skew product flows in more detail in
Section~\ref{1}, but also refer the mathematically interested reader to standard
references such as \cite{Arnold1998RandomDynamicalSystems,huang/yi:2007} for
further reading.

The advantage of the skew product setting lies in the fact that the classical
notion of equilibrium points -- which does not make sense anymore for
time-dependent systems such as (\ref{e.forced_Allee}) -- can be replaced by that of
{\em random} or {\em non-autonomous equilibria}. These are defined as measurable
functions $x:\Theta\to\R,\ \theta\mapsto x(\theta)$ that satisfy
$\Xi^t(\theta,x(\theta))=x(\omega_t(\theta))$. Hence, a non-autonomous
equilibrium can be thought of as a curve, surface or higher-dimensional
submanifold of the product space $\Theta\times\R$ that can be represented as a
graph over the base space $\Theta$, is invariant under the skew product flow
$\Xi$ and is composed of solutions of (\ref{e.forced_Allee}) with varying
initial conditions. With this new notion of an equilibrium, fold bifurcations
in forced systems can be described, in perfect analogy to the classical case, as
the collision and subsequent extinction of a stable and an unstable equilibrium
\cite{nunez/obaya:2007,AnagnostopoulouJaeger2012SaddleNodes}. This process is
shown in Figure~\ref{f.forced_folds}(a)--(d) where two such equilibrium
manifolds approach each other and then merge to form a neutral equilibrium.
  
\begin{figure}[h!]
\includegraphics[scale=0.4]{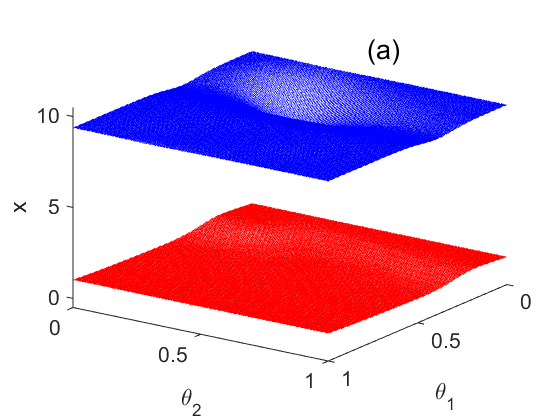} \qquad \includegraphics[scale=0.4]{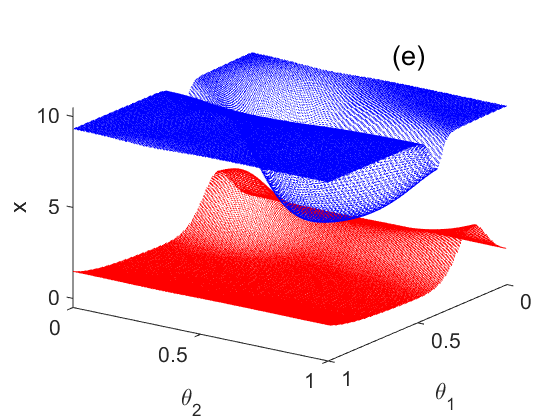}\\
\includegraphics[scale=0.4]{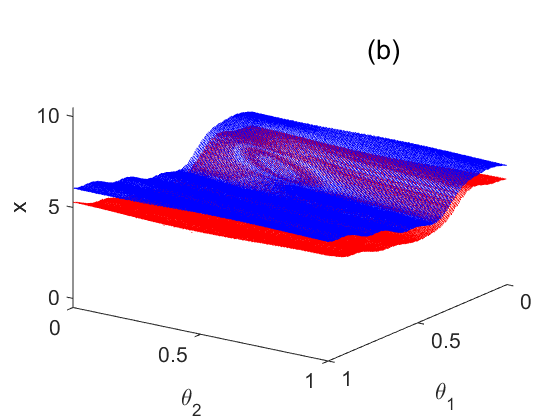} \qquad \includegraphics[scale=0.4]{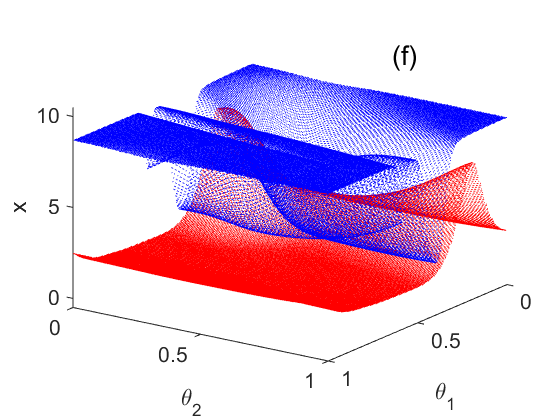}\\
\includegraphics[scale=0.4]{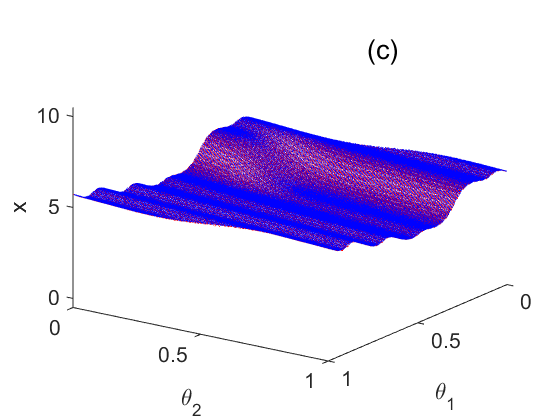} \qquad \includegraphics[scale=0.4]{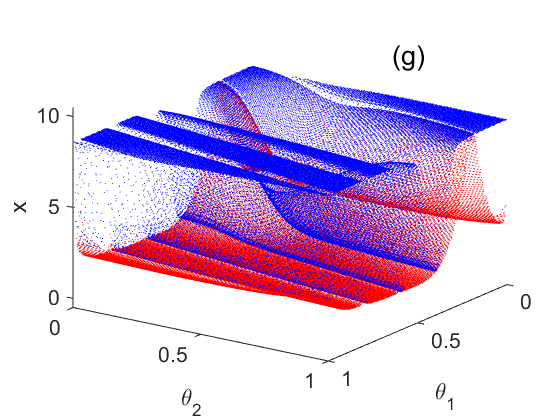}\\
\includegraphics[scale=0.4]{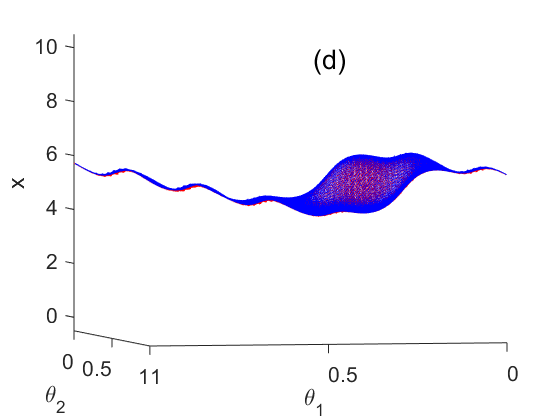} \qquad \includegraphics[scale=0.4]{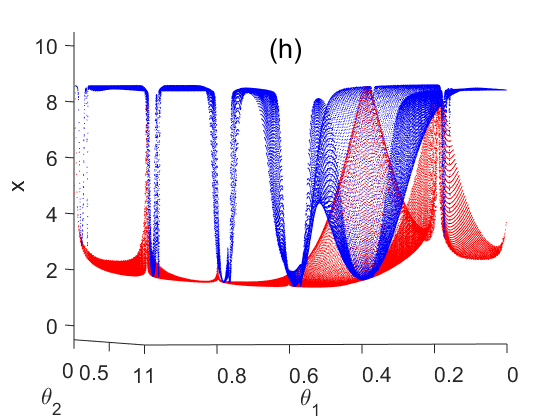}
    \caption{{\bf (a)--(d):} Smooth fold bifurcation in the qpf Allee model with
      parameters $r=35$, $K=10$, $S=0.1$, $q=3$ and $\kappa=41$ and $\beta=2$ in
      (a), $\beta=7.8$ in (b) and $\beta=7.8455$ in (c) and (d). The attractor
      is shown in blue, the repeller in red. The last two figures show the
      merged attractor and repeller at the bifurcation point from two different
      angles. The rotation vector $\rho$ is $\rho=(5\omega,5\pi)$ in all cases,
      where $\omega$ is the irrational part of the golden mean. \\ {\bf
        (e)--(h):} Non-smooth fold bifurcation in the qpf Allee model with
      parameters $r=80$, $K=10$, $S=0.1$, $q=5$ and $\kappa=51.2$ and $\beta=2$
      in (e), $\beta=9$ in (f) and $\beta=9.6282$ in (g) and (h). Again the last
      two figures show the attractor (blue) and the repeller (red) at the
      bifurcation point from two different angles. The rotation vector $\rho$ is
      $\rho=(2\omega,2\pi)$. } \label{f.forced_folds}
   \end{figure}

In contrast to the unforced case, however, there is a second possibility how
such a collision can occur. As the value of the non-autonomous equilibria depend
on the forcing variable $\theta$, the two curves or surfaces can also collide
only for some values of $\theta$, without merging together uniformly. This
pattern is shown in Figure \ref{f.forced_folds}(e)--(h). In this case, one
speaks of a {\em non-smooth fold bifurcation}, in which the neutral equilibrium
at the bifurcation point is replaced by an attractor-repeller pair. Moreover, in
the case of quasiperiodic forcing the stable and unstable non-autonomous
equilibria are called {\em strange non-chaotic attractors (SNA)} and {\em
  strange non-chaotic repellers (SNR)} due to their unusual combination of a
fractal geometry and non-chaotic dynamics
\cite{grebogi/ott/pelikan/yorke:1984,bjerkloev:2005,FeudelKuzetsovPikovsky2006StrangeNonchaoticAttractors,haro/puig:2006,Jaeger2009CreationOfSNA,fuhrmann2013NonsmoothSaddleNodesI,FuhrmannGroegerJaeger2015SNADimensions}.

It is this dichotomy between smooth and non-smooth fold bifurcations which
explains the different behaviour of the Lyapunov exponents observed in
Figure~\ref{f.LE_smooth/nonsmooth}. In order to make this precise, we denote by
$x^s_\beta$ the non-zero stable equilibrium of (\ref{e.forced_Allee}) at
parameter $\beta$,\foot{Note that there is always an equilibrium at zero,
  which is a natural requirement for any population dynamical model and ensured
  by the multiplicative form of the forcing in (\ref{e.forced_Allee}). Due to
  the Allee effect, the zero equilibrium is stable as well and presents the
  unique global attractor of the system after the bifurcation.} and by
$\lambda(x^s_\beta)$ its associated Lyapunov exponent. We refer to
Section~\ref{Preliminaries} for the precise definitions. Moreover, in order to
obtain rigorous results we will need to apply a general framework for
non-autonomous fold and saddle-node bifurcations that has been established in
\cite{AnagnostopoulouJaeger2012SaddleNodes}. An important condition that is
required there is the concavity of the fibre maps in the considered region
which is a consequence of the concavity of the right side of the respective
non-autonomous ODE. In order to ensure this concavity in (\ref{e.forced_Allee}),
we need to restrict to suitable parameter ranges, as specified in the following.
\begin{remark} \label{r.parameter_range}
  We let
  \begin{equation}
    b(r,K,S)=\frac{r}{K^2}\cdot \left(\frac{K-S}{2}\right)^2 \eqand
    \gamma(K,S)=\frac{1}{9}\cdot\left(\frac{K+S}{K-S}\right)^2 \ 
  \end{equation}
and assume that 
\begin{equation}
  \label{e.kappa_condition}
\kappa\ < \ b(r,K,S)\cdot (1-\gamma(K,S)) \ . 
\end{equation}
Then, as we explain in detail in Section~\ref{ApplicationToAllee}, the family
(\ref{e.forced_Allee}) with forcing term
(\ref{e.quasiperiodic_forcing_term_example}) or
(\ref{e.random_forcing_term_example}) undergoes a fold bifurcation in the
parameter interval
\begin{equation}
  \label{e.parameter_interval}
J(r,K,S) \ = \ [b(r,K,S)\cdot (1-\gamma(K,S)),b(r,K,S)+1] \ . 
\end{equation}
It should be mentioned, however, that this restriction in the parameter ranges
is rather a technical condition and could easily be improved, in particular by
using numerical methods, in order to include a broader range of parameters.  The
crucial condition is that the time-$t$-maps of skew product flows induced by
(\ref{e.forced_Allee}) are concave for some $t>0$. Hence, even if not all of our
examples satisfy condition (\ref{e.kappa_condition}), it seems more reasonable
to rely on the numerical evidence for the occurrence of fold bifurcations in
these cases than to highly technical proofs that do not add further
insight. However, all the rigorous statements provided below will be restricted
to this parameter range.
\end{remark}

\begin{intro_theorem}
  \label{t.lyapunov_gap} Suppose that (\ref{e.kappa_condition}) holds. If the Allee model (\ref{e.forced_Allee}) with forcing term given by 
 (\ref{e.quasiperiodic_forcing_term_example}) or
  (\ref{e.random_forcing_term_example}) undergoes a non-smooth fold bifurcation
  at the critical parameter $\beta_c\in J(r,K,S)$, then we have that
  \begin{equation}
    \label{e.lyapunov_gap} \lim_{\beta\nearrow\beta_c} \lambda(x^s_\beta) \ = \ \lambda(x^s_{\beta_c}) \ < \  0 \  . 
  \end{equation}
  If the fold bifurcation is smooth, then we have $\lim_{\beta\nearrow\beta_c}
  \lambda(x^s_\beta) = 0$.  The analogous results hold for the unstable
  equilibrium $x^u_\beta$.
\end{intro_theorem}

{\bf Relevance of non-smooth bifurcations.}\quad An immediate question that can
be asked in the context of the above observations is whether non-smooth fold
bifurcations present a very relevant phenomenon, or if they are rather
`pathological' and may not play an important role for the description of
real-world processes. However, in the case of quasiperiodic forcing, the
wide-spread occurrence of non-smooth bifurcations and the related existence of
SNA has been observed in a large number of numerical and experimental studies
and in a variety of different contexts, ranging from classical and electronic
oscillators to quantum mechanics, conceptual climate models and astrophysics
(e.g.\ \cite{romeiras/etal:1987,Ditto/etal:1989,witt/feudel/pikovsky:1997,Venkatesanetal2000SNAinDuffing,%
  haro/puig:2006,MitsuiCrucifixAihara2015BifurcationsAndSNSinClimateModels,%
  RizwanaRaja2015WienBrideOscillators,Zhang2013WadaSNA,Lindneretal2015StrangeNonchaoticStars}). In
addition, the simulations in Figures~\ref{f.LE_smooth/nonsmooth}(b) and
Figure~\ref{f.forced_folds}(e)--(h) provide similar numerical evidence for the
existence of nonsmooth bifurcations in the qpf Allee model
(\ref{e.forced_Allee}) and (\ref{e.quasiperiodic_forcing_term_example}). These
findings are backed up by rigorous results in
\cite{fuhrmann2013NonsmoothSaddleNodesI,Fuhrmann2016SNAinFlows}, showing that
non-smooth fold bifurcations occur for open sets of parameter families of
quasiperiodically forced scalar ODE's. They can therefore be robust and
persistent under small perturbations of the system.  In the light of these
results, one may say that fold bifurcations in quasiperiodically forced models
may be either smooth or non-smooth, depending on the precise form of the model
and the shape and strength of the forcing, and both of the cases are
sufficiently widespread and persistent to be relevant in practical
considerations and applications.

In the case of bounded random forcing, the situation is different in that this
balance swings completely towards the side of non-smooth bifurcations. Roughly
speaking, any forcing by a sufficiently random external process inevitably leads
to the non-smoothness of the bifurcation. For the case of our model system, this
is established by the following result.
\begin{intro_theorem} \label{t.random_nonsmooth} Suppose that (\ref{e.kappa_condition}) holds. 
  Then any fold bifurcation that occurs in the forced Allee model
  (\ref{e.forced_Allee}) with random forcing term
  (\ref{e.random_forcing_term_example}) at a critical parameter $\beta_c\in
  J(r,K,S)$ is non-smooth.
\end{intro_theorem}
Altogether, the possible non-smoothness of bifurcations is an issue that should
arguably be dealt with if one aims at a comprehensive understanding of critical
transitions.
\medskip

{\bf Critical slowing down and finite-time Lyapunov exponents.}\quad The
interpretation of the Lyapunov gap in a non-smooth fold bifurcation depends on
the precise meaning given to the notion of recovery rates. If these are
identified with the Lyapunov exponents, then it follows that, unlike in
classical fold bifurcations, there are no slow recovery rates in non-smooth fold
bifurcations. However, it seems reasonable to say that the intuitive meaning of
recovery rates, as used in experimental studies like
\cite{Schefferetal2009EWSforCT}, is better captured by the mathematical notion
of {\em finite time Lyapunov exponents}. Instead of measuring the asymptotic
stability of an orbit, these only take into account the expansion or contraction
around an orbit over some finite time span. Given $T>0$, we denote the Lyapunov
exponent at time $T$ of the flow generated by (\ref{e.forced_Allee}) and
starting at an initial condition $(\theta,x)\in\Theta\times\R$ by
$\lambda_T(\theta,x)$.

In a smooth fold bifurcation, it is known that all finite time Lyapunov
exponents in the basin of attraction of the stable equilibrium $x^s_\beta$ will
be very close to $\lambda(x_\beta^s)$, provided the time $T$ is
sufficiently large \cite{sturman/stark:2000}. In contrast to this, the
non-smooth case shows a characteristic spreading of these quantities, which can
be observed in Figure~\ref{f.finite_time_exponents}.

\begin{figure}[h!]
\includegraphics[scale=0.43]{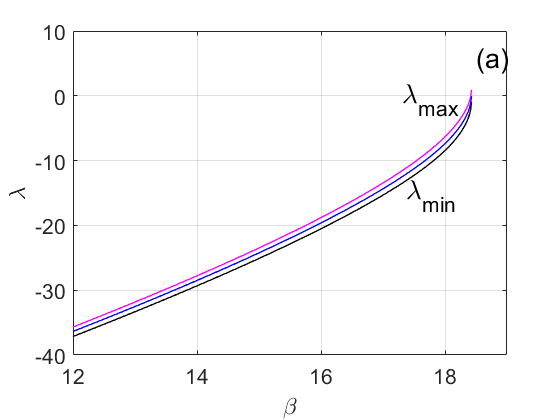} \quad
\includegraphics[scale=0.43]{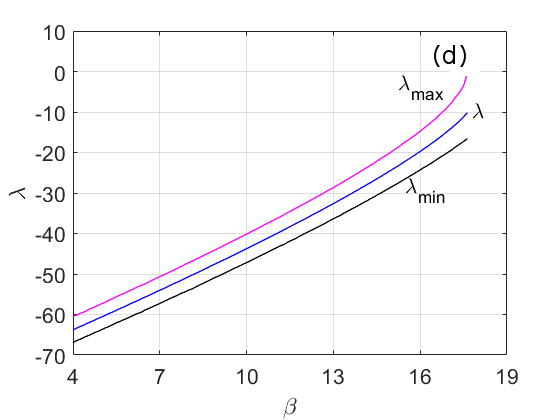} 

\includegraphics[scale=0.43]{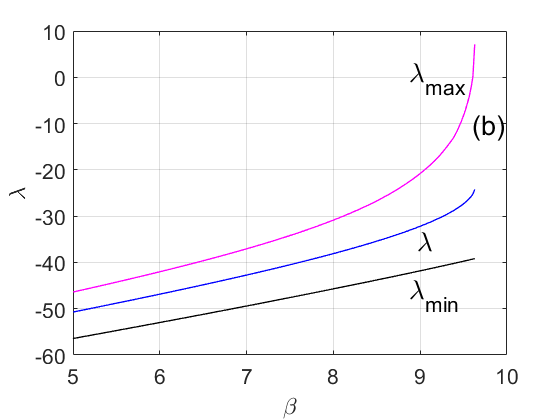} \quad 
\includegraphics[scale=0.43]{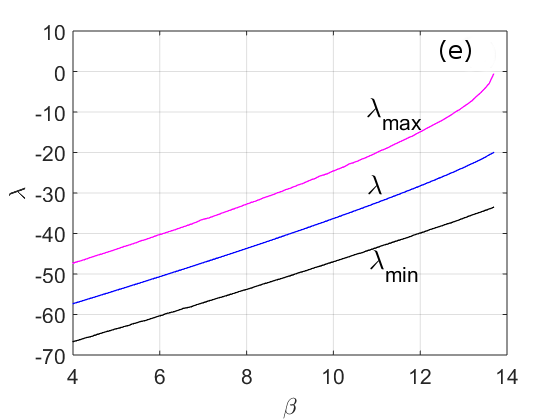}

\includegraphics[scale=0.43]{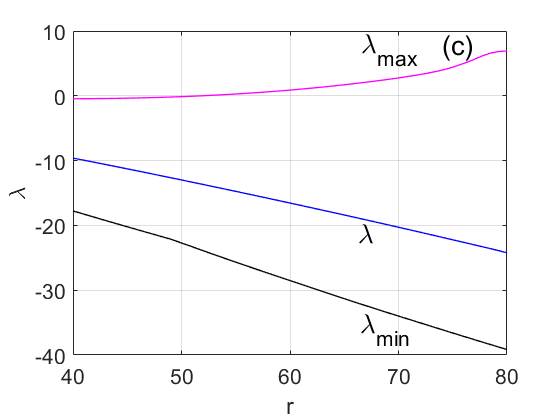} \quad 
\includegraphics[scale=0.43]{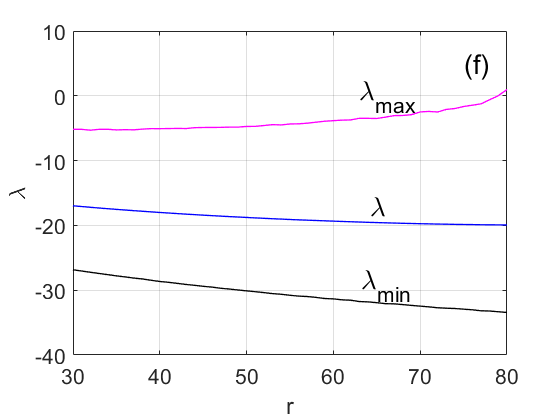}

  \caption{The above plots (a)-(f) show the behaviour of the finite-time
    Lyapunov exponents during fold bifurcations in the forced Allee model. The
    middle curve is always the time 2000 Lyapunov exponent (as an approximation
    of the asymptotic Lyapunov exponent), whereas the upper and the lower curves
    correspond to the maximal and minimal time $4/3$ Lyapunov exponents,
    respectively. (a) shows the case of a smooth fold bifurcation in the qpf
    Allee model with parameter values $r=80,\ K=10,\ S=0.1,\ \kappa= 4$ and
    $q=1$. (b) shows the case of a nonsmooth fold bifurcation in the same model
    with $r=80,\ K=10,\ S=0.1,\ \kappa=51.2$ and $q=5$. (c) shows a
    quasiperiodic case again, but this time with the simultaneous variation of
    parameters as in Figure~\ref{f.LE_nonsmooth_alternative_families}(c). (d)
    and (e) show the case of a non-smooth fold bifurcation in the randomly
    forced Allee model with parameters $r=80,\ K=10,\ S=0.1$ and $\kappa=2$ and
    $\kappa=6$, respectively. Finally, (f) shows the random case imultaneous
    variation of parameters as in
    Figure~\ref{f.LE_nonsmooth_alternative_families}(f). } \label{f.finite_time_exponents}
\end{figure}

In order to translate this observation into a rigorous statement, we denote the
largest time $T$ Lyapunov exponents that is `observable' on the attractor
$x^s_\beta$ by $\lambda^\mathrm{max}_t(x^s_\beta)$, the minimal one by
$\lambda^\mathrm{min}(x^s_\beta)$. (We refer to Section~\ref{Range} for the
precise definition.) The behaviour differs according to whether the forcing is
quasiperiodic or random.

\begin{intro_theorem}\label{t.finite_time_exponent_range} Suppose that (\ref{e.kappa_condition}) holds. If the forced Allee model
  (\ref{e.forced_Allee}) with quasiperiodic forcing term
  (\ref{e.quasiperiodic_forcing_term_example}) undergoes a non-smooth fold
  bifurcation at some critical parameter $\beta\in J(r,K,S)$, then we have that
  \begin{eqnarray}
  \lim_{\beta\nearrow\beta_c} \lambda_T^\mathrm{max}(x^s_\beta) & \geq
  & \lambda(x^u_{\beta_c}) \ > \ 0 \ , \\ \lim_{\beta\nearrow\beta_c}
  \lambda_T^\mathrm{min}(x^u_\beta) & \leq & \lambda(x^s_{\beta_c}) \ <
  \ 0 \ .
  \end{eqnarray}
  In the case of random forcing (\ref{e.random_forcing_term_example}), we have that 
  \begin{equation}
    \lim_{\beta\nearrow\beta_c} \lambda_T^\mathrm{max}(x^s_\beta) \ \geq \ 0 \ .
  \end{equation}
\end{intro_theorem}
Both the statement and the numerical results imply that at least in theory
non-smooth fold bifurcations can be anti\-cipated and detected beforehand via a
spread in the distribution of finite-time Lyapunov exponents, which reaches into
the positive region. However, at the same time this highlights a variety of
practical problems that may arise when trying to establish early-warning signals
for forced systems on the basis of recovery rates. Unlike for Lyapunov
exponents, which are asymptotic quantities and usually show a very uniform
behaviour, the use of finite-time Lyapunov exponents requires to make a number
of choices. First of all, there is the question of the time-scale (the choice of
$T$) for which these quantities should be measured. When $T$ is too small, it is
likely that positive finite-time exponents will be observed already far from any
bifurcation (depending on the geometry of the system). Conversely, if $T$ is
chosen too large, positive finite-time exponents may exist, but may only be
observed with very small probabilities (thus requiring many measurements for a
reliable signal). In any case, even with the right choice of the time-scale and
sufficient data, examples as the one shown in
Figure~\ref{f.finite_time_exponents}(c) will remain difficult to treat.\medskip

{\bf Distribution of finite-time Lyapunov exponents.}\quad Finally, at the critical parameter, we take a
brief look at the distribution of finite-time Lyapunov exponents on different
timescales, which are shown in Figure~\ref{f.distributions} 
(for the qpf case; for results in the random case, see Figure~\ref{f.distributions_random}). 
The probability of observing exponents above or close to zero
is plotted in Figure~\ref{f.distributions_evolution}(a) and decreases quickly
(see Figure~\ref{f.distributions_evolution_random} for the respective plots in the random case). 
However, our simulations are
somewhat inconclusive concerning the rate of decay, which seems to be somewhere
between polynomial and exponential.

\begin{figure}[h!]
   \includegraphics[scale=0.28]{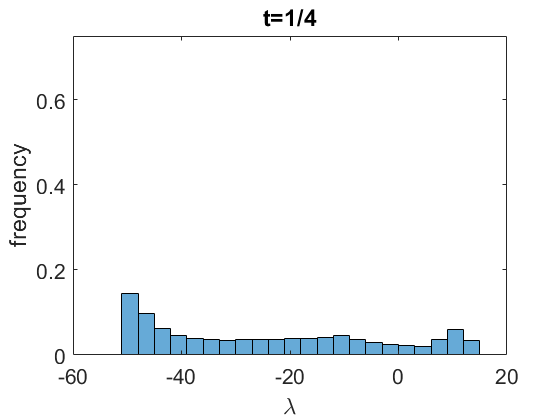} \quad
    \includegraphics[scale=0.28]{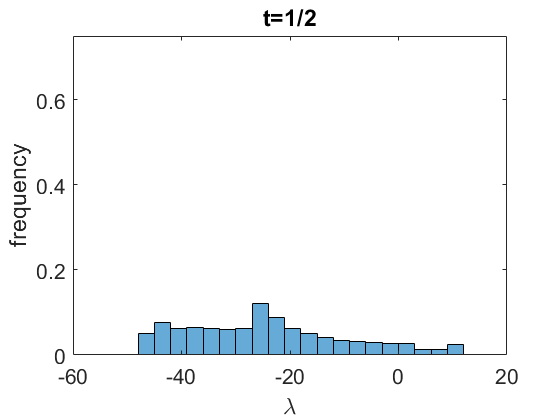} \quad
  \includegraphics[scale=0.28]{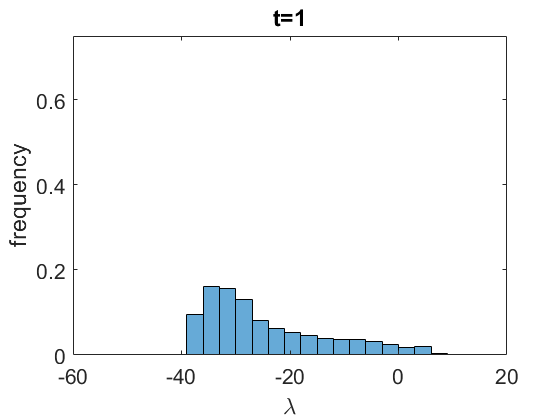}
\medskip

   \includegraphics[scale=0.28]{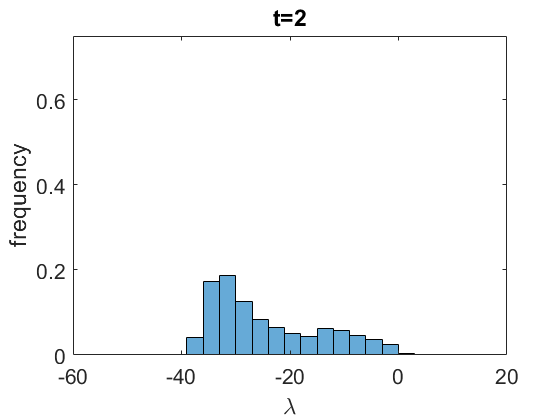} \quad
   \includegraphics[scale=0.28]{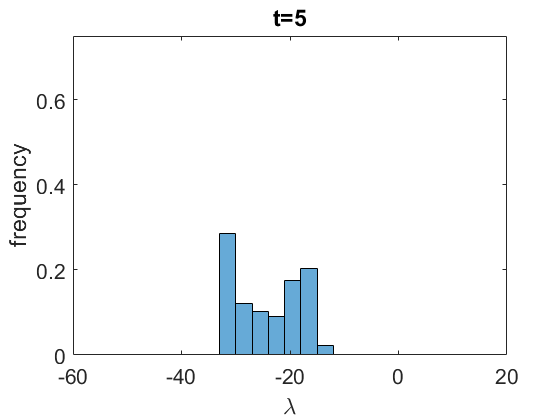} \quad
\includegraphics[scale=0.28]{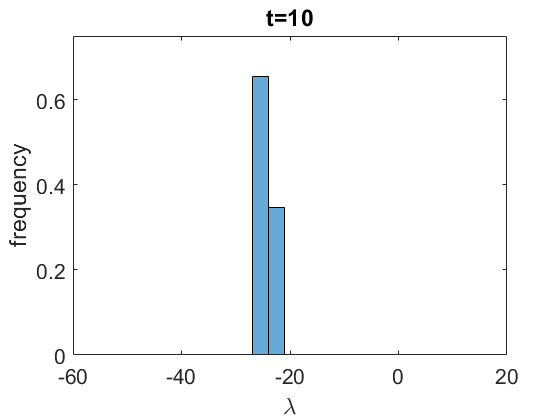}
  \caption{Distributions of the finite-time Lyapunov exponents in the qpf Allee
    model (\ref{e.forced_Allee}) with parameters $r=80$, $K=10$, $S=0.1$,
    $\kappa=51.2$ and $\beta=9.629$ on different timescales, computed with sliding
    windows over a trajectory of length $t=20000$. } \label{f.distributions}
\end{figure}

\begin{figure}[h!]
   \includegraphics[scale=0.23]{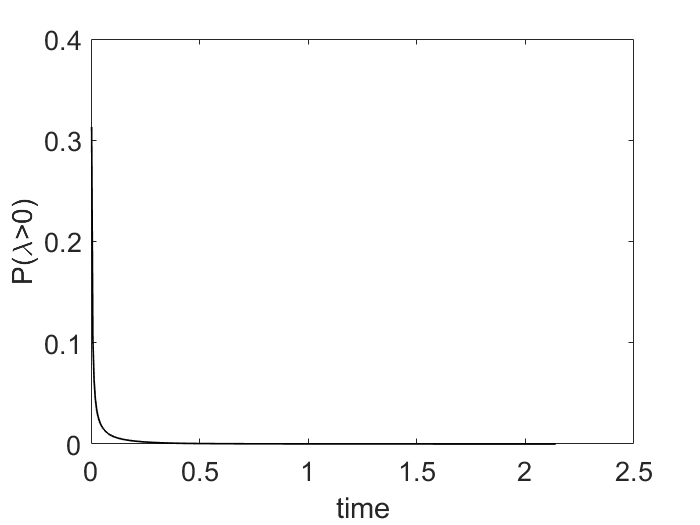} \quad
  \includegraphics[scale=0.23]{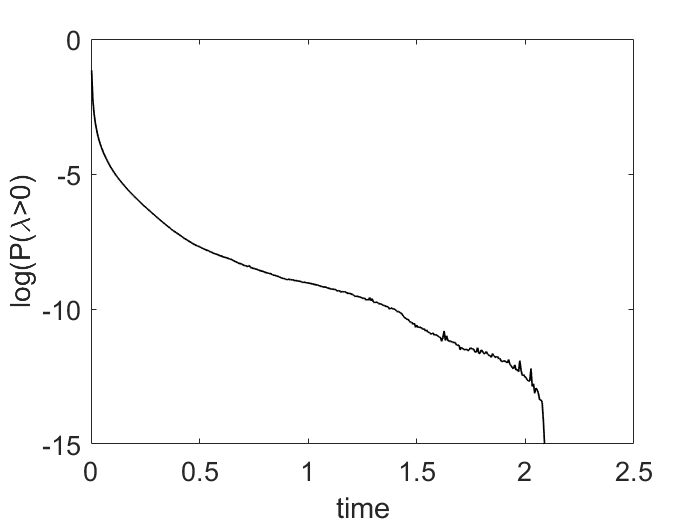} \quad
\includegraphics[scale=0.23]{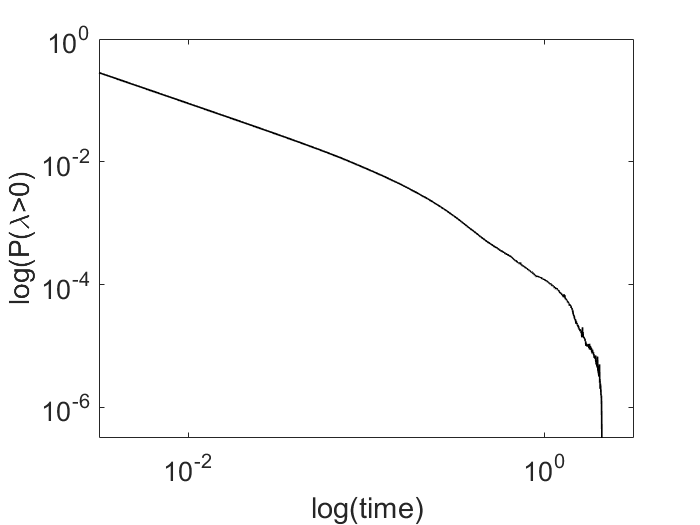} 
     \caption{A plot of the relative frequency of positive exponents (as
       observed in Figure~\ref{f.distributions}) on (a) standard, (b)
       logarithmic and (c) $\log$-$\log$-scale.
     } \label{f.distributions_evolution}
\end{figure}
\medskip

\begin{figure}[h!]
   \includegraphics[scale=0.28]{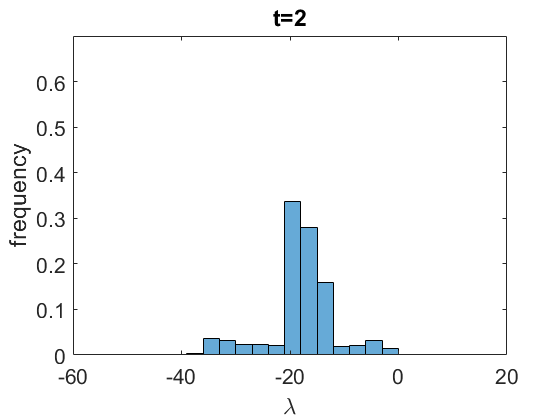} \quad
    \includegraphics[scale=0.28]{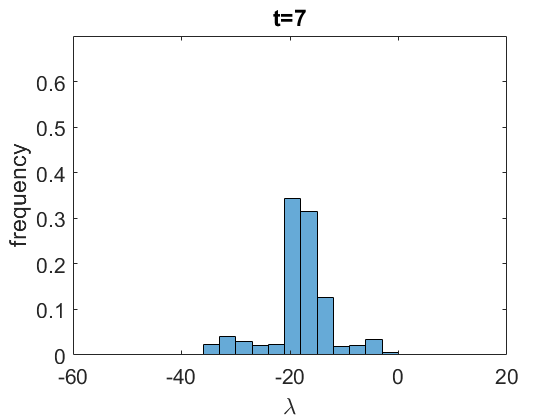} \quad
  \includegraphics[scale=0.28]{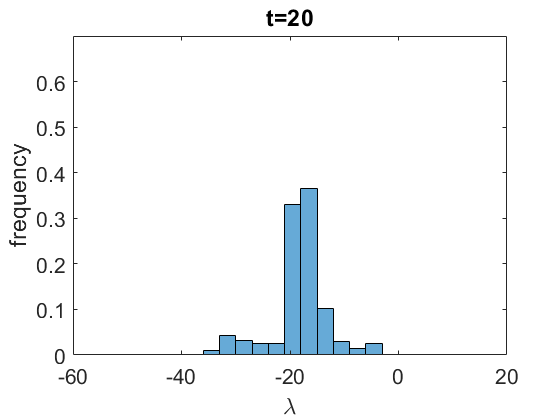}
\medskip
   
   \includegraphics[scale=0.28]{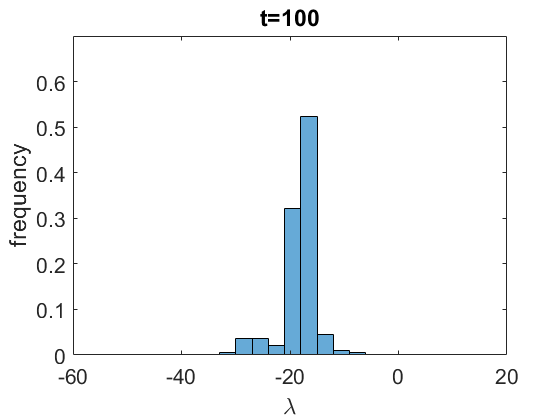} \quad
   \includegraphics[scale=0.28]{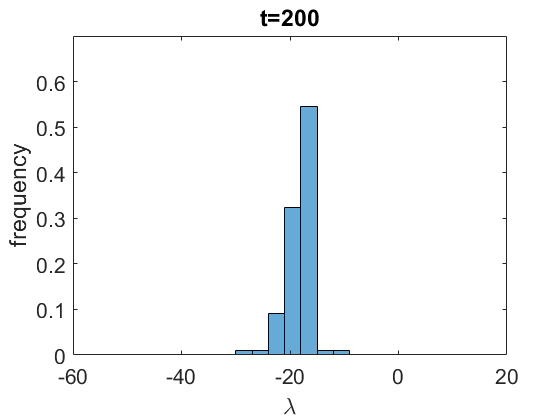} \quad
\includegraphics[scale=0.28]{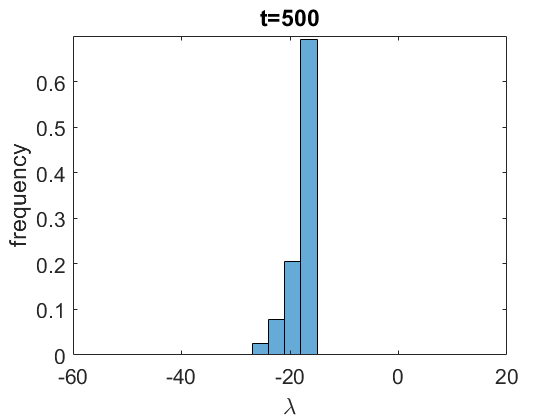}
  \caption{Distributions of the finite-time Lyapunov exponents in the randomly
    forced Allee model (\ref{e.forced_Allee}) with parameters $r=80$, $K=10$,
    $S=0.1$, $\kappa=6$ and $\beta=13.978$ on different timescales, computed
    with sliding windows over a trajectory of length
    $t=20000$. } \label{f.distributions_random}
\end{figure}\medskip

\begin{figure}[h!]
   \includegraphics[scale=0.28]{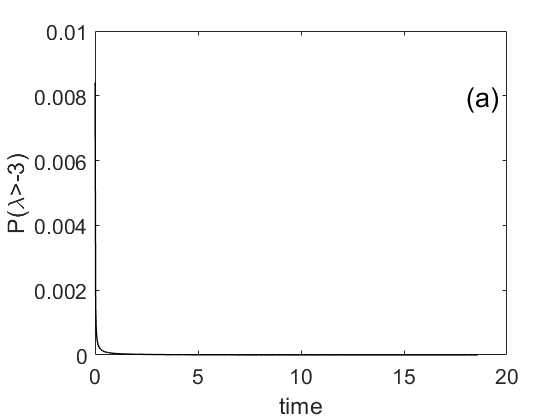} \quad
  \includegraphics[scale=0.28]{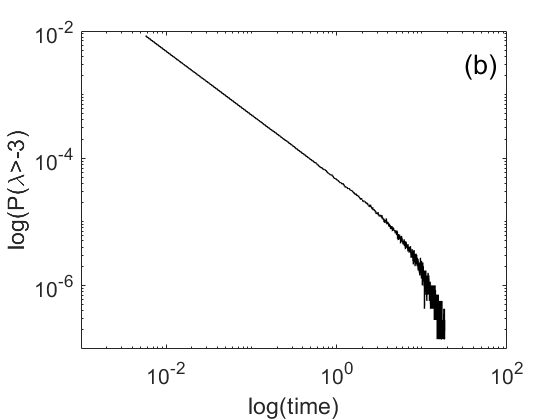} \quad
\includegraphics[scale=0.28]{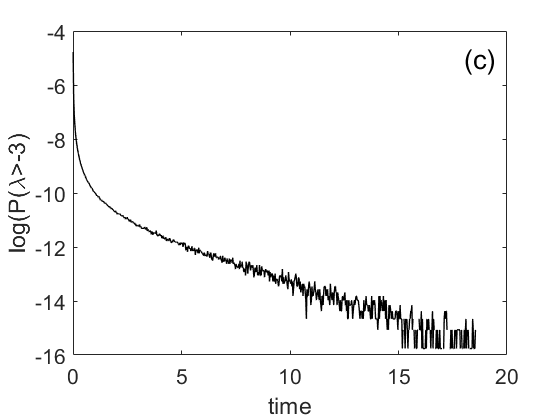} 
     \caption{A plot of the relative frequency of positive exponents (as
       observed in Figure~\ref{f.distributions_random}) on (a) standard, (b)
       logarithmic and (c) $\log$-$\log$-scale.
     } \label{f.distributions_evolution_random}
\end{figure}
\medskip

{\bf Concluding remarks.}\quad In the above context, it should also be pointed
out that although finite-time Lyapunov exponents are -- by definition --
observable in finite time and may therefore in principle be accessible to
experimental measurements, it is a difficult task to achieve and implement this
for any real-life system. At the same time, a meaningful and practical
definition of autocorrelation is difficult to provide for forced system with
moving random equilibria. Hence, the practical implementation of early-warning
signals for critical transitions in forced systems remains a wide open problem,
even in the simplest case of fold bifurcations.

On the theoretical side, an imminent problem that we tried to highlight by the
above discussion is to give a precise mathematical meaning to terms like
recovery rates, critical slowing down as early warning signals and other notions
that come up in the context of critical transitions. If theory and applications
are supposed to go hand in hand, this will be an indispensable basis for further
progress. The results and findings presented here should be understood as a
contribution to that discussion.  \medskip

{\bf Structure of the article.}\ In Section~\ref{Preliminaries}, we collect the
required preliminary facts concerning the mathematical theory of non-autonomous
dynamics and skew product systems, with a particular emphasis on invariant
graphs and fold bifurcations in this setting. The application to the forced
Allee model (\ref{e.forced_Allee}) is discussed in
Section~\ref{ApplicationToAllee}. In Section~\ref{DiscreteTimeModel}, we also
introduce some discrete-time skew product systems, which may be thought of as
simplified models for the time-one maps of the skew product flow induced by the
forced Allee model. Section~\ref{Relevance} is then devoted to the discussion of
non-smooth fold bifurcation and also contains the proof of (a more general
version of) Theorem~\ref{t.random_nonsmooth}. The existence of the Lyapunov gap,
stated in Theorem~\ref{t.lyapunov_gap}, is proven in
Section~\ref{LyapunovExponents}, which also contains a result on the slope of
the Lyapunov exponents at the bifurcation point (in the setting of the
discrete-time model from Section~\ref{DiscreteTimeModel}). Finite-time Lyapunov
exponents are then defined and discussed in Section~\ref{Range}, including the
proof of (a more general version of) Theorem~\ref{t.finite_time_exponent_range}.
\medskip

{\bf Acknowledgments.}\quad This project has received funding from the European
Union's Horizon 2020 research and innovation program under the Marie
Sk\l odowska-Curie grant agreements No 643073 and No 750865. TJ acknowledges
support by a Heisenberg grant of the German Research Council (DFG grant OE
538/6-1).

\section{Preliminaries} \label{Preliminaries}

\subsection{Skew product flows and invariant graphs}

In order to treat continuous-time and discrete-time dynamics alongside, we let
$\T$ be either equal to $\R$ (continuous-time) or $\Z$ (discrete-time).  In both
cases, a dynamical system is a pair $(Y,\Xi)$ of a set $Y$ and a $\T$-flow $\Xi$
on $Y$, that is, a mapping
\begin{equation}
  \Xi:\T\times Y \to Y \quad , \quad (t,y)\mapsto \Xi^t(y)
\end{equation}
which satisfies the flow properties
\begin{equation}
  \Xi^0(y) \ = \ y \eqand \Xi^{t+s}(y) \ = \ \Xi^t(\Xi^s(y)) \ . 
\end{equation}
In the discrete-time case, this implies that $\Xi^t(y)=f^t(y)$, where $f:Y\to Y$
is the bijective map given by $f(y)=\Xi^1(y)$. 

We always assume that $Y$ is equipped with a $\sigma$-algebra $\cB$.
A probability measure $\mu$ on $Y$ is called $\Xi$-invariant if $\mu\circ \Xi^t=\mu$ for all $t\in \T$. 
The set of all $\mu$-invariant probability measures on $(Y,\cB)$ is denoted by
$\cM(\Xi)$.
Given $\mu\in\cM(\Xi)$, we call the quadruple $(Y,\cB,\mu,\Xi)$ a
{\em measure-preserving dynamical system (mpds)}. 
We refer to \cite{Arnold1998RandomDynamicalSystems} and references therein for details and
background.

If $Y$ is a metric space and
$\Xi$ is continuous on the product space $\T\times Y$, we call the pair
$(Y,\Xi)$ a {\em topological dynamical system (tds)}. 
In this case, we throughout assume $\cB$ to be given by the Borel $\sigma$-algebra on $Y$.\medskip

Non-autonomous dynamics are modeled by skew product systems. Given a tds
$(\Theta,\omega)$ or an mpds $(\Theta,\cB,\mu,\omega)$, a {\em skew product flow
  with base $\Theta$ and phase space $X$} is a flow on $Y=\Theta\times X$ of the
form
\begin{equation}
  \Xi:\T\times \Theta\times X\to \Theta\times X \quad , \quad
  (t,\theta,x)\mapsto \Xi^t(\theta,x)=(\omega^t(\theta),\xi^t(\theta,x)) \ .
\end{equation}
Hence, if $\pi_\Theta:\Theta\times X$ is the canonical projection to $\Theta$,
then $\pi_\Theta\circ \Xi^t(\theta,x)=\omega^t(\theta)$. 
The maps $X\ni x\mapsto \xi^t(\theta,x)\in X$, with $\theta\in \Theta$ fixed, 
are called {\em fibre maps}. 
If $X$ is a metric space, we assume the fibres maps to be continuous without further mentioning.
If $X$ is a smooth manifold and all the fibre maps
$\xi^t(\theta,\cdot)$ are $r$~times differentiable, we call $\Xi$ an $\omega$-forced $\cC^r$
flow. If $X=\R$ and the fibre maps are all monotonically
increasing, we say $\Xi$ is an {\em $\omega$-forced monotone flow}.

As mentioned above, in this context the notion of an equilibrium point has to be
replaced by that of a `moving equilibrium', to which we refer as an {\em
  invariant graph}, whose position depends on the forcing variable $\theta$. We
say a measurable function $\varphi:\Theta\rightarrow X$ is an {\em invariant
  graph} of the flow $\Xi$, if it satisfies the condition
\begin{equation}\label{e.invariant_graphs}
  \Xi^t_{\theta}(\varphi(\theta)) \ = \ \varphi(\omega^t(\theta))
\end{equation}
for all $\theta\in\Theta$ and $t\in\T$.\foot{We use $\varphi$ instead of $x$ to
  denote invariant graphs from now on (unlike in the introduction) to stress the
  fact that these are functions.} Here, we usually do not distinguish between
invariant graphs that coincide almost everywhere with respect to the given
reference measure in the base, which in the qpf case is just the Lebesgue
measure. Hence, whenever we speak of invariant graphs, we implicitly mean
equivalence classes of functions. This is very natural when $\Xi$ is forced by
an mpds, but may become a subtle issue as soon as $\Theta$ is a metric space and
topological properties of invariant graphs come into play. For instance, by
saying that an invariant graph is continuous, we mean that there exists a
(uniquely determined) continuous representative of the respective equivalence
class.  It is worth mentioning that in the case of semi-continuous graphs, there
may exist several different semi-continuous representatives in the same
equivalence class -- an issue that we will come back to in
Section~\ref{LyapunovExponents} below.  In the random case, we may not require
that (\ref{e.invariant_graphs}) is satisfied pointwise, but only almost
surely. More precisely, if $\mu$ is an $\omega$-invariant measure and
(\ref{e.invariant_graphs}) is satisfied $\mu$-almost surely, then $\varphi$ is
called a $(\Xi,\mu)$-invariant graph.

It turns out that there is an intimate relation between invariant graphs and the
invariant ergodic measures of the system.
Any
$(\Xi,\mu)$-invariant graph $\varphi$ clearly defines a $\Xi$-invariant measure
$\mu_\varphi$ given by
\begin{equation}\label{e.graph_measure}
  \mu_\varphi(A) \ = \ \mu(\{\theta\in\Theta\mid (\theta,\varphi(\theta)) \in
  A\}).
\end{equation}
A partial converse to this statement for forced monotone flows is provided by
the following result, which essentially goes back to Furstenberg
\cite{furstenberg:1961} (see also \cite{Arnold1998RandomDynamicalSystems}) and
highlights the significance of invariant graphs from an ergodic-theoretical
viewpoint. Given $\mu\in\cM(\omega)$, we denote by $\cM_\mu(\Xi)$ the set of
$\Xi$-invariant probability measures on $\Theta\times X$ which project to
$\mu$ in the first coordinate.
\begin{theorem}[see {\cite[Theorem 1.8.4]{Arnold1998RandomDynamicalSystems}} and {\cite[Theorem 4.1]{furstenberg:1961}}]
\label{1}
Suppose $\Xi$ is an $\omega$-forced monotone flow, $\mu\in\mathcal{M}(\omega)$
and $\nu\in\mathcal{M}_{\mu}(\Xi)$. Then there exists a $(\Xi,\mu)$-invariant
graph $\varphi$ such that $\nu=\mu_\varphi$.
\end{theorem}
Hence, for monotone skew product flows there is a one-to-one correspondence
between invariant ergodic measures and the invariant graphs of the system.
\smallskip

Similar to the autonomous case, the stability of an invariant graph $\varphi$
can be characterised in terms of its {\em Lyapunov exponent}.  The latter is
given by
\begin{equation}\label{lyapunov}
 \lambda_{\mu}(\varphi)\ = \ \lim_{t\rightarrow\infty}\frac{1}{t}\int_{\Theta}\log\|\partial_x
 \xi^t(\theta,\varphi(\theta))\| \ d\mu(\theta) \ , 
\end{equation} where $\partial_x$ denotes the derivative with respect to $x$ and
$\mu\in\cM(\omega)$ is a given reference measure in the base.  Note that the
limit in \eqref{lyapunov} always exists due to Kingman's Ergodic theorem
\cite{Arnold1998RandomDynamicalSystems} or, in the case of forced
one-dimensional flows ($X=\R$), by Birkhoff's Ergodic Theorem. It is known that,
under some mild assumptions, an invariant graph with negative Lyapunov exponent
attracts a set of initial conditions of positive measure
\cite{Jaeger2003NegativeSchwarzian} (with respect to the product measure
$\mu\times\Leb$ if $X=\R^d$, where $\Leb$ denotes the Lebesgue measure on
$\R^d$). Hence, the graph $\varphi$ is called an \emph{attractor} in this case,
and a \emph{repeller} if $\lambda(\varphi)>0$
\cite{fuhrmann2013NonsmoothSaddleNodesI}.

An important notion in the context of forced systems is that of pinched sets and
pinched invariant graphs \cite{glendinning:2002, stark:2003,
  FabbriJaegerJohnsonKeller2005ForcedSharkovskii,JaegerStark2006Classification,Jaeger2007StructureOfSNA}.  Suppose
that $\Theta$ is a compact metric space, $X=[a,b]\ssq\R$ and
$\varphi^{-},\varphi^{+}:\Theta\rightarrow X$. Further, assume that
$\varphi^{-}$ is lower semi-continuous and $\varphi^{+}$ is upper
semi-continuous and $\varphi^-\leq\varphi^+$. Then $\varphi^{-}$ and
$\varphi^{+}$ are called {\em pinched} if there exists a point $\theta\in\Theta$
with $\varphi^{-}(\theta)=\varphi^{+}(\theta)$. If we only have that for any
$\eps>0$ there exists $\theta_\eps$ with
$|\varphi^+(\theta_\eps)-\varphi^-(\theta_\eps)|<\eps$, then we call $\varphi^+$
and $\varphi^-$ {\em weakly pinched}. In the case of random forcing, we have the
following measure-theoretic analogue.  Suppose $(\Theta,\mathcal{B},\mu)$ is a
measure space, $X=[a,b]\subseteq\mathbb{R}$ and
$\varphi^{-}\leq\varphi^{+}:\Theta\rightarrow X$ are measurable. Then
$\varphi^{-}$ and $\varphi^{+}$ are called {\em measurably pinched} if the set
$A_{\delta}:=\{\theta\in\Theta\mid
\varphi^{+}(\theta)-\varphi^{-}(\theta)<\delta\}$ has positive measure for all
$\delta>0$. Otherwise, we call $\varphi^{-}$ and $\varphi^{+}$ $\mu$-uniformly
separated.  
For strictly ergodic (that is, minimal and uniquely ergodic)
forcing processes, all three notions of pinching coincide, see \cite[Lemma
  3.5]{AnagnostopoulouJaeger2012SaddleNodes}. In this case, two pinched
invariant graphs always coincide on a dense subset of $\Theta$.

\subsection{Fold bifurcation scenario.}  With the above notions, we can now 
formulate the bifurcation scenario -- both in a deterministic and a random
setting -- which is taken from \cite{AnagnostopoulouJaeger2012SaddleNodes} and
will provide the general framework for our further studies.  We start with the
deterministic case. Given $A\ssq\Theta\times X$ and $\theta\in\Theta$, we let
$A_\theta=\{x\in X\mid (\theta,x)\in A\}$.

\begin{theorem}[{\cite[Theorem 6.1]{AnagnostopoulouJaeger2012SaddleNodes}}]
\label{qpf}
\noindent Let $\omega$ be a flow on a compact metric space $\Theta$ and suppose
$(\Xi_{\beta})_{\beta\in[0,1]}$ is a parameter family of $\omega$-forced
monotone $\mathcal{C}^{2}$ flows. Further assume that there exist continuous
functions $\gamma^{-},\gamma^{+}:\Theta\rightarrow X$ with
$\gamma^{-}<\gamma^{+}$ such that the following conditions hold for all
$\beta\in[0,1]$, $\theta\in\Theta$ and all $t\geq 0$, where applicable.
\begin{enumerate}
  \item[(i)] There exist two distinct continuous $\Xi_{0}$-invariant graphs and
    no $\Xi_{1}$-invariant graph in $\Gamma=\{(\theta,x)\mid
    \gamma^-(\theta) < x < \gamma^+(\theta)\}$;
  \item[(ii)]
    $\xi^t_\beta(\theta,\gamma^{\pm}(\theta))\leq\gamma^{\pm}(\omega^t(\theta))$;
  \item[(iii)] the maps
    $(\beta,\theta,x)\mapsto\partial_{x}^{i}\xi^t_{\beta}(\theta,x)$ with
    $i=0,1,2$ and $(\beta,\theta,x)\mapsto\partial_{\beta}
    \xi^t_{\beta}(\theta,x)$ are continuous;
  \item[(iv)] $\partial_x \xi^t_{\beta}(\theta,x)>0~\forall~x\in\Gamma_{\theta}$;
  \item[(v)] $\partial_{\beta}
    \xi^t_{\beta}(\theta,x)<0~\forall~x\in\Gamma_{\theta}$;
  \item[(vi)]
    $\partial^2_x\xi^t_{\beta}(\theta,x)<0~\forall~x\in\Gamma_{\theta}$.
\end{enumerate}
Then there exists a unique critical parameter $\beta_{c}\in[0,1]$ such that
\begin{itemize}
  \item If $\beta<\beta_{c}$, then there exist two continuous
    $\Xi_{\beta}$-invariant graphs $\varphi_{\beta}^{-}<\varphi_{\beta}^{+}$ in
    $\Gamma$. For any $\omega$-invariant measure $\mu$ we have
    $\lambda_{\mu}(\varphi_{\beta}^{-})>0$ and
    $\lambda_{\mu}(\varphi_{\beta}^{+})<0$.
  \item If $\beta=\beta_{c}$, then either there exists exactly one continuous
    $\Xi_{\beta}$-invariant graph $\varphi_{\beta}$ in $\Gamma$ ({\em smooth
    bifurcation}), or there exists a pair of weakly pinched
    $\Xi_{\beta}$-invariant graphs $\varphi_{\beta}^{-}\leq\varphi_{\beta}^{+}$ in
    $\Gamma$ with $\varphi_{\beta}^-$ lower and $\varphi_{\beta}^+$ upper
    semi-continuous ({\em non-smooth bifurcation}). If $\mu$ is an $\omega$-invariant
    measure, then in the first case $\lambda_{\mu}(\varphi_{\beta})=0$. In the
    second case $\varphi_{\beta}^{-}(\theta)=\varphi_{\beta}^{+}(\theta)$
    $\mu$-almost surely implies $\lambda_{\mu}(\varphi^{\pm}_{\beta})=0$,
    whereas $\varphi_{\beta}^{-}(\theta)<\varphi_{\beta}^{+}(\theta)$
    $\mu$-almost surely implies $\lambda_{\mu}(\varphi_{\beta}^{-})>0$ and
    $\lambda_{\mu}(\varphi_{\beta}^{+})<0$.
  \item If $\beta>\beta_{c}$, then no $\Xi_{\beta}$-invariant graphs exist in $\Gamma$.
\end{itemize}
\end{theorem}

\begin{remark}
  \alphlist
\item  We note that the result in \cite{AnagnostopoulouJaeger2012SaddleNodes}
is stated for convex fibre maps but the above version for concave fibre maps is
equivalent and discussed in \cite[Remark 6.2
  (c)]{AnagnostopoulouJaeger2012SaddleNodes}.
\item Likewise, the statement in \cite{AnagnostopoulouJaeger2012SaddleNodes} is
  given for the closed region $\overline{\Gamma}$ instead of the open set
  $\Gamma$ that we use here (for convenience later on), but the proof in
  \cite{AnagnostopoulouJaeger2012SaddleNodes} can be adjusted with minor
  modifications.
\item 
Non-continuous invariant graphs with negative Lyapunov exponents, as they appear
in non-smooth fold bifurcation of quasiperiodically forced systems, are called
{\em strange non-chaotic attractors} (SNA)
\cite{grebogi/ott/pelikan/yorke:1984,keller:1996,stark:2003,FeudelKuzetsovPikovsky2006StrangeNonchaoticAttractors,nunez/obaya:2007,AnagnostopoulouJaeger2012SaddleNodes}.
\listend
\end{remark}

\begin{remark}\label{rem: flows induced by odes}
  Continuous-time skew product flows are typically defined via non-autonomous
  ODE's of the form
  \begin{equation}
    \label{e.nonautonomous_ode}
    x'(t)=F(\omega^t(\theta),x) \ . 
  \end{equation}
   In fact, \eqref{e.nonautonomous_ode} a priori only yields a \emph{local} flow where trajectories
   may diverge and hence not be defined for all times $t\in \R$.
   As we will only deal with bounded solutions (see also Lemma~\ref{lem01}), this issue is not of further importance.
   We refer the interested reader to \cite{Fuhrmann2016SNAinFlows} for more details.
  
   Now, in order to apply the above statements to flows defined by equations of the form \eqref{e.nonautonomous_ode},
   it is crucial that the validity of the assumptions can be read off directly from the differential
   equations. 
   Fortunately, there is a rather immediate translation between the
   properties of parameter families of non-autonomous vector fields $F_\beta$
   and the relevant properties of the resulting skew product flow.

   First, the curves $\gamma^\pm$ can usually be chosen constant, in which case
   (i) may be obvious or be checked checked by hand for the respective models
   and (ii) follows from $F_\beta(\theta,\gamma^\pm(\theta))< 0$ for all
   $\theta\in\Theta$.  Secondly, by standard results on the regularity of
   solutions of an ODE with respect to parameters, it suffices to assume that
   for each $\theta \in \Theta$ the mapping $[0,1]\times \R \times \R
   \ni(\beta,t,x)\mapsto F_\beta(\omega^t(\theta),x)$ is continuous, $\mathcal
   C^1$ with respect to $\beta$, and $\mathcal C^2$ with respect to $x$ in order
   to ensure that $\Xi_{\beta}$ is indeed $\mathcal C^2$ and continuously
   differentiable with respect to $\beta$.  Hence, the above expressions are
   well-defined and (iii) is verified, as well.  The monotonicity in (iv)
   follows immediately from the uniqueness of the solutions to
   \eqref{e.nonautonomous_ode}.  The monotonicity condition (v) always holds if
   $\beta\mapsto F_\beta(\omega^t(\theta),x)$ is monotonically decreasing.
   Finally, the concavity of the fibre maps required in (vi) is a consequence of
   the concavity of $F_\beta$ in the considered region.  We refer to
   \cite{AnagnostopoulouJaeger2012SaddleNodes,Fuhrmann2016SNAinFlows}, as well
   as to the discussion of the application to the forced Allee model in
   Section~\ref{ApplicationToAllee}, for further details.
\end{remark}

The above remarks equally apply to the following measure-theoretic version of
the bifurcation scenario.
\begin{theorem}[{\cite[Theorem 4.1]{AnagnostopoulouJaeger2012SaddleNodes}}]
\label{random}
Let $(\Theta,\mathcal{B},\mu,\omega)$ be a measure preserving dynamical system
and suppose $(\Xi_{\beta})_{\beta\in[0,1]}$ is a parameter family of
$\omega$-forced monotone $\mathcal{C}^{2}$ flows. Further assume that there
exist measurable functions $\gamma^{-},\gamma^{+}:\Theta\rightarrow X$ with
$\gamma^{-}<\gamma^{+}$ such that the following conditions hold for all
$\beta\in[0,1]$, $\mu$-almost all $\theta\in\Theta$ and all positive $t\in \T$, where applicable.
\begin{enumerate}
 \item[(i)] There exist two $\mu$-uniformly separated $(\Xi_{0},\mu)$-invariant
   graphs but no $(\Xi_{1},\mu)$-invariant graph in $\Gamma$;
  \item[(ii)]$\xi^t_\beta(\theta,\gamma^{\pm}(\theta))\leq\gamma^{\pm}(\omega^t(\theta))$;
  \item[(iii)] the maps $(\beta,t,x)\mapsto \xi^t_\beta(\theta,x)$ and
    $(\beta,t,x)\mapsto \partial_x\xi^t_\beta(\theta,x)$ are continuous;
  \item[(iv)] $\partial_x\xi^t_{\beta,\theta}(x)>0~\forall~x\in\Gamma_{\theta}$;
  \item[(v)] for some $t>0$ there exist constants $C<c_{1}\leq 0$ such that
    $C\leq\partial_{\beta} \xi^t_\beta(\theta,x)\leq
    c_{1}~\forall~x\in\Gamma_{\theta}$;
  \item[(vi)]
    $\partial_x^2\xi^t_\beta(\theta,x)<0~\forall~x\in\Gamma_{\theta}$;
  \item[(vii)]the function $\eta(\theta)=\sup \left\{\left.|\log
    \partial_x\xi^t_\beta(\theta,x)|\right|x\in\Gamma_{\theta},\beta\in[0,1]\right\}$
    is integrable with respect to $\mu$.
\end{enumerate}
Then there exists a unique critical parameter $\beta_{\mu}\in[0,1]$ such that
\begin{itemize}
  \item If $\beta<\beta_{\mu}$, then there exist exactly two
    $(\Xi_{\beta},\mu)$-invariant graphs
    $\varphi_{\beta}^{-}<\varphi_{\beta}^{+}$ in $\Gamma$ which are
    $\mu$-uniformly separated and satisfy $\lambda(\varphi_{\beta}^{-})>0$ and
    $\lambda(\varphi_{\beta}^{+})<0$.
  \item If $\beta=\beta_{\mu}$, then either there exists exactly one
    $(\Xi_{\beta},\mu)$-invariant graph $\varphi_{\beta}$ in $\Gamma$, or there
    exist two measurably pinched invariant graphs
    $\varphi_{\beta}^{-}\leq\varphi_{\beta}^{+}$ in $\Gamma$. In the first case,
    $\lambda_{\mu}(\varphi_{\beta})=0$; in the second case,
    $\lambda_{\mu}(\varphi_{\beta}^{-})>0$ and
    $\lambda_{\mu}(\varphi_{\beta}^{+})<0$.
  \item If $\beta>\beta_{\mu}$, then no $f_{\beta}$-invariant graphs exist in $\Gamma$.
\end{itemize}
\end{theorem}
In analogy to the deterministic setting, we again speak of a smooth bifurcation
if there exists a unique neutral invariant graph at the bifurcation point, and
of a non-smooth bifurcation if there exists an attractor-repeller pair.

\subsection{Application to the forced Allee model}
\label{ApplicationToAllee} 

We now aim to verify that the forced Allee model (\ref{e.forced_Allee})
satisfies the assumptions of Theorems \ref{qpf} and \ref{random},
respectively. More precisely, we specify the admissible parameter ranges stated
in Remark~\ref{r.parameter_range} in the introduction and show that the
respective conditions are met for all admissible parameters.  As pointed out in
Remark~\ref{rem: flows induced by odes}, conditions (iii), (iv) and (v) in
Theorems~\ref{qpf} and \ref{random} follow directly from the specific form of
the scalar field (\ref{e.forced_Allee}). Moreover, condition (vii) in
Theorem~\ref{random} holds as well, since all the involved functions are bounded
(and therefore, in particular, integrable). However, it remains to specify a
suitable parameter range and appropriate functions $\gamma^\pm$ so that (i),
(ii) and (vi) hold.

 It is easy to check that the fold bifurcation of the unforced equation
 (\ref{e.simple_Allee}) takes place at
  \begin{equation}\label{e.brKS}
   \beta \ = \ b(r,K,S) \ := \ \frac{r}{K^2}\cdot \left(\frac{K-S}{2}\right)^2 
   \ .
  \end{equation}
Moreover, the neutral equilibrium point at the bifurcation is
$x_0=\frac{K+S}{2}$.  If $\kappa<b(r,K,S)$ and $\beta\leq b(r,K,S)-\kappa$, then
we have that $V_\beta(\theta,x_0)>0$ for all $\theta\in \Theta$ and both forcing
terms (\ref{e.quasiperiodic_forcing_term_example}) and
(\ref{e.random_forcing_term_example}) (note here that $F\leq 1$). At the same
time, given $\beta<b(r,K,S)$, the unforced Allee model (\ref{e.simple_Allee})
has equilibrium points $x=0$ and
\[
x_\beta^\pm \ = \ \frac{K+S}{2} \ \pm \halb\sqrt{(K-S)^2 - \frac{4\beta K^2}{r}}
\ = \ \frac{K+S}{2} \ \pm \frac{K-S}{2} \cdot \sqrt{1-\bar \beta} \ ,
\]
where
\[
\bar \beta \  = \ \frac{4\beta K^2}{r(K-S)^2} \ .
\]
As the forcing is always downwards (recall that the forcing term is $-\kappa F$
with $F\geq 0$), this implies in particular that
$V_\beta(\theta,x^\pm_{\beta_0})<0$ for all $\theta\in\Theta$ and $\beta\geq
\beta_0$. Hence, we obtain a forward invariant region $\Theta\times
     [x_0,x^+_{\beta_0}]$ and a backward invariant region $\Theta\times
     [x^-_{\beta_0},x_0]$, where $\beta_0$ will be specified below. Using the
     concavity of $V_\beta$, equally shown below, this implies the existence of
     two invariant graphs in
     $[\gamma^-,\gamma^+]=[x^-_{\beta_0},x^+_{\beta_0}]$. Similarly, if
     $\beta>b(r,K,S)$, then the bifurcation has already taken place and there
     will not be any invariant graph above the equilibrium at $0$. Hence,
     conditions (i) and (ii) are satisfied.

It remains to ensure the concavity of $V_\beta(\theta,\cdot)$ in the considered
region $\Theta\times(\gamma^-,\gamma^+)$, where $\gamma^\pm$ still need to be
specified. The second derivative of $V_{\kappa,\beta}$ with respect to $x$ is
given by
\[
\partial_x^2 V_{\kappa,\beta}(\theta,x) \ = \ \frac{r}{K^2}\cdot (-6x + 2(K+S))
\ ,
\]
and is thus independent of $\beta$ and $\kappa$. We have
\[
\partial_x^2 V_{\kappa,\beta}(\theta,x) < 0 \quad \equi \quad x \ >
\ \frac{K+S}{3} \ .
\]
Hence, we simply need to choose $\beta_0$ such that
$x^-_{\beta_0}\geq \frac{K+S}{3}$. By the above, this means that we require
\[
x^-_\beta \ = \ \frac{K+S}{2} \ - \frac{K-S}{2} \cdot \sqrt{1-\bar \beta} \ \geq
\ \frac{K+S}{3} \ ,
\]
which is equivalent to
\[
\bar\beta \ \geq \ 1-\gamma(K,S)\ ,
\]
where $\gamma(K,S)=\frac{1}{9}\left(\frac{K+S}{K-S}\right)^2$, and hence to
\[
\beta \ \geq \ b(r,K,S)\cdot \left( 1-\gamma(K,S)\right) \ .
\]
This means that if $\kappa>0$ satisfies
\[
  \kappa \ < \ b(r,K,S) \cdot \gamma(K,S)
  \]
  and we let
  \[
      J(r,K,S) \ = \ \left[ b(r,K,S)\cdot \left(
      1-\gamma(K,S)\right),b(r,K,S)+1\right] \ ,
  \]
then the parameter family $(V_{\kappa,\beta})_{\beta\in J(r,K,S)}$ satisfies all
the assertions of Theorem~\ref{qpf} (modulo rescaling the parameter interval
$J(r,K,S)$) and therefore undergoes a non-autonomous fold bifurcation.

\subsection{Forcing processes}\label{ForcingProcesses}
For later use, we introduce forcing processes both in discrete and continuous time. 
Quasiperiodic motion in discrete time is given by a rotation
$\omega:\T^d\to\T^d,\ \theta\mapsto \theta+\rho\bmod 1$ which is {\em
  irrational}, in the sense that its rotation vector $\rho=(\rho_1\ld \rho_d)$
has incommensurate entries.\foot{Here $\rho_1\ld \rho_d$ are called {\em
    incommensurate} if $n_0+\sum_{j=1}^d n_j\rho_j=0$ implies
  $n_0=n_1=\ldots=n_d=0$.} In this case, the transformation $\omega$ is minimal
and uniquely ergodic, with the Lebesgue measure on $\T^d$ as the unique
invariant probability measure. The continuous time analogue is an irrational
Kronecker flow $\omega:\R\times\T^d\to\T^d,\ \omega^t(\theta)=\theta+t\rho$,
where $\rho$ (or some scalar multiple thereof) is again incommensurate.

In order to model random forcing in discrete time, we will simply use Bernoulli
processes as examples. Hence, we let $\Sigma=\{0,1\}^\Z$ and equip this space
with the Bernoulli measure $\mu$ with probabilities $1/2$ for the symbolds $0$
and $1$.  Actually, we could likewise set $\Sigma$ to be $[0,1]^\Z$ and $\mu$ to
be the infinite product $\Leb_{[0,1]}^\Z$ of the Lebesgue measure on $[0,1]$, or
even replace $\Leb_{[0,1]}$ by any measure $\lambda$ on $[0,1]$ whose
topological support includes $0$ and $1$.  In any case, the dynamics on $\Sigma$
are given by the shift map $\sigma: \Sigma\ni (\theta_n)_{n\in \Z} \mapsto
(\theta_{n+1})_{n\in \Z}$ which is ergodic with respect to each such measure.

A slight complication occurs in the case of continuous-time random forcing. As
mentioned in the introduction, we would like to use $\sin(W_t)$ as a forcing
term in (\ref{e.forced_Allee}). Hence, it is natural to consider the Wiener space,
that is, the space of continuous real-valued functions $\cC(\R,\R)$ equipped
with the Borel $\sigma$-algebra generated by uniform topology and the classical
Wiener measure $\Proj$. However, in order to obtain a skew product flow we need
a measure-preserving transformation on our probability space. If
$\theta\in\cC(\R,\R)$ is a path of a Brownian motion, the standard
measure-preserving shift on the Wiener space is given by
$\omega^t(\theta)(s)=\theta(s+t)-\theta(t)$. The problem that occurs is the fact
that if we now want to define a forcing term $f$ on $\cC(\R,\R)$ by
evaluating the sinus at $\theta(0)$, that is, $f(\theta)=\sin(2\pi\theta(0))$,
then $f(\omega^t(\theta))=0$ for all $t\in\R$ (the standard Brownian motion
starts in zero, and the classical shift respects this property).  We therefore
use a slightly modified version of this process to model bounded random forcing
for our purposes. To that end, we let $p:\cC(\R,\R)\to \cC(\R,\kreis)=\Theta$ be
the projection of real-valued to circle-valued functions (induced by the
canonical projection $\pi:\R\to\kreis$) and let $\Proj_0=p_*\Proj$ be the
push-forward of $\Proj$. Further, we let $S:\kreis\times\Theta\to \Theta,\,
(x,\theta)\mapsto \theta+x$ and equip $\Theta$ with the measure
$\nu=S_*(\Leb_{\kreis}\times \Proj_0)$. By definition, $\nu$ has equidistributed
marginals and can therefore be seen to be invariant under the shift
$\omega:\R\times\Theta\to\Theta$ defined by
$\omega^t(\theta)(s)=\theta(t+s)$. This construction will allow us to define a
forcing term simply by evaluating the sinus (viewed as a function on
$\kreis$) at $\theta(0)$.

\subsection{A simplified discrete-time model} \label{DiscreteTimeModel} As a basic model for the discrete-time case, we will consider the 
parameter families of skew product maps
\begin{equation}
  \label{e.discrete_time_examples}\textstyle
  f_\beta : \Theta\times \R \to\Theta\times\R \quad , \quad (\theta,x)\mapsto
  \left(\omega(\theta),\arctan(\alpha x) - \kappa \cdot F(\theta) - \beta\right)
  \
\end{equation}
with real parameters $\alpha>\pi/2$, $\beta\in [0,1]$ and $\kappa\in (0,\tilde \beta_c)$, where
$\tilde \beta_c=\arctan(\sqrt{\alpha -1})-\sqrt{\alpha-1}/\alpha$ is the critical parameter at which the fold bifurcation of the autonomous family
$x\mapsto \arctan(\alpha x) - \beta$ occurs.
The forcing processes we consider are either defined on $\Theta=\T^d$ and given by a rotation
$\omega:\theta\mapsto\theta+\rho$ with rotation vector $\rho\in\T^d$
(quasiperiodic forcing) or on $\Theta=\Sigma$, where $\Sigma=\{0,1\}^\Z$ and
$\omega$ is given by the shift $\sigma$ on $\Sigma$ (random forcing), all as
in Section~\ref{ForcingProcesses} above.
For the forcing function $F$, we use
\begin{equation}
  F(\theta) \ = \ \frac{\sin(2\pi\theta)+1}{2}
\end{equation}
in the qpf case and 
\begin{equation}
  F(\theta) \ = \ \theta_0
\end{equation}
in the random case.

\begin{figure}[h!]
\includegraphics[scale=0.35]{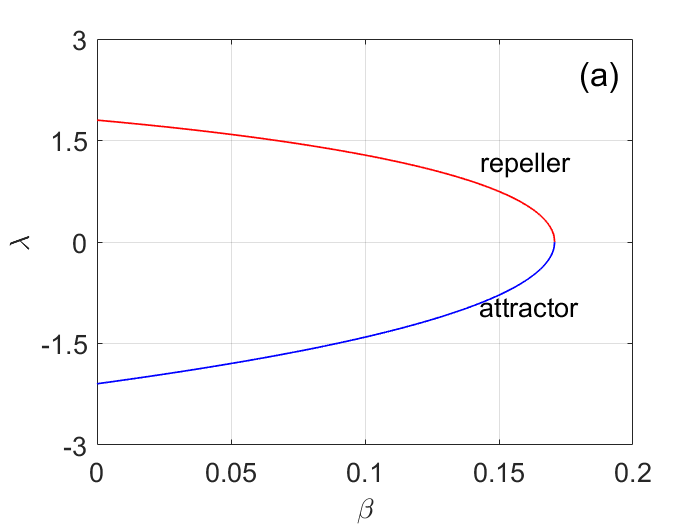} \quad
\includegraphics[scale=0.44]{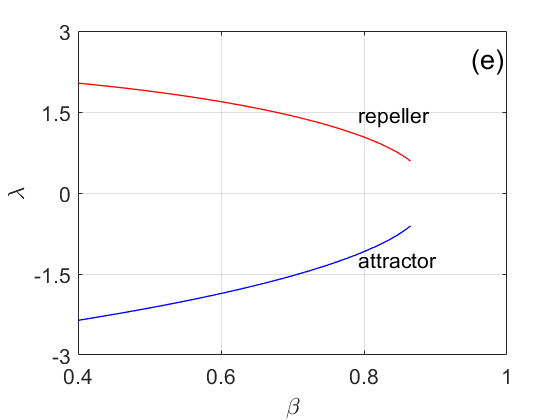}

\includegraphics[scale=0.35]{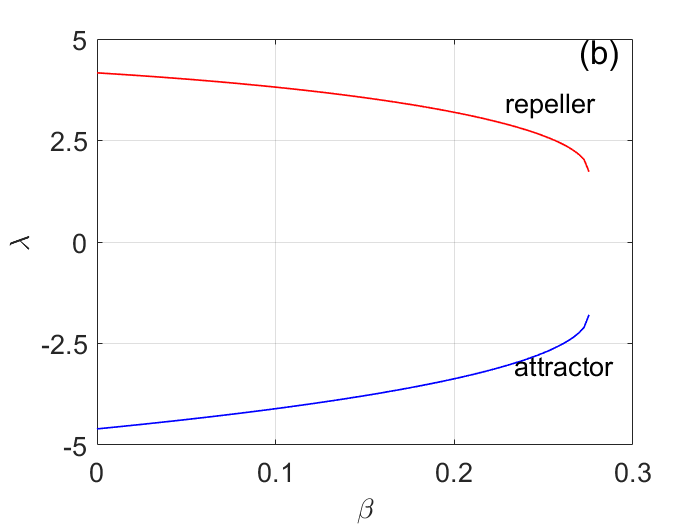} \quad 
\includegraphics[scale=0.44]{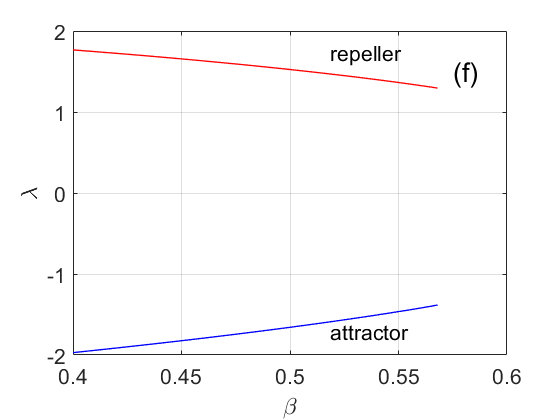}

\includegraphics[scale=0.35]{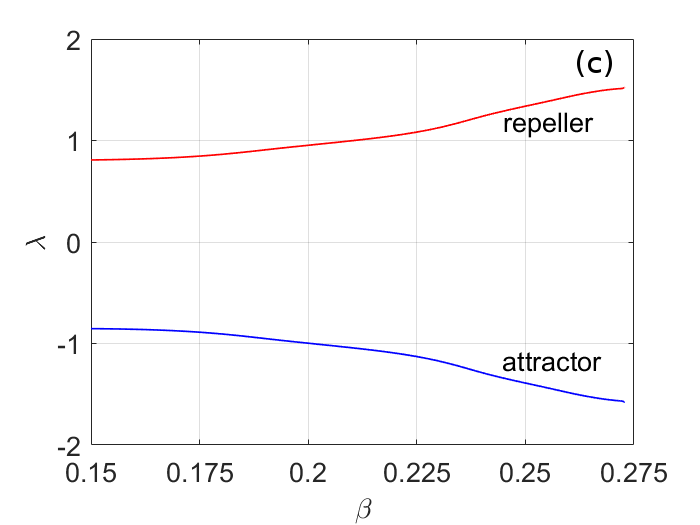}\quad 
\includegraphics[scale=0.43]{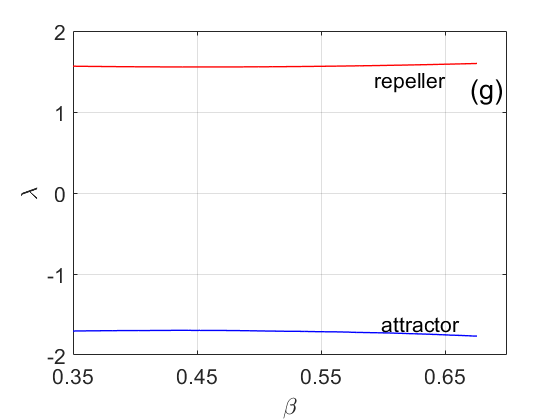}

\includegraphics[scale=0.43]{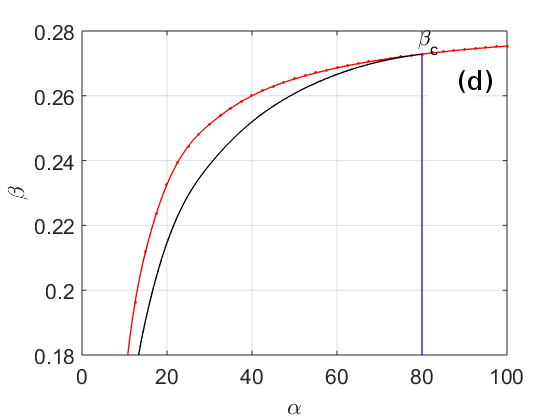}\quad
\includegraphics[scale=0.43]{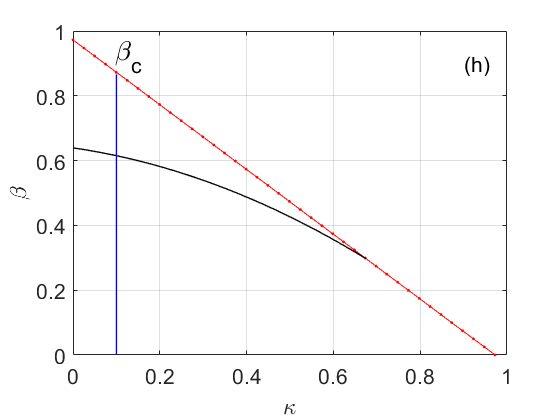}

  \caption{Lyapunov exponents during saddle-node bifurcations in the
    family~(\ref{e.discrete_time_examples}).  (a) smooth bifurcation in the qpf
    case, with parameters $\alpha=10$ and $\kappa=1$. The bifurcation occurs at
    $\beta_c=0.341502$.  (b) a non-smooth bifurcation in the same model, with
    parameters $\alpha=100$ and $\kappa=1$. The bifurcation occurs at
    $\beta_c=0.5507468$.  (c) non-smooth bifurcation with simultaneous variation
    of parameters $\alpha$ and $\beta$ along the black curve in (d).  (e) and
    (f) Lyapunov exponents in the randomly forced case, with parameters
    $\alpha=10$ and $\kappa=0.1$ and bifurcation parameter $\beta=0.866$ in (e)
    and $\alpha=10$ and $\kappa=0.4$ and bifurcation parameter $\beta=0.566$ in
    (f). In (g), the parameters $\kappa$ and $\beta$ are varied again at the
    same time along the black parameter curve shown in
    (h).}\label{f.LE_smooth/nonsmooth_discretetime}
\end{figure}

The behaviour of the attractors and repellers during smooth and non-smooth
bifurcations in the qpf case are shown in
Figure~\ref{f.forced_folds_discretetime}. This figure also illustrates some key
features of non-smooth bifurcations in qpf systems and allows us to give a
heuristic description of the mechanism that causes the non-smoothness. The
rigorous description of this mechanism is the basis for the mathematical
analysis of non-smooth bifurcations in
\cite{Jaeger2009CreationOfSNA,fuhrmann2013NonsmoothSaddleNodesI}. As can be seen
in Figure~\ref{f.forced_folds_discretetime}, when the attracting and repelling
graphs approach each other in a non-smooth way, they develop a sequence of
`peaks'. These appear in an ordered way, and the next peak is always the image
of the previous one and is generated as soon as the latter reaches into the
region with large derivatives which is centred around the $0$-line
$\kreis\times\{0\}$. The first peak is located around the minimum of the blue
curve in (d). The second peak emerges in (e) and is fully developed in (f),
where a number of further peaks can be seen as well. Thereby, the movement of
each peak is amplified by the large derivatives close to zero (of magnitude
$\alpha$, see (\ref{e.discrete_time_examples})). For this reason, as $\beta$ is
increased, the speed by which the peaks move as $\beta$ is varied increases
exponentially with the order of the peak, whereas its width decreases
exponentially (since each peak is stretched vertically due to the expansion
around $0$). In the limit, the two curves touch each other with the tips of the
peaks. Note that only a finite number of peaks can be observed at the
bifurcation point in (f), since these quickly become too thin to be visible in
numerical simulations. However, it is known that the region between the two
graphs in (f) is actually filled densely by further peaks
\cite{GroegerJaeger2013SNADimensions,FuhrmannGroegerJaeger2015SNADimensions}. We
refer to the introduction of \cite{Jaeger2009CreationOfSNA} for a more detailed
discussion.

\begin{figure}[h!]
  \includegraphics[scale=0.35]{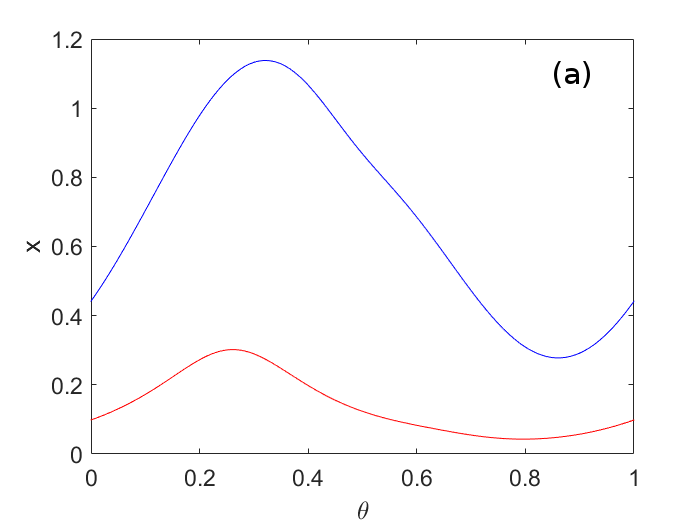} \qquad
  \includegraphics[scale=0.35]{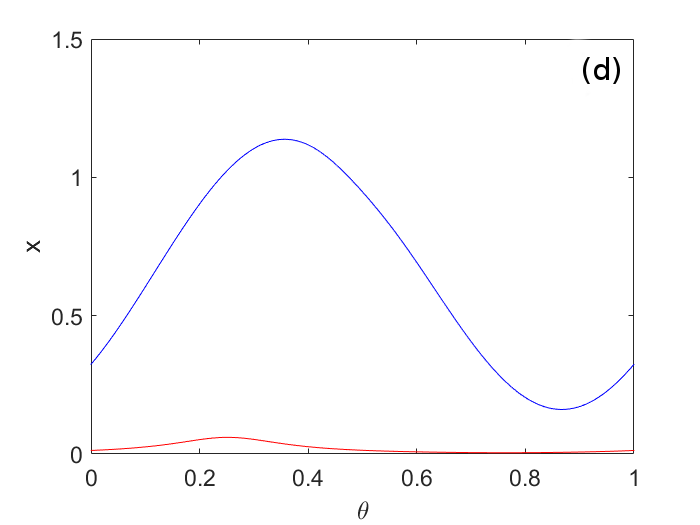}

 \includegraphics[scale=0.35]{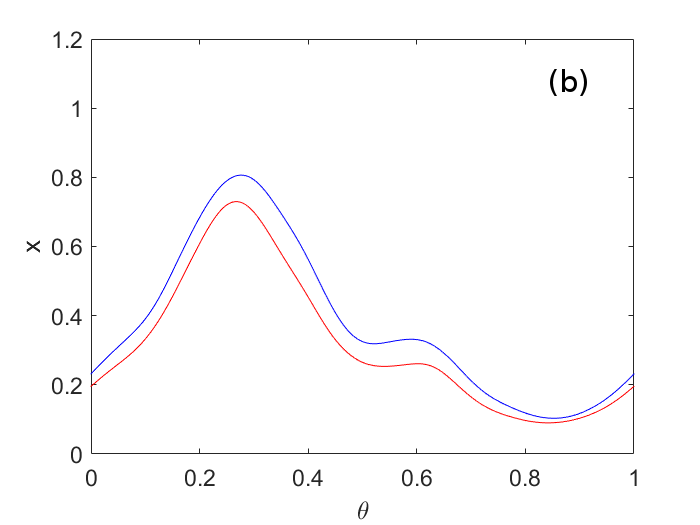} \qquad
 \includegraphics[scale=0.35]{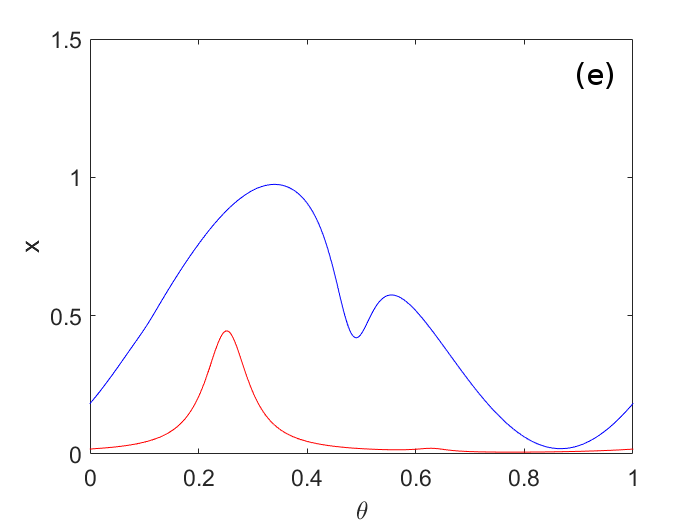}

\includegraphics[scale=0.35]{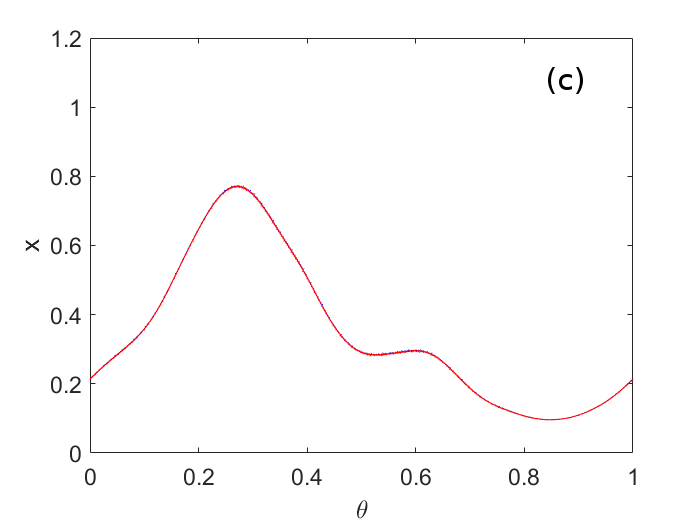} \qquad
\includegraphics[scale=0.35]{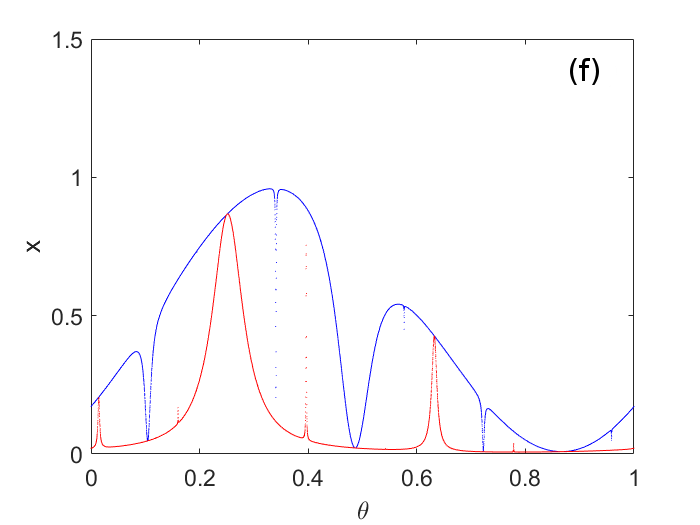}
    \caption{{\bf (a)--(c):} Smooth saddle-node bifurcation in
      (\ref{e.discrete_time_examples}) with quasiperiodic forcing. Parameter
      values are $\alpha=10$, $\kappa=1$, $\rho=\omega$ (golden mean) and (a) $\beta=0.1708$ (b) $\beta=0.34$ and
      (c) $\beta=0.341502$. \\ {\bf (d)--(f):} Non-smooth saddle-node bifurcation
      bifurcation in in (\ref{e.discrete_time_examples}) with quasiperiodic
      forcing. Parameter values are $\alpha=100$, $\kappa=1$, $\rho=\omega$ and (a)
      $\beta=0.4$ (b) $\beta=0.54$ and (c)
      $\beta=0.5507486$.} \label{f.forced_folds_discretetime}
   \end{figure}

The range of the finite time Lyapunov exponents for the same parameter
families as in Figure~\ref{f.LE_smooth/nonsmooth_discretetime} is shown in
Figure~\ref{f.ftle_discretetime}.
\begin{figure}[h!]
\includegraphics[scale=0.35]{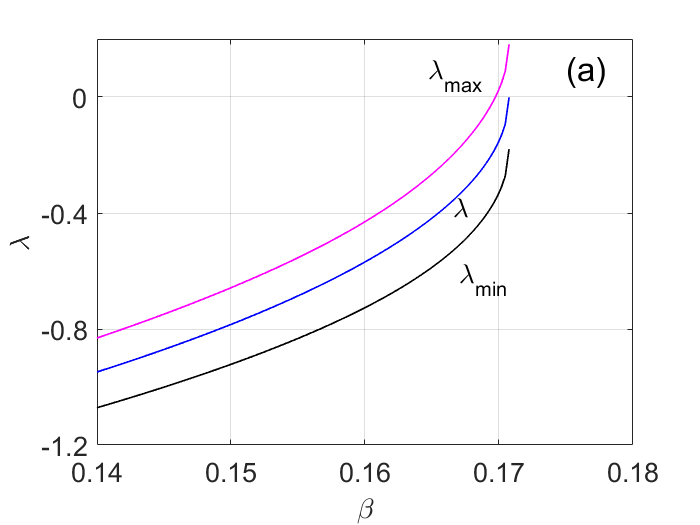} \quad
\includegraphics[scale=0.35]{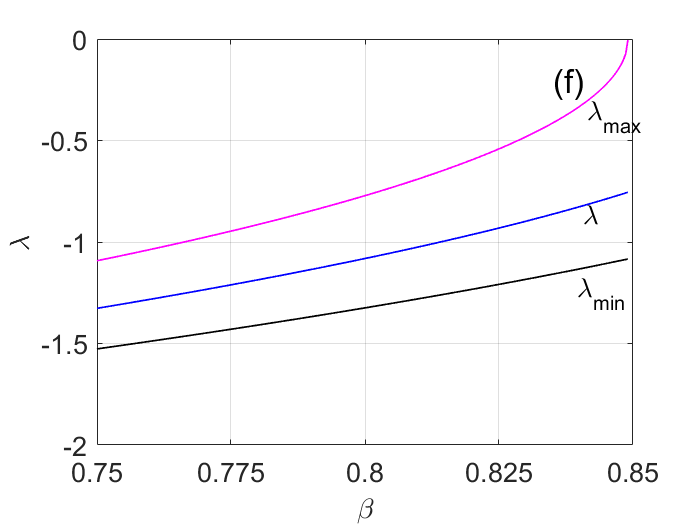}

\includegraphics[scale=0.35]{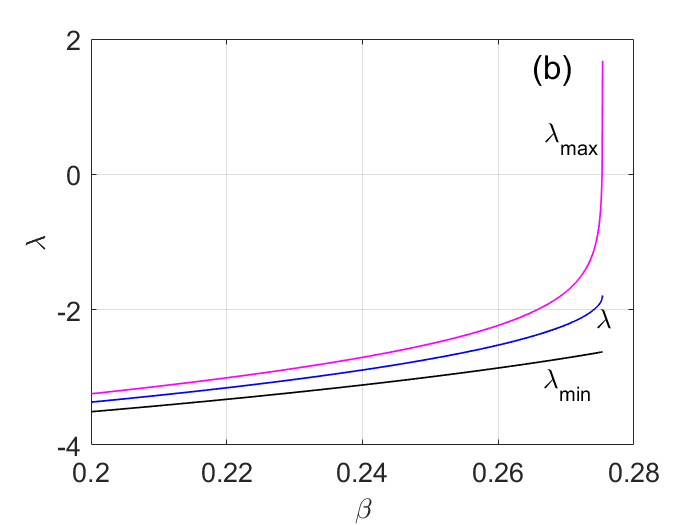}
\includegraphics[scale=0.44]{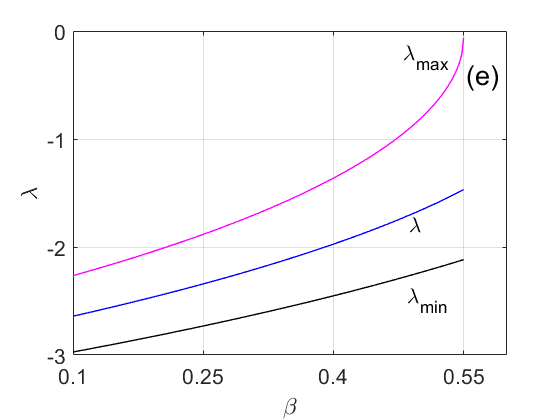}

\includegraphics[scale=0.44]{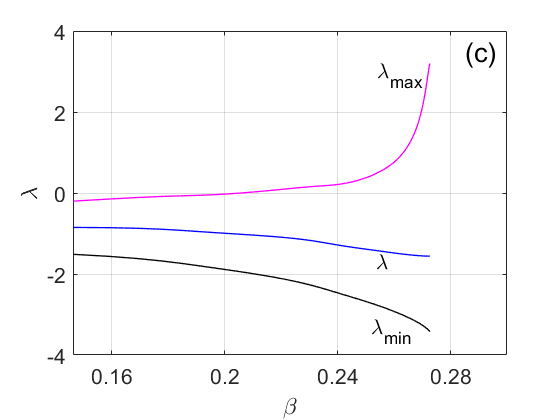}
\includegraphics[scale=0.44]{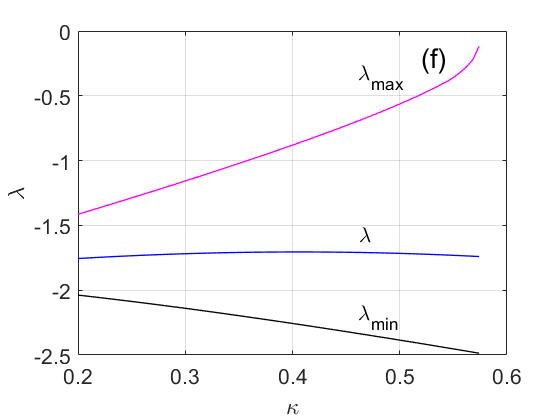}

\caption{Minimal and maximal finite-time Lyapunov exponents (with the asymptotic
  one plotted in the middle) for the parameter families used in
  Figure~\ref{f.LE_smooth/nonsmooth_discretetime} (in the same order). The time
  is $n=5$ in all cases.  }
\label{f.ftle_discretetime}
\end{figure}
Finally, the distribution of finite-time Lyapunov exponents (on different
time-scales) is plotted in Figure~\ref{f.distributions_discrete}, and the
corresponding relative frequencies as a function of time in
Figure~\ref{f.distributions_evolution_discrete}. In both cases, we restrict to
the qpf case.

\begin{figure}[h!]
   \includegraphics[scale=0.22]{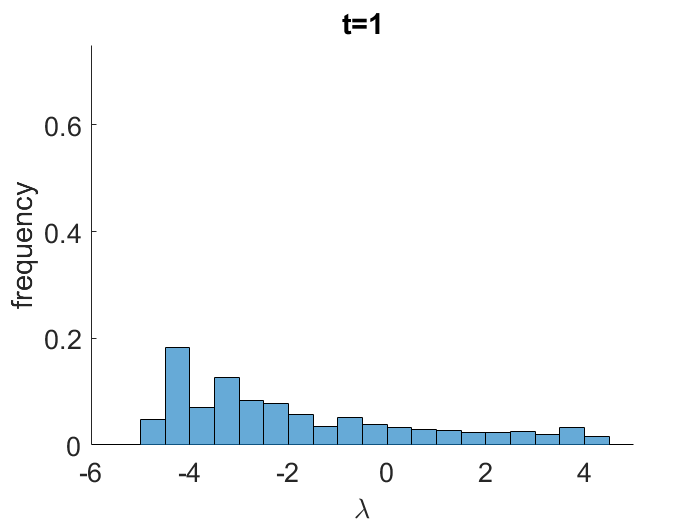} \quad
    \includegraphics[scale=0.22]{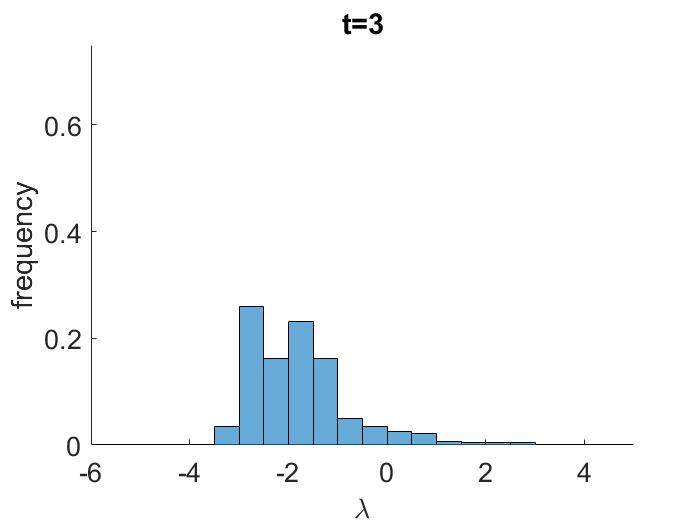} \quad
  \includegraphics[scale=0.22]{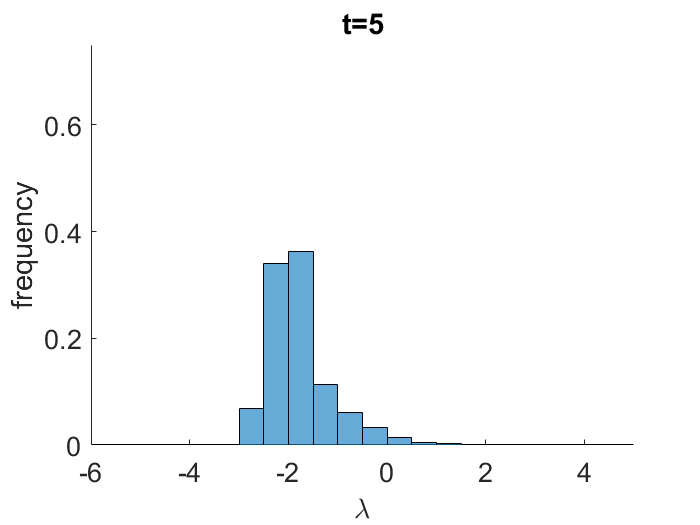}
\medskip

   \includegraphics[scale=0.23]{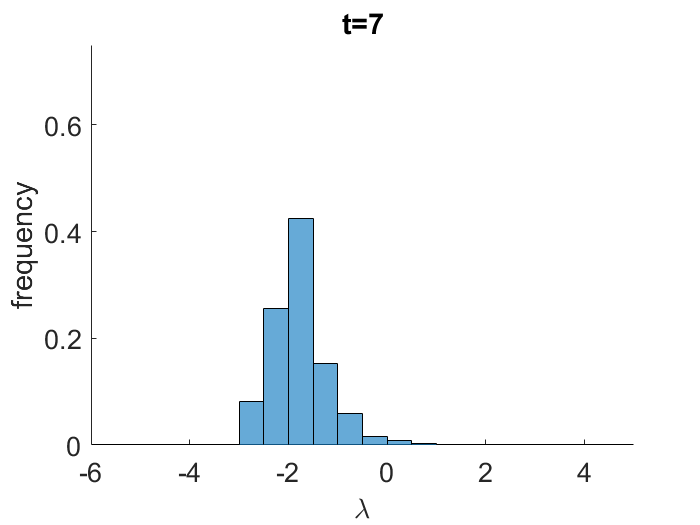} \quad
   \includegraphics[scale=0.23]{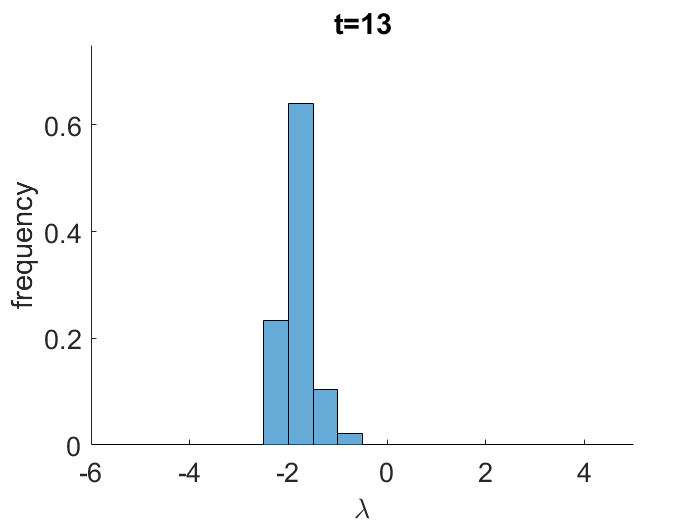} \quad
\includegraphics[scale=0.23]{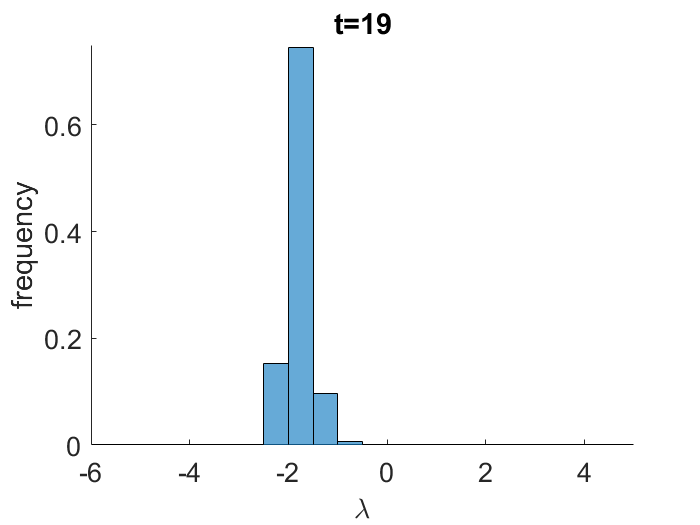}
  \caption{Distributions of the finite-time Lyapunov exponents in the qpf
    discrete-time system (\ref{e.discrete_time_examples}) with parameters
    $\alpha=100$, $\kappa=1$ and $\rho$ the golden
    mean. } \label{f.distributions_discrete}
\end{figure}

\begin{figure}[h!]
   \includegraphics[scale=0.23]{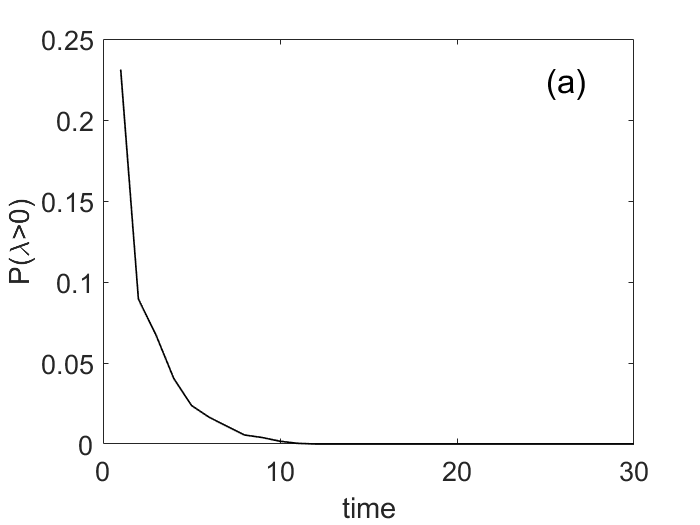} \quad
  \includegraphics[scale=0.23]{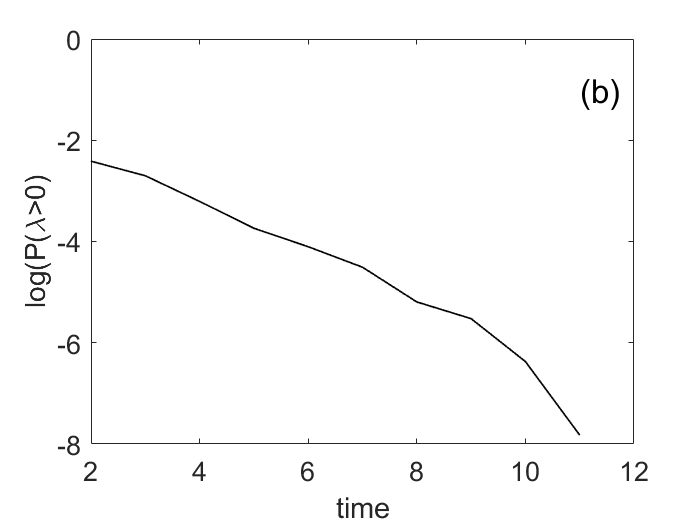} \quad
\includegraphics[scale=0.23]{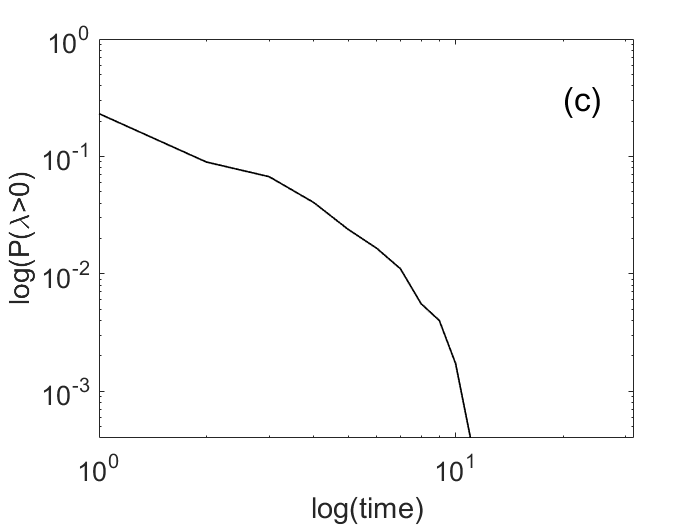} 
     \caption{A plot of the relative frequency of positive exponents (as
       observed in Figure~\ref{f.distributions_discrete}) on a (a) standard, (b)
       logarithmic and (c) $\log$-$\log$-scale.
     } \label{f.distributions_evolution_discrete}
\end{figure}
\medskip

\newpage

\section{Abundance of nonsmooth fold bifurcations} \label{Relevance}

\subsection{Quasiperiodic forcing}

\noindent The results discussed in the previous section provide a general
setting for non-autonomous saddle-node bifurcations. We shall now take a closer
look at non-smooth bifurcations and discuss their widespread occurrence in
forced systems. In the case of quasiperiodic forcing, the latter is
well-established by a number of rigorous results both in the discrete- and
continuous-time case. Thereby, arithmetic properties of the rotation numbers or
vectors in the base play a crucial role.  Given $ \tau,\kappa>0$, we say $\rho
\in \T^d$ is \emph{Diophantine (of type $(\tau,\kappa)$)} if
\begin{align*}
 \forall k \in \Z^d \setminus \{0\}:\inf_{p\in \Z}\left|p+\sum_{i=1}^d \rho_i k_i \right|\geq
 \tau |k|^{-\kappa}.
\end{align*}
We note thta the set of rotation vectors that satisfy a Diophantine condition of
type $(\tau,\kappa)$ for some $\tau>0$ (with $\kappa$ fixed) is of full Lebesgue
measure for all $\kappa>\frac{d+1}{d}$. 

Let us start by considering the discrete-time case.  Generalising example
(\ref{e.discrete_time_examples}), consider discrete-time flows given by
quasiperiodically forced monotone interval maps of the form
\begin{equation}\label{qpf1}
f:\Theta\times\R\rightarrow\Theta\times\R,\quad (\theta,x)\mapsto(\omega(\theta),
f_{\theta}(x))\ ,
\end{equation}
where $\Theta=\T^d$, $\omega:\Theta\to\Theta, \theta\mapsto\theta+\rho$ is again
an irrational rotation with rotation vector $\rho$ and $f_{\theta}(\cdot)$ is
$\mathcal{C}^{2}$ and strictly increasing on $X$.  For a given rotation vector
$\rho\in\T^d$, we further consider the space of one-parameter families
\begin{equation*} \mathcal{F_{\rho}}=\left\{(f_{\beta})_{\beta\in
    [0,1]}:f_{\beta} \textrm{ is of form \eqref{qpf1} and } (\beta,\theta,x)
  \mapsto f_{\beta,\theta}(x) \textrm{ is } \cC^2\right\}
\end{equation*}
equipped with the metric
\begin{equation*}
 d((f_{\beta})_{\beta\in [0,1]},(g_{\beta})_{\beta\in[0,1]})=\sup_{\beta\in
   [0,1]}(\|f_{\beta}-g_{\beta}\|_{2}+\|\partial_{\beta}f_{\beta}-\partial_{\beta}g_{\beta}\|_{0}).
\end{equation*}
Then the following result is established in
\cite{fuhrmann2013NonsmoothSaddleNodesI}, with precursors in
\cite{bjerkloev:2005,Jaeger2009CreationOfSNA}.

\begin{theorem}[\cite{fuhrmann2013NonsmoothSaddleNodesI}]
\label{fg}
Suppose that $\rho\in\T^d$ is Diophantine. Then there exists a non-empty open
set \ $\mathcal{U}\subseteq\mathcal{F_{\omega}}$ such that each
$(f_{\beta})_{\beta\in[0,1]}\in\mathcal{U}$ satisfies the assertions of
Theorem~\ref{qpf} and undergoes a non-smooth saddle-node bifurcation.
\end{theorem}

This confirms that the set of parameter families with a non-smooth saddle-node
bifurcation is large in a certain sense, and that the phenomenon can occur in a
robust way (both corresponding to the openness of the set $\cU$). Thereby, it is
important to note that in \cite{fuhrmann2013NonsmoothSaddleNodesI} the set $\cU$
in this result is characterised by explicit $\cC^2$-estimates. This makes it
possible to check if it contains a given parameter family and therefore provides
explicit examples of non-smooth saddle-node bifurcations.
\begin{cor}[\cite{fuhrmann2013NonsmoothSaddleNodesI}]\label{fg1}
If $\rho$ is Diophantine and $\alpha$ is sufficiently large, then the parameter
family $(f_\beta)_{\beta\in[0,1]}$ defined in (\ref{e.discrete_time_examples})
belongs to the set \ $\cU$ and hence undergoes a non-smooth saddle-node bifurcation.
\end{cor}

Continuous-time analogues of these results have been established in
\cite{Fuhrmann2016SNAinFlows}. In this case, one considers non-autonomous vector
fields, given by differentiable functions of the form
\begin{equation}\label{e.qpf_vector_fields}
  V : \T^d\times\R \to \R \ . 
\end{equation}
which induce quasiperiodically forced flows via the corresponding differential
equation
\begin{equation} \label{e.qpf_flow_equation}
  x'(t) \ = \ V(\omega_t(\theta_0),x(t)) \ , 
\end{equation}
where $\omega:\R\times\T^d\to\T^d,\ (t,\theta)\mapsto \theta+t\rho$ is an
irrational Kronecker flow with rotation vector $\rho\in\T^d$, as described
above. We let
\begin{equation}
  \cV \ = \ \left\{ (V_\beta)_{\beta\in[0,1]}\mid V \textrm{ is of the form
    (\ref{e.qpf_vector_fields}) and } (\beta,\theta,x)\mapsto V_\beta(\theta,x)
  \textrm{ is } \cC^2\right\} \
\end{equation}
and equip $\cV$ with the metric
\begin{equation*}
 d((V_{\beta})_{\beta\in [0,1]},(W_{\beta})_{\beta\in[0,1]})=\sup_{\beta\in
   [0,1]}(\|V_{\beta}-W_{\beta}\|_{2}+\|\partial_{\beta}V_{\beta}-\partial_{\beta}W_{\beta}\|_{0}).
\end{equation*}

\begin{theorem}[\cite{Fuhrmann2016SNAinFlows}] For any Diophantine $\rho\in\T^d$, 
there exists an open set \ $\cU_\rho\ssq\cV$ such that for any
$(V_\beta)_{\beta\in[0,1]}\in\cV$ the flow induced by
(\ref{e.qpf_vector_fields}) satisfies the assertions of Theorem~\ref{qpf} and
undergoes a non-smooth fold bifurcation.
\end{theorem}

Again, the explicit characterisation of the set $\cU_\rho$ given in
\cite{Fuhrmann2016SNAinFlows} makes it in principle possible to check for
non-smooth bifurcations in specific examples. However, this is considerably more
technical than in the discrete-time case. Moreover, the application to the
forced Allee model (\ref{e.forced_Allee}) with quasiperiodic forcing
(\ref{e.quasiperiodic_forcing_term_example}) would require a number of highly
non-trivial and technical modifications. Therefore, we refrain from going into
more details here and just point out that Figures~\ref{f.LE_smooth/nonsmooth}(b)
and \ref{f.forced_folds}(e)-(h) provide substantial numerical evidence for the
occurrence of non-smooth fold bifurcations in this case.

\subsection{Random forcing}

In contrast to the quasiperiodic case, the influence of bounded random noise on
saddle-node bifurcations has not been studied systematically so far. Our aim for
the remainder of this section is to establish the occurrence of non-smooth
bifurcations in a broad class of randomly forced monotone flows and maps. To
that end, we introduce the notion of an autonomous reference system. 
Let $\gamma^-<\gamma^+\in\R$ and suppose $(g_\beta)_{\beta\in[0,1]}$ is a one-parameter
family of differentiable flows $g_\beta:\T\times \R\to\R$ with the following
properties (which are supposed to hold for all $\beta\in[0,1],t\in\T$ and
$x\in\R$, where applicable).
\begin{itemize}
\item[(g1)] $g_0$ has two fixed points in the interval $[\gamma^-,\gamma^+]$,
  whereas $g_1$ has none;
\item[(g2)] $g^t_\beta(\gamma^\pm)\leq \gamma^\pm$;
\item[(g3)] $\partial_x g^t_\beta(x) > 0$; 
\item[(g4)] the mapping $(\beta,t,x)\mapsto g^t_\beta(x)$ is continuous;
\item[(g5)] the mapping $\beta\mapsto g^t_\beta(x)$ is differentiable and
  $\partial_\beta g^t_\beta(x)<0$;
\item[(g6)] $\partial^2_x g^t_\beta(x)<0$ for all $x\in[\gamma^-,\gamma^+]$ {\em (concavity)}.
\end{itemize}
We call such a family $(g_\beta)_{\beta\in[0,1]}$ an {\em (autonomous) reference
  family}.

\begin{remark}\alphlist
\item
  Properties (g1)--(g6) imply that the family $(g_\beta)_{\beta\in[0,1]}$
  undergoes a fold bifurcation in the interval $[\gamma^-,\gamma^+]$: Due to the
  concavity in (g6), $g_\beta$ can have at most two fixed points in this region,
  with the upper one attracting and the lower one repelling. By (g1), the map
  $g_0$ has two such fixed points. Due to the monotone dependence on the
  parameter assumed in (g5), these two fixed points have to move towards each
  other as $\beta$ is increased. They have to vanish before $\beta=1$, as $g_1$
  has no fixed points, and the only possibility to do so is to collide at a
  unique bifurcation parameter $\beta_c$.
\item If $\T=\Z$ (that is, in the discrete time case), the simplest way to obtain a reference family of this kind
 is to fix some strictly increasing map $g:\R\to\R$ such that $g$
  maps the points $\gamma^\pm$ below themselves, is strictly concave on
  $[\gamma^-,\gamma^+]$ and has two fixed points in $[\gamma^-,\gamma^+]$, but
  $g-1$ does not have any fixed points in this interval. Then $g_\beta=g-\beta$
  satisfies the above properties. \listend
\end{remark}

Our main result of this section now states that under some mild conditions, any
random perturbation of such a reference family will undergo a non-smooth fold
bifurcation.

\begin{theorem} \label{t.random_nonsmoothness}
  Suppose that $(\Theta,\cB,\nu,\omega)$ is an mpds and
  $(\Xi_\beta)_{\beta\in[0,1]}$ is a parameter family of $\omega$-forced monotone flows
  that satisfies the assumptions of Theorem~\ref{random} with constant curves
  $\gamma^\pm$. Further, assume that $(g_\beta)_{\beta\in [0,1]}$ is an
  autonomous reference family such that the following conditions hold for all $\beta \in [0,1]$.
  \begin{itemize}
  \item[(i)] For all $x\in X$, $t>0$ and $\nu$-almost all $\theta\in\Theta$ we have
    $g_\beta^t(x)\leq \xi_\beta^t(\theta,x)$. {\em (Lower bound)}
  \item[(ii)] For all $\eps,T>0$ there exists a set $A_{\eps,T}\ssq\Theta$ of positive measure
    $\nu(A_{\eps,T})>0$ such that 
    \begin{equation}
      |\xi^t_\beta(\theta,x)-g_\beta(x)| \ \leq \ \eps \eqand |\partial_x
      \xi^t_\beta(\theta,x)-\partial_x g^t_\beta(x)| \ \leq  \ \eps 
    \end{equation}
    for all $\theta\in A_{\eps,T},\ |t|\leq T$ and $x\in[\gamma^-,\gamma^+]$. {\em (Shadowing)}
  \item[(iii)] For $\nu$-almost every $\theta\in\Theta$ there exists $t\in\T$ and
    $\delta>0$ such that $\xi_\beta^t(\omega^{-t}(\theta),x)\geq g_\beta^t(x)+\delta$ for
    all $x\in[\gamma^-,\gamma^+]$.  {\em (Separation)}
  \end{itemize}
Then $(\Xi_\beta)_{\beta\in[0,1]}$ undergoes a non-smooth fold bifurcation and
the bifurcation parameter $\beta_c$ is the same as in the reference family
$(g_\beta)_{\beta\in[0,1]}$.
\end{theorem}

\begin{remark}
  Note that by construction the forced Allee model (\ref{e.forced_Allee}) with
  random forcing term~(\ref{e.random_forcing_term_example}) satisfies the
  assumptions of Theorem~\ref{t.random_nonsmoothness}, with the unforced Allee
  model as a reference family. This then implies
  Theorem~\ref{t.random_nonsmooth}.
\end{remark}
In order to prove Theorem~\ref{t.random_nonsmoothness}, we first provide the
following auxiliary statement about the equivalence of the existence of
invariant graphs and the existence of orbits that remain in the region
$\Gamma=\Theta\times[-\gamma,\gamma]$ at all times.

\begin{lem}\label{lem01}
Suppose $\Xi$ is a monotone skew product flow of the form
(\ref{e.skew_product_flow}) with an mpds $(\Omega,\cB,\nu,\omega)$ in the
base. Further, assume that there exist measurable curves
$\gamma^-\leq\gamma^+:\Omega\to X$ that satisfy
\begin{equation}\label{mapbelow}
  \xi^t(\theta,\gamma^\pm(\theta)) \ \leq \ \gamma^\pm(\omega(\theta))
\end{equation}
for $\nu$-almost every $\theta\in\Theta$ and all $t\geq 0$. Let
\begin{equation}\label{e.boundary_condition}
  \Gamma \ = \ \{(\theta,x)\mid \theta\in\Theta,\ \gamma^-(\theta)\leq x \leq
  \gamma^+(\theta) \} \ .
\end{equation}
Then there exists a $(\Xi,\nu)$-invariant graph $\varphi$ in $\Gamma$ if and
only if
\begin{equation} \label{e.lowerboundedness}
   \xi^t(\theta,\gamma^+(\theta)) \ \geq \ \gamma^-(\omega^t(\theta))
\end{equation}
holds for all $t\geq 0$ and $\nu$-almost every $\theta\in\Omega$.
\end{lem}

\begin{proof} An important ingredient for the proof are the {\em graph transforms} $\Xi^t_*\gamma$ of
a measurable function $\gamma:\Theta\to X$.
For $t\in\T$, these are defined by
\begin{equation}\label{grtrans}
\Xi^t_{*}\gamma(\theta) \ =
\ \xi^t(\omega^{-t}(\theta),\gamma(\omega^{-t}(\theta))) \ .
\end{equation}
If $t\geq 0$, we speak of a {\em forwards transform} and if $t\leq 0$, of
a {\em backwards transform}. 
 We define
\begin{equation}\label{sequences}
 \gamma^{+}_{t}=\Xi^t_{*}\gamma^{+}~\textrm{and}~\gamma^{-}_{t}=\Xi^{-t}_{*}\gamma^{-}.
\end{equation}
Then (\ref{mapbelow}) together with the monotonicity of the fibre maps implies
that the family of functions $ \gamma^{+}_{t}$ is decreasing in $t$. Similarly,
the family $\gamma^{-}_{t}$ is increasing (note here that
$\xi^t(\theta,\gamma^-(\theta))\leq \gamma^-(\theta)$ for $t>0$ implies
$\xi^{t}(\theta,\gamma^-(\theta))\geq\gamma^-(\theta)$ for $t<0$).

 Suppose now that (\ref{e.lowerboundedness}) holds for all $t>0$ and
 $\nu$-almost every $\theta\in\Theta$. Then $\gamma^+_t$ is bounded below by
 $\gamma^-$ for all $t>0$ and thus converges $\nu$-almost everywhere to a
 function
\begin{equation}
  \varphi^{+}(\theta) \ = \ \lim_{t\to\infty}\gamma^+_t(\theta) \ .
\end{equation}
Due to the continuity of the fibre maps, we have that
\begin{equation}\label{invariance}
  \begin{split}
    \xi^s(\theta,\varphi^{+}(\theta)) & = \ \xi^s\left(\theta,\lim_{t\to\infty} \gamma^{+}_{t}(\theta)\right)
 \  = \ \lim_{t\to\infty} \xi^s(\theta,\gamma^{+}_{t}(\theta))  \\ & = \ \lim_{t\to\infty}
  \gamma^{+}_{t+s}(\omega^s(\theta)) \ = \ \varphi^{+}(\omega^s(\theta))
  \end{split}
\end{equation}
$\nu$-almost everywhere.
Hence, $\varphi^+$ is the desired invariant graph.

 Conversely, assume that there exists an invariant graph $\varphi$ in $\Gamma$.
 Then the monotonicity of the fibre maps gives
\begin{equation}
 \xi^t(\theta,\gamma^{+}(\theta))\geq \xi^t(\theta,\varphi(\theta)) \ =
 \ \varphi(\omega^t(\theta))\ \geq \ \gamma^-(\omega^t(\theta))
\end{equation}
for $\nu$-almost every $\theta\in\Theta$. 
\end{proof}

\begin{remark}\label{r.bounding_graphs}
  As we have seen in the above proof, if $\Xi$ satisfies the assumptions of
  Lemma~\ref{lem01} and has at least one invariant graph, then the formula
\begin{equation}
  \label{e.bounding_graph}
\varphi^{+}(\theta) \ = \ \lim_{t\to\infty}\gamma^+_t(\theta) \ =
\ \lim_{t\to\infty} \xi^t(\omega^{-t}(\theta),\gamma^+(\omega^{-t}(\theta)))
\end{equation}
yields one such graph. This way of defining an invariant graph is called {\em
  pullback construction} and generally works if the graph is an attractor.  In a
similar fashion, it is possible to show that the {\em pushforward construction}
\begin{equation}
  \label{e.lower_bounding_graph}
\varphi^{-}(\theta) \ = \ \lim_{t\to-\infty}\gamma^-_t(\theta) \ =
\ \lim_{t\to\infty} \xi^{-t}(\omega^{t}(\theta),\gamma^-(\omega^{t}(\theta)))
\end{equation}
also defines an invariant graph, which usually is a repeller.  However, the
graphs $\varphi^-$ and $\varphi^+$ may coincide, as in the case of a smooth fold
bifurcation (see Theorem~\ref{random}).
\end{remark}

We can now turn to the 

\begin{proof}[Proof of Theorem~\ref{t.random_nonsmoothness}.] Let $\beta_c$ be the bifurcation
parameter for the family $(\Xi_\beta)_{\beta\in[0,1]}$ and $\tilde\beta_c$ the
one for the reference family $(g_\beta)_{\beta\in[0,1]}$. 
We first show that $\beta_c=\tilde\beta_c$.
Denote the unique fixed point of 
$g_{\tilde\beta_c}$ in $[\gamma^-,\gamma^+]$ by $x_0$. 
Then, for
$\nu$-almost all $\theta\in\Theta$ and all $\beta<\tilde \beta_c$, we obtain
\[
\xi^t_\beta(\theta,\gamma^+) \ \geq \ \xi^t_{\tilde \beta_c}(\theta,\gamma^+)
\ \geq \ g^t_{\tilde \beta_c}(\gamma^+) \ \geq \ g^t_{\tilde\beta_c}(x_0) \ =
\ x_0 \ \geq \ \gamma^- \ .
\]
Hence, Lemma~\ref{lem01} implies that $\Xi_{\beta}$ has at least one invariant
graph in $\Gamma$ for all $\beta\leq\tilde\beta_c$, and thus $\beta_c\geq
\tilde\beta_c$.

Conversely, suppose that $\beta>\tilde\beta_c$. As $g_\beta$ has no fixed points
in $[\gamma^-,\gamma^+]$ and $g^t_\beta(\gamma^+)<\gamma^+$ for all $t\geq 0$,
we obtain that $g_\beta^T(\gamma^+)<\gamma^-$ for some $T>0$. Let $\eps>0$ be
such that $g_\beta^T(\gamma^+)<\gamma^--\eps$.  By assumption, the set
$A_{\eps,T}$ in the statement of the theorem has positive measure. For any
$\theta\in A_{\eps,T}$, we obtain
\[
\xi_\beta^T(\theta,\gamma^+) \ \leq \ g^T_\beta(\gamma^+)+\eps \  < \ \gamma^-
\ .
\]
Due to Lemma~\ref{lem01}, this excludes the existence of an invariant graph in
$\Gamma$ for $\beta>\beta_c$.  We therefore obtain $\beta_c\leq
\tilde\beta_c$. Together with the above, this yields $\beta_c=\tilde\beta_c$.

It remains to show the non-smoothness of the bifurcation. 
To that end, we set
\[
\varphi^+(\theta) \ = \ \lim_{t\to\infty} \xi_{\beta_c}^t(\omega^{-t}(\theta),\gamma^+)
\quad \text{ and } \quad 
\varphi^-(\theta) \ = \ \lim_{t\to\infty} \xi_{\beta_c}^{-t}(\omega^t(\theta),\gamma^-) \ . 
\]
As discussed in Remark~\ref{r.bounding_graphs}, $\varphi^+$ and $\varphi^-$ are
well-defined invariant graphs.  To finish the proof, we have to show that
$\varphi^-(\theta)<\varphi^+(\theta)$ $\nu$-almost surely.

To that end, observe that condition (i) in Theorem~\ref{t.random_nonsmoothness}
together with the monotonicity of the fibre maps implies
$\xi^t_{\beta_c}(\theta,x)\leq g^t_{\beta_c}(x)$ for all
$x\in[\gamma^-,\gamma^+]$, all $t<0$ and $\nu$-almost all $\theta\in \Theta$.
As a consequence, we have $\varphi^-(\theta) \leq x_0$ almost surely.  Now, by
assumption (iii) we have that for $\nu$-almost every $\theta\in\Theta$ there
exists $\delta>0$ and $s\in\T$ such that
$\xi^s_{\beta_c}(\omega^{-s}(\theta),x_0)>g^s_{\beta_c}(x_0)+\delta=x_0+\delta$. Thus,
we obtain
\begin{eqnarray*}
\varphi^+(\theta) & = & \lim_{t\to\infty} \xi^t_{\beta_c}(\omega^{-t}(\theta),\gamma^+)
\\ & = & \lim_{t\to\infty}
\xi^s_{\beta_c}(\omega^{-s}(\theta),\xi^{t-s}_{\beta_c}(\omega^{-t}(\theta),\gamma^+))
\geq\lim_{t\to\infty}
\xi^s_{\beta_c}(\omega^{-s}(\theta),g^{t-s}_{\beta_c}(\gamma^+))
\\ & \geq &
\xi^s_{\beta_c}(\omega^{-s}(\theta),x_0) \ \geq \ x_0+\delta \
\end{eqnarray*}
and hence, in particular, $\varphi^+(\theta)>x_0\geq \varphi^-(\theta)$. This
finishes the proof.
\end{proof}

\section{Lyapunov exponents in nonsmooth fold bifurcations}
\label{LyapunovExponents}

\subsection{Lyapunov gap in nonsmooth fold bifurcations} The aim of this section is to provide a 
proof of Theorem~\ref{t.lyapunov_gap}, which we restate here in a more general
form.
\begin{theorem} \label{t.lyapunov_gap2} Suppose that $(\Xi_\beta)_{\beta\in[0,1]}$ is a parameter family of forced
 monotone $\mathcal C^2$-flows that satisfies the assumptions of Theorem~\ref{qpf} (for
 deterministic forcing) or Theorem~\ref{random} (for random forcing). Further,
 assume that the fold bifurcation that occurs in this family at the critical
 parameter $\beta_c$ is non-smooth. Then
   \begin{equation}
     \label{e.lyapunov_gap2} \lim_{\beta\nearrow\beta_c} \lambda(\varphi^+_\beta) \ = \ \lambda(\varphi^+_{\beta_c}) \ < \  0 \  . 
  \end{equation}
  If the fold bifurcation is smooth, then $\lim_{\beta\nearrow\beta_c}
  \lambda(\varphi^+_\beta) = 0$.  The analogous results hold for the unstable
  equilibrium $\varphi^-_\beta$.
\end{theorem}
\proof We only consider the deterministic case, as the random case can be dealt
with in a completely analogous way. Let $(\Xi_\beta)_{\beta\in[0,1]}$ satisfy
the assumptions of Theorem~\ref{qpf}.  We claim that for each $\theta\in \Theta$
we have that $\varphi^+_{\beta_c}(\theta)$ coincides with
\[
\varphi(\theta) \ = \ \lim_{\beta\nearrow\beta_c} \varphi^+_\beta(\theta) \ .
\]
Note that $\varphi$ is well-defined due to the monotone dependence of
$\xi_\beta^t$ on $\beta$ (see assumption (v) in Theorem~\ref{qpf}) which results
in $\varphi_\beta\geq \varphi_{\beta'}$ whenever $\beta<\beta'\leq \beta_c$.

In order to see that $\varphi$ is an invariant graph, fix $\theta\in \Theta$,
$t>0$ and $\eps>0$. Choose $\delta>0$ such that
$|\xi^t_{\beta_c}(\varphi(\theta)) - \xi^t_{\beta_c}(x)|<\eps$ for all $x\in
B_\delta(\varphi(\theta))$ and at the same time
$|\xi_\beta^t(x)-\xi_{\beta'}^t(x)|<\eps$ whenever $|\beta-\beta'|<\delta$. Note
that such $\delta$ exists due to the uniform continuity of $(\beta,x)\mapsto
\xi^t_\beta(\theta,x)$.  Let $\beta<\beta_c$ be such that
$\beta_c-\beta<\delta$, $\varphi^+_\beta(\theta)-\varphi(\theta)<\delta$ and
$\varphi^+_\beta(\omega^t(\theta))-\varphi(\omega^t(\theta))<\eps$.  We obtain
\begin{eqnarray*}
|\xi^t_{\beta_c}(\varphi(\theta)) - \varphi(\omega^t(\theta))| & \leq &
\ |\xi^t_\beta(\varphi(\theta))-\varphi^+_\beta(\omega^t(\theta))| + 2\eps \\ &
= & \ |\xi^t_\beta(\varphi(\theta))-\xi^t_\beta(\varphi^+_\beta(\theta))| +
2\eps \ \leq \ 3\eps \ .
\end{eqnarray*}
As $\eps>0$ was arbitrary, this proves
$\xi^t_{\beta_c}(\varphi(\theta))=\varphi(\omega^t(\theta))$ and hence the
invariance of $\varphi$ under $\Xi_{\beta_c}$.  Now, since the graphs
$\varphi^+_{\beta_c}$ are monotonically decreasing in $\beta$, we have
$\varphi\geq \varphi^+_{\beta_c}$. As there is no $\Xi_{\beta_c}$-invariant
graph above $\varphi^+_{\beta_c} $ in the considered region $\Gamma$, we obtain
$\varphi^+_{\beta_c} = \varphi$. Using dominated convergence, this proves the
statement about the Lyapunov exponents.  \qed\medskip

\subsection{Slope at the bifurcation point}
\label{Slope}

Although this is not in our main focus, we want to comment in this section on a
particular qualitative difference between non-smooth fold bifurcations in the
quasiperiodically forced and the randomly forced case. As it can be seen from
Figure~\ref{f.LE_smooth/nonsmooth}(b)-(d), the slope of the Lyapunov exponent
of the attractor increases strongly towards the bifurcation in the
quasiperiodically forced case, whereas it only increases slightly
or even remains constant in the random case.

In fact, the heuristic description of non-smooth fold bifurcations in qpf systems
given in Section~\ref{DiscreteTimeModel} suggests that $\partial_\beta\lambda(\varphi^+_\beta)$ should
actually increase to infinity as $\beta\nearrow\beta_c^+$. The reason is that
due to the concavity of the right side of the vector field, the Lyapunov
exponent increases whenever the graph decreases. Thereby, the quantitative
contribution of each peak that develops should be the product between its width
and its speed, which is more or less constant since both decrease, respectively
increase, with the same exponential rate. Hence, every peak contributes a similar
amount to the slope of the Lyapunov exponent, and as there are infinitely many
peaks, this slope grows to infinity as the bifurcation is approached. In
principle, we believe that this heuristic explanation could be made precise by
using the machinery for the proof of non-smooth fold bifurcations in
\cite{fuhrmann2013NonsmoothSaddleNodesI,Fuhrmann2016SNAinFlows} which, however, goes
beyond our current scope.

For the case of random forcing, we provide a proof for the boundedness of the
slope of the Lyapunov exponent of $\varphi^+_\beta$ as $\beta\nearrow \beta_c$.
In order to avoid too many technicalities and to not obstruct the view on the underlying mechanism,
we restrict to the case of the
discrete-time example~(\ref{e.discrete_time_examples}). We note, however, that
the proof can be generalised to broader classes of monotone skew product maps
and, with some more work required, to continuous-time systems.

\begin{theorem}\label{nsmthrand} Suppose
$(f_\beta)_{\beta\in[0,1]}$ is the family of skew product maps given by
  (\ref{e.discrete_time_examples}) with $\Theta=\{0,1\}^\Z$ the Bernoulli space equipped with the shift map $\sigma$
  and the Bernoulli measure $\mu$. 
  Let
  $\beta_c=\arctan(\sqrt{\alpha -1})-\sqrt{\alpha-1}/\alpha-\kappa$. Then
  $(f_{\beta})_{\beta\in [0,1]}$ satisfies the hypothesis of
  Theorem~\ref{random} (with $\gamma^-=0$ and $\gamma^+=\pi/2$) and undergoes a
  non-smooth saddle-node bifurcation with critical parameter $\beta_c$. 
  Moreover, there exists a constant  $C>0$ such
  that 
  \[
   |\partial_\beta \lambda_\mu(\varphi_\beta^\pm)|\ \leq \ C
  \]
 for all $\beta\in [0,\beta_c]$.
\end{theorem}
\noindent Before we turn to the proof of the Theorem \ref{nsmthrand}, we first
need the following preliminary result.
\begin{lem}\label{lem: uniform convergence of iterated boundary lines}
In the situation of Theorem~\ref{nsmthrand} we let
$\gamma^\pm_{n,\beta}=f^{\pm n}_{\beta*}\gamma^{\pm}$, where the graph transform
$f^n_{\beta*}$ is defined as in the proof of Theorem~\ref{lem01}. Then
$\gamma^\pm_{n,\beta}(\theta)$ converges to $\varphi_\beta^\pm(\theta)$
uniformly in $\beta$ and $\theta$ (with $\beta\in [0,\beta_c]$ and $\theta\in
\Sigma$) as $n\to\infty$.

Moreover, for all $\theta\in \Sigma$, the map 
$[0,\beta_c) \ni \beta\mapsto \varphi_{\beta}^\pm(\theta)$ is differentiable
and for every $\beta'\in[0,\beta_c)$, we have that
\begin{align}\label{eq: uniform convergence of derivatives}
\partial_\beta\gamma^\pm_{n,\beta}(\theta)\ = \ -\sum_{i=1}^n\, 
\prod_{\ell=1}^{i-1} \partial_x{f^{\pm1}}_{\beta,\sigma^{-\ell}\theta}(f_{\beta,\sigma^{-n}\theta}^{\pm (n-\ell)}(\gamma^+)) \ 
\stackrel{n\to\infty}{\longrightarrow} \ \partial_\beta\varphi_\beta^\pm(\theta)\end{align}
uniformly in $\beta$ and $\theta$ for all $\beta\in[0,\beta']$ and all
$\theta\in \Sigma$.
\end{lem}
\begin{proof}
 We only consider $\gamma^+_{n,\beta}$ and $\varphi_\beta^+$, the statements for
 $\gamma^-_{n,\beta}$ and $\varphi_\beta^-$ follow analogously.  Recall that
\[
  \textstyle f_\beta : \Sigma\times \R \to\Sigma\times\R \quad , \quad
  (\theta,x)\mapsto \left(\omega(\theta),g(x)- \kappa
  \cdot \theta_0 - \beta\right) \ ,
\]
with $g(x)=\arctan(\alpha x)$.

First, observe that
$\partial_x^2 f^{n}_{\beta,\theta}(x)<0$ for all $n\in\N$ as long as
$x,f_{\beta,\theta}(x)\ld f_{\beta,\omega^{n-1}(\theta)}(x)>0$
since the composition of concave increasing functions is again concave.  
Second, note that with
$\theta^*=\ldots 1,1,1\ldots\in \Sigma$, we have
 \begin{align}\label{eq: monotonicity}
 f_{\beta_c,\theta^*}(x)= f_{\beta,\theta}(x)-(\beta_c-\beta)-\kappa(1-\theta_0)
 \end{align}
for all $\beta\in [0,\beta_c]$ and all $\theta\in \Sigma, x\in X$.

 Hence, we obtain that for all $\beta\in[0,\beta_c],\ \theta\in\Sigma$ and all $n,n'\in \N$ with 
 $n\geq n'$ we have
 \begin{align*}
 | \gamma_{n',\beta}^+(\theta)-\gamma_{n,\beta}^+(\theta)| &\ =
 f^{n'}_{\beta,\sigma^{-n'}\theta}(\gamma^+)-
 f^{n'}_{\beta,\sigma^{-n'}\theta}(f^{n-n'}_{\beta,\sigma^{-n}\theta}(\gamma^+)) \\ & \leq\ 
 f^{n'}_{\beta,\sigma^{-n'}\theta}(\gamma^+)-
 f^{n'}_{\beta,\sigma^{-n'}\theta}(f^{n-n'}_{\beta_c,\theta^*}(\gamma^+))\\ &=\ f^{n'-1}_{\beta,\sigma^{1-n'}\theta}(f_{\beta,\sigma^{-n'}\theta}(\gamma^+))-
 f^{n'-1}_{\beta,\sigma^{1-n'}\theta}(f_{\beta,\sigma^{-n'}\theta}(f^{n-n'}_{\beta_c,\theta^*}(\gamma^+)))
 \\ &\leq f^{n'-1}_{\beta,\sigma^{1-n'}\theta}(f_{\beta_c,\theta^*}(\gamma^+))-
 f^{n'-1}_{\beta,\sigma^{1-n'}\theta}(f_{\beta_c,\theta^*}(f^{n-n'}_{\beta_c,\theta^*}(\gamma^+)))\\ &\leq\ 
 f^{n'-2}_{\beta,\sigma^{2-n'}\theta}(f^2_{\beta_c,\theta^*}(\gamma^+))-
 f^{n'-2}_{\beta,\sigma^{2-n'}\theta}(f^2_{\beta_c,\theta^*}(f^{n-n'}_{\beta_c,\theta^*}(\gamma^+)))\\ &\leq\ldots\leq\ 
 f^{n'}_{\beta_c,\theta^*}(\gamma^+)-
 f^{n'}_{\beta_c,\theta^*}(f^{n-n'}_{\beta_c,\theta^*}(\gamma^+)) \\ & = \  |
 \gamma_{n',\beta_c}^+(\theta^*)-\gamma_{n,\beta_c}^+(\theta^*)|,
 \end{align*}
 where we used the monotonicity of the fibre maps in the first inequality and the 
 above mentioned concavity together with \eqref{eq:
   monotonicity} in the steps to the fourth, fifth and sixth line.  This proves the first
 part.
 
 Next, we show that $\partial_\beta\gamma^+_{n,\beta}(\theta)$ converges
 uniformly in $\theta$ and $\beta$ which immediately implies the second part.
 To that end, we first provide a uniform upper bound on $$\sup\limits_{\theta\in
   \Sigma,\beta\in[0,\beta']}\partial_x
 f_{\beta,\theta}(\gamma_{n,\beta}^+(\theta))$$ for all $n\in \N$.
 Similarly as
 in the proof of Theorem~\ref{t.random_nonsmoothness}, we see that
 $$x_{\min}(\beta)\ := \ \min_{\theta\in
   \Sigma}\varphi_{\beta}^+(\theta)\geq x_0(\beta')$$ for all $\beta\in[0,\beta']$,
 where $x_0(\beta')$ is the upper fixed point of the map $g-\kappa-\beta'$.
 Hence, we have
 \begin{align}\label{eq: defn alpha}
 \begin{split}
 0 & \leq\ \sup_{\theta\in
   \Sigma,\beta\in[0,\beta']}\partial_x f_{\beta,\theta}(\gamma_{n,\beta}^+(\theta))
 =\ \sup_{\theta\in
   \Sigma,\beta\in[0,\beta']}g'(\gamma_{n,\beta}^+(\theta))\ \leq
 \ g'(\varphi^+_\beta(\theta)) 
 \\ &  
\ \leq\ g'(x_{0}(\beta'))\ =: \  c \ < \ 1 \ ,
 \end{split}
 \end{align}
where we used the concavity of $g$ and the monotone dependence of
$\varphi_\beta^+(\theta)$ on $\beta$.

Now, observe that
\begin{align*}
 \partial_\beta\gamma^+_{n,\beta}(\theta) &=\ \partial_\beta
 f^n_{\beta,\sigma^{-n}\theta}(\gamma^+) \\ & = \ \partial_\beta
 f_{\beta,\sigma^{-1}\theta}( f^{n-1}_{\beta,\sigma^{-n}\theta}(\gamma^+))+ \partial_x
 f_{\beta,\sigma^{-1}\theta}( f^{n-1}_{\beta,\sigma^{-n}\theta}(\gamma^+)) \cdot
 \partial_\beta f^{n-1}_{\beta,\sigma^{-n}\theta}(\gamma^+)\\ &= \ \ldots\  =\ 
 \sum_{i=1}^n\partial_\beta f_{\beta,\sigma^{-i}\theta}(
 f^{n-i}_{\beta,\sigma^{-n}\theta}(\gamma^+)) \prod_{\ell=1}^{i-1}
 \partial_x f_{\beta,\sigma^{-\ell}\theta}(f_{\beta,\sigma^{-n}\theta}^{n-\ell}(\gamma^+))\\ &=-\sum_{i=1}^n\,
 \prod_{\ell=1}^{i-1}
 \partial_x f_{\beta,\sigma^{-\ell}\theta}(f_{\beta,\sigma^{-n}\theta}^{n-\ell}(\gamma^+)).
\end{align*}
Together with \eqref{eq: defn alpha}, we hence obtain for $n\geq n'$
\begin{align*}
 |\partial_\beta\gamma^+_{n',\beta}(\theta)-\partial_\beta\gamma^+_{n,\beta}(\theta)|=
 \sum_{i=n'+1}^n\, \prod_{\ell=1}^{i-1} \partial_x
 f_{\beta,\sigma^{-\ell}\theta}(f_{\beta,\sigma^{-n}\theta}^{n-\ell}(\gamma^+))\leq
 \sum_{i=n'+1}^n c^{i-1},
\end{align*}
which proves the statement.
\end{proof}

We can now turn to the 
\begin{proof}[Proof of Theorem~\ref{nsmthrand}] 
We keep the notation as in the previous proof.
Observe that $(g-\kappa-\beta)_{\beta\in[0,1]}$ is an autonomous reference family for 
$(f_\beta)_{\beta\in[0,1]}$.
Therefore, the fact that $(f_\beta)_{\beta\in[0,1]}$ 
undergoes a non-smooth saddle-node bifurcation with critical parameter $\beta_c$
(given by the bifurcation parameter of the family $(g-\kappa-\beta)_{\beta}$)
is a direct consequence of Theorem~\ref{t.random_nonsmoothness}. 
Hence, it remains to prove the existence
of a uniform bound on the slope of the Lyapunov exponent. As before, we only
consider $\varphi_\beta^+$. Further, we show the statement for $\beta\in
[0,\beta_c)$ which immediately yields the full statement by means of the mean
  value theorem. \smallskip

 Given $\beta \in [0,\beta_c)$, observe that
 \begin{align*}
  \partial_\beta \lambda_\mu(\varphi_\beta^+) & = \ \partial_\beta
  \int_{\Sigma}\!  \log \partial_x
  f_{\beta,\theta}(\varphi^+_\beta(\theta))\,d\theta \\ &  =\ \partial_\beta
  \int_{\Sigma}\! \log g'(\varphi^+_\beta(\theta))\,d\theta= \int_{\Sigma}\!
  \frac{g''(\varphi^+_\beta(\theta))}{g'(\varphi^+_\beta(\theta))} \cdot
  \partial_\beta \varphi^+_\beta(\theta) \,d\theta.
 \end{align*}
 With $c\geq \sup_{x\in[0,\pi/2]} |g''(x)/g'(x)|$, we hence obtain
 \begin{align*}
  |\partial_\beta \lambda_\mu(\varphi_\beta^+)|\leq 
   c\int_{\Sigma}\! |\partial_\beta \varphi^+_\beta(\theta)| \,d\theta.
 \end{align*}
 Now, let
 \[
 \alpha\ = \ \sup_{\theta\in \Sigma,\, \theta_0=0}f_{\beta_c,\sigma \theta}'(f_{\beta_c,\theta}(x_c))
 \]
 and note that
 \[
\alpha \ = \ g'(g(x_{c})-\beta_c) \ = \ g'(g(x_{c})-\kappa-\beta_c+\kappa) \ = \ g'(x_c+\kappa) \ <
\ g'(x_{c}) \ = \ 1 \ ,
\]
where $x_c$ is the neutral fixed point of the map $g - \kappa- \beta_c$. Then
 \begin{align*}
\int_{\Sigma}\! |\partial_\beta \varphi^+_\beta(\theta)|
\,d\theta&=\int_{\Sigma}\! |\partial_\beta \lim_{n\to\infty}
\gamma^+_{n,\beta}(\theta)| \,d\theta= \lim_{n\to\infty} \int_{\Sigma}\!
|\partial_\beta \gamma^+_{n,\beta}(\theta)| \,d\theta\\ &\stackrel{\eqref{eq:
    uniform convergence of derivatives}}{=} \lim_{n\to\infty} \int_{\Sigma}\,
\sum_{i=1}^n\, \prod_{\ell=1}^{i-1} \partial_x
f_{\beta,\sigma^{-\ell}\theta}(f_{\beta,\sigma^{-n}\theta}^{n-\ell}(\gamma^+))
\,d\theta\\ &= \ \lim_{n\to\infty}\sum_{i=1}^n\, \int_{\Sigma}\,
\prod_{\ell=1}^{i-1}
\partial_x f_{\beta,\sigma^{-\ell}\theta}(f_{\beta,\sigma^{-n}\theta}^{n-\ell}(\gamma^+))
\,d\theta \\ & \leq \ \lim_{n\to\infty}\sum_{i=1}^n\, \int_{\Sigma}\,
\prod_{\ell=1}^{i-1}
\partial_x f_{\beta,\sigma^{-\ell}\theta}(f_{\beta,\sigma^{-\ell-1}\theta}(x_c))
\,d\theta\\ &\leq\ \lim_{n\to\infty}\sum_{i=1}^n\,\sum_{k=0}^{i-1} \alpha^k \cdot
\mu\left( \left\{\theta\in \Sigma\colon k =  \#\{1<\ell\leq
i\colon\theta_{-\ell}=0\}\right\}\right)\\
&=\ \lim_{n\to\infty}\sum_{i=1}^n\,\sum_{k=0}^{i-1}\alpha^k\cdot \binom{i-1}{k}
(1/2)^{i-1}\\ &  =\ \lim_{n\to\infty}\sum_{i=1}^n\,\sum_{k=0}^{i-1}\binom{i-1}{k}
(\alpha/2)^k (1/2)^{i-1-k}\\ &=\ \lim_{n\to\infty}\sum_{i=1}^n\,
(\alpha/2+1/2)^{i-1}<\infty \ . 
 \end{align*}
Since $\alpha$ is independent of $\beta$, the statement follows.
\end{proof}

\section{Range of finite-time Lyapunov exponents} \label{Range}

In this section, we provide a proof of
Theorem~\ref{t.finite_time_exponent_range}, which we restate below in a more
general form.  To that end, let us introduce the maximal finite-time Lyapunov
exponents on the attractor. As invariant graphs only need to be defined almost
surely, we just take into account exponents that can be `seen' on a set of
positive measure by setting
\begin{equation*}\label{lambdamax}
 \lambda_{k}^{\max}(\varphi^{+}_{\beta})=\sup\left\{\lambda\in\mathbb{R}\left|\mu_{\varphi^{+}_{\beta}}
 (\{(\theta,x)|\lambda_{k}(f_{\beta},\theta,x)\geq\lambda\})>0\right.\right\}.
\end{equation*}
Here, the graph measure $\mu_{\varphi^+_\beta}$ is as discussed in  Section~\ref{1}.
Note that if the forcing is quasiperiodic, then the attractors prior to
the bifurcation are all continuous so that we actually have
$\lambda_{k}^{\max}(\varphi^{+}_{\beta}) =
\max\left\{\lambda_{k}(\theta,\varphi^{+}_{\beta}(\theta))\left|\theta\in\T^{1}\right.\right\}$
whenever $\beta<\beta_c$.

We first consider the case of quasiperiodic forcing,
where the general statement we aim at reads as follows. 
\begin{theorem}\label{qpffinite_time}
Suppose $(\Xi_{\beta})_{\beta\in[0,1]}$ is a parameter family of qpf monotone
flows that satisfies the hypothesis of Theorem~\ref{qpf}. Then for all $k\in\N$
we have
\begin{equation} \label{e.finite_time_lyap_qpf}
\varliminf_{\beta\nearrow\beta_{c}}\lambda_{k}^{\max}(\varphi^{+}_{\beta})\ \geq \ \lambda(\varphi^{-}_{\beta_{c}}) \ .
\end{equation}
\end{theorem}
Before we turn to the proof, however, we have to address some subtleties
concerning the topology of pinched invariant graphs in this setting.  Suppose we
observe a non-smooth bifurcation as characterised in Theorem~\ref{qpf}, so that
there exist exactly two graphs $\varphi^-_{\beta_c}<\varphi^+_{\beta_c}$ (up to
modifications on sets of measure zero), where $\varphi^-_{\beta_c}$ is lower and
$\varphi^+_{\beta_c}$ is upper semicontinuous. Let
$A^+=\supp(\mu_{\varphi^+_{\beta_c}})$ and
$A^-=\supp(\mu_{\varphi^-_{\beta^+}})$, where $\supp(\nu)$ denotes the
topological support of a measure $\nu$.\foot{Given a Borel measure $\nu$ on some
  second countable metric space $X$, the support of $\nu$ is defined as
  $\supp(\nu)=\{x\in X\mid \nu(B_\delta(x))>0 \ \forall \delta>0\} = X\smin
  \bigcup_{U \textrm{ open, } \nu(U)=0} U$. It is easy to see that $\supp(\nu)$
  is always closed and can be characterised as the smallest closed set $A\ssq X$
  with $\nu(X\smin A)=0$. Moreover, if $\nu$ is invariant under some continuous
  transformation $f$, then so is $\supp(\nu)$.} Then $A^+$ is
$\Xi_{\beta_c}$-invariant, and consequently the upper and lower bounding graphs
$\varphi^u_{A^+}$ and $\varphi^l_{A^+}$ given by
\[
\varphi^u_{A^+}(\theta) \ = \ \sup A^+_\theta \eqand
\varphi^l_{A^+}(\theta)=\inf A^+_\theta
\]
are $\Xi_\beta$-invariant graphs, with $\varphi^u_{A^+}$ upper and
$\varphi^l_{A^+}(\theta)$ lower semicontinuous (see \cite{stark:2003}).  As
$\varphi^\pm_{\beta_c}$ are the only $\Xi_\beta$-invariant graphs in the
considered region $\Gamma$, we must have
$\varphi^l_{A^+}(\theta)=\varphi^-_{\beta_c}$ and
$\varphi^u_{A^+}=\varphi^+_{\beta_c}$ almost surely (see \cite{stark:2003} for
more details), but the graphs may differ on a set of measure zero (which is the
subtle complication that requires particular care in the proof of
Theorem~\ref{qpffinite_time} below). Neverthless, this implies in particular
that $(\theta,\varphi^-_{\beta_c}(\theta))\in A^+$ almost surely, so that
$\mu_{\varphi^-_{\beta_c}}(A^+)=1$ and hence $A^-\ssq A^+$. As the converse
inclusion follows in the same way, we have $A^-=A^+$. One particular consequence
of this discussion is the following
\begin{lem}
  \label{l.graphclosures}
  For $\mu$-almost every $\theta_0\in\Theta$ and every $\delta>0$ there exists a
  set $B\ssq\Theta$ of positive measure such that
  \[\{(\theta,\varphi^+_{\beta_c}(\theta))\mid \theta\in B\} \ \ssq \
  B_\delta((\theta_0,\varphi^-_{\beta_c}(\theta_0))) \ . 
  \]
\end{lem}
We can now turn to the
\begin{proof}[Proof of Theorem \ref{qpffinite_time}]
Fix $k\in\N$.
First, we claim that there exists a set of positive
measure of $\theta\in\Theta$ such that
\begin{equation}\label{qpffinite_1}
  \lambda_{k}(\theta,\varphi^{-}_{\beta_{c}}(\theta))\ \geq \ \lambda(\varphi^{-}_{\beta_{c}}).
\end{equation}
In order to see this, assume for a contradiction that
$\lambda_{k}(\theta,\varphi^{-}_{\beta_{c}}(\theta))<\lambda(\varphi^{-}_{\beta_{c}})$
for almost all $\theta$.  Clearly, this implies the existence of $\delta>0$ and
of a set $\Theta_\delta\ssq\Theta$ of positive measure such that
\begin{equation*}\label{}
  \lambda_{k}(\theta,\varphi^{-}_{\beta_{c}}(\theta))\ \leq
  \ \lambda(\varphi^{-}_{\beta_{c}}) - \delta
\end{equation*}
for all $\theta \in \Theta_\delta$. 
Due to Birkhoff's Ergodic Theorem (and the ergodicity of the flow $\omega$ on $\Theta$),
the orbit of almost every $\theta\in \Theta$ visits the set $\Theta_\delta$
with positive frequency.
This implies that the pointwise Lyapunov
exponent of almost every $\theta$ satisfies
\begin{equation*}
  \lambda(\theta,\varphi^{-}_{\beta_{c}}(\theta))\ <
  \ \lambda(\varphi^{-}_{\beta_{c}}).
\end{equation*}
This, however, contradicts Birkhoff's Ergodic Theorem according to which we have
\begin{equation*}
  \lambda(\theta,\varphi^{-}_{\beta_{c}}(\theta))\ =
  \ \lambda(\varphi^{-}_{\beta_{c}})
\end{equation*}
almost surely.  Hence, there exists a positive measure set of $\theta$ which
satisfies \eqref{qpffinite_1}.

Now, let $\varepsilon>0$ be given.  Due to Lemma~\ref{l.graphclosures}, there is
$\delta>0$ and $\theta_0\in\Theta$ which satisfies (\ref{qpffinite_1}) and a set
$B\ssq\Theta$ of positive measure such that
$(\theta,\varphi^+_{\beta_c}(\theta))\in
B_\delta(\theta_0,\varphi^-_{\beta_c}(\theta_0))$ for all $\theta\in B$. If
$\delta$ is chosen small enough, then it follows by continuity that
\[
\lambda_k(\theta,\varphi^+_{\beta_c}(\theta)) \ \geq
\ \lambda(\varphi^-_{\beta_c}) -\eps
\]
for all $\theta\in B$. As $\eps>0$ was arbitrary, we obtain that
$\lambda^{\max}_k(\varphi^+_{\beta_c})\geq
\lambda(\varphi^-_{\beta_c})$. Finally, as $\lim_{\beta\nearrow\beta_c}
\varphi^+_\beta(\theta)=\varphi^+_{\beta_c}(\theta)$ almost surely (see the proof of Theorem~\ref{t.lyapunov_gap2}), we obtain
(\ref{e.finite_time_lyap_qpf}) again by continuity. 
\end{proof}
\noindent We now turn to the random case. In this case, we have to restrict to
the setting of Theorem~\ref{t.random_nonsmoothness} (instead of the more general
situation of Theorem~\ref{random}).

\begin{theorem}\label{recovery_rates}
Suppose that $(\Xi_\beta)_{\beta\in[0,1]}$ is a parameter family of randomly
forced monotone flows that satisfies the assumptions  of
Theorem~\ref{t.random_nonsmoothness}. Then for all $k\in\R$
\begin{equation*}
  \lim_{\beta\nearrow\beta_{c}}\lambda_{k}^{\max}(\varphi^{+}_{\beta}(\theta))=0.
\end{equation*}
\end{theorem}
\begin{proof} Let $(g_\beta)_{\beta\in[0,1]}$ be the autonomous reference family
  from Theorem~\ref{t.random_nonsmoothness}. Then we have that $g_{\beta_c}$ has
  a unique fixed point $x_0\in[\gamma^-,\gamma^+]$ whose Lyapunov exponent
  vanishes, that is, $\log \partial_t g^t_{\beta_c}(x_0)=0$ for all $t>0$. By
  continuity, this means that given $t\in\R$ and $\eps>0$ there exists
  $\delta>0$ such that $|x-x_0|<\delta$ and $|\beta-\beta_c|<\delta$ implies
  $|\log(g_\beta^t(x))|/t<\eps$. Moreover, as $\lim_{t\to\infty}
  g^t_{\beta_c}(\gamma^+)=x_0$, we can further require that
  $g^t_\beta(\gamma^+)< x_0+\delta/2$ for all $\beta\in
  [\beta_c-\delta',\beta_c]$ and some $t,\delta'>0$.

  Now, by assumption (ii) of Theorem~\ref{t.random_nonsmoothness} there exists a
  set $A_{\delta/2,t}\ssq\Theta$ of positive measure such that for all
  $\theta\in\omega^t(A_{\delta/2,t})$ we have
  \[
  \xi^t_\beta(\omega^{-t}(\theta),\gamma^+) \ \leq \ x_0+\delta \ .
  \]
  As $\varphi^+_\beta(\theta)$ is the monotone limit of the sequence
  $\xi_\beta^t(\omega^{-t}(\theta),\gamma^+)$ (see the proof of
  Theorem~\ref{t.random_nonsmoothness}) and is bounded below by $x_0$, this
  implies that $x_0\leq \varphi^+_\beta(\theta)\leq x_0+\delta$ and therefore,
  using condition (ii) again, $|\lambda_t(\theta,\varphi^+_\beta(\theta))|<\eps$
  for all $\theta\in A_{\delta/2,k}$ and $\beta\in[\beta_c-\delta',\beta_c]$. As
  $\eps>0$ was arbitrary, this completes the proof.
\end{proof}

\bibliographystyle{alpha}


\end{document}

%% file: bibstandard.tex
\newcommand{\eqand}{\ensuremath{\quad \textrm{and} \quad}}

\newcommand{\foot}{\footnote}

\newcommand{\equi}{\ensuremath{\Leftrightarrow}}

\newcommand{\ld}{\ensuremath{,\ldots,}}
\newcommand{\ssq}{\ensuremath{\subseteq}}
\newcommand{\smin}{\ensuremath{\setminus}}
\newcommand{\eps}{\ensuremath{\varepsilon}}

\newcommand{\supp}{\ensuremath{\mathrm{supp}}}

\newcommand{\Leb}{\ensuremath{\mathrm{Leb}}}

\newcommand{\kreis}{\ensuremath{\mathbb{T}^{1}}}

\newcommand{\alphlist}{\begin{list}{(\alph{enumi})}{\usecounter{enumi}\setlength{\parsep}{2pt}
      \setlength{\itemsep}{1pt} \setlength{\topsep}{5pt}
      \setlength{\partopsep}{3pt}}}
\newcommand{\arablist}{\begin{list}{(\arabic{enumi})}{\usecounter{enumi}\setlength{\parsep}{2pt}
          \setlength{\itemsep}{1pt} \setlength{\topsep}{5pt}
          \setlength{\partopsep}{3pt}}}
\newcommand{\romanlist}{\begin{list}{(\roman{enumi})}{\usecounter{enumi}\setlength{\parsep}{2pt}
              \setlength{\itemsep}{1pt} \setlength{\topsep}{5pt}
              \setlength{\partopsep}{3pt}}}
\newcommand{\Romanlist}{\begin{list}{(\Roman{enumi})}{\usecounter{enumi}\setlength{\parsep}{2pt}
              \setlength{\itemsep}{1pt} \setlength{\topsep}{5pt}
              \setlength{\partopsep}{3pt}}}
\newcommand{\bulletlist}{\begin{list}{$\bullet$}{\setlength{\parsep}{2pt}
                \setlength{\itemsep}{1pt} \setlength{\topsep}{5pt}
                \setlength{\partopsep}{3pt}\setlength{\leftmargin}{15pt}}} 
\newcommand{\Alphlist}{\begin{list}{(\Alph{enumi})}{\usecounter{enumi}\setlength{\parsep}{2pt}
      \setlength{\itemsep}{1pt} \setlength{\topsep}{5pt}
      \setlength{\partopsep}{3pt}}}
 \newcommand{\listend}{\end{list}}

\newcommand{\T}{\ensuremath{\mathbb{T}}}

\newcommand{\N}{\ensuremath{\mathbb{N}}} 
\newcommand{\R}{\ensuremath{\mathbb{R}}}
\newcommand{\Z}{\ensuremath{\mathbb{Z}}}

\newcommand{\Proj}{\ensuremath{\mathbb{P}}}

\newcommand{\cB}{\mathcal{B}}
\newcommand{\cC}{\mathcal{C}}

\newcommand{\cM}{\mathcal{M}}

\newcommand{\cU}{\mathcal{U}}
\newcommand{\cV}{\mathcal{V}}

\newcommand{\halb}{\ensuremath{\frac{1}{2}}}


%% file: ForcedFoldBifurcations.bbl
\begin{thebibliography}{alpha}

\bibitem[AJ02]{AnagnostopoulouJaeger2012SaddleNodes}
V.~Anagnostopoulou and T.~J{\"a}ger.
\newblock Nonautonomous saddle-node bifurcations: random and deterministic
  forcing.
\newblock {\em J.\ Diff.\ Eq.}, 253(2):379--399, 2012.

\bibitem[Arn98]{Arnold1998RandomDynamicalSystems}
L.~Arnold.
\newblock {\em Random Dynamical Systems}.
\newblock Springer, 1998.

\bibitem[Bje05]{bjerkloev:2005}
K.~Bjerkl{\"o}v.
\newblock Dynamics of the quasiperiodic {Schr\"odinger} cocycle at the lowest
  energy in the spectrum.
\newblock {\em Comm. Math. Phys.}, 272:397--442, 2005.


\bibitem[CCP{\etal}11]{carpenter2011early}
S.R. ~Carpenter, J.J.~Cole, M.L.~Pace, ~R.~Batt, ~W.A.~Brock, ~T.~Cline, ~J.~Coloso, ~J.R.~Hodgson, ~J.F.~Kitchell, ~D.A.~Seekell, ~L.~Smith, ~B.Weidel.
\newblock Early warnings of regime shifts: a whole-ecosystem experiment.
\newblock {\em Nature}, 332(6033):1079--1082, 2011.

\bibitem[DRC{\etal}89]{Ditto/etal:1989}
W.L. Ditto, S.~Rauseo, R.~Cawley, C.~Grebogi, G.-H. Hsu, E.~Kostelich, E.~Ott,
  H.~T. Savage, R.~Segnan, M.~L. Spano, and J.~A. Yorke.
\newblock Experimental observation of crisis-induced intermittency and its
  critical exponents.
\newblock {\em Phys. Rev. Lett.}, 63(9):923--926, 1989.

\bibitem[DSvN{\etal}08]{dakos2008slowing}
V.~Dakos, ~M.~Scheffer, ~E.H.~Van Nes, ~V.~Brovkin ~V.~Petoukhov, ~H.~Held.
\newblock Slowing down as an early warning signal for abrupt climate change.
\newblock {\em PNAS}, 105(38):14308--14312, 2008.

\bibitem[FGJ18]{FuhrmannGroegerJaeger2015SNADimensions}
G.\ Fuhrmann, M.\ Gr{\"oger}, and T.~J{\"a}ger.
\newblock Non-smooth saddle-node bifurcations of forced monotone interval maps
  {II: Dimensions of strange attractors}.
\newblock {\em Ergodic Theory Dynam.\ Systems}, 38(8):2989--3011, 2018.

\bibitem[FJJK05]{FabbriJaegerJohnsonKeller2005ForcedSharkovskii}
R.~Fabbri, T.~J{\"a}ger, R.~Johnson, and G.~Keller.
\newblock A {S}harkovskii-type theorem for minimally forced interval maps.
\newblock {\em Topol.\ Methods Nonlinear Anal.}, 26:163--188, 2005.

\bibitem[FKP06]{FeudelKuzetsovPikovsky2006StrangeNonchaoticAttractors}
U.~Feudel, S.~Kuznetsov, and A.~Pikovsky.
\newblock {\em Strange nonchaotic attractors}, volume~56 of {\em World
  Scientific Series on Nonlinear Science. Series A: Monographs and Treatises}.
\newblock World Scientific Publishing Co. Pte. Ltd., Hackensack, NJ, 2006.
\newblock Dynamics between order and chaos in quasiperiodically forced systems.

\bibitem[Fuh16a]{fuhrmann2013NonsmoothSaddleNodesI}
G.~Fuhrmann.
\newblock Non-smooth saddle-node bifurcations of forced monotone interval maps
  {I: Existence of an SNA}.
\newblock {\em Ergodic Theory Dynam.\ Systems}, 36(4):1130--1155, 2016.

\bibitem[Fuh16b]{Fuhrmann2016SNAinFlows}
G.~Fuhrmann.
\newblock Non-smooth saddle-node bifurcations {III: Strange attractors in
  continuous time}.
\newblock {\em J.\ Diff.\ Eq.}, 261(3):2109--2140, 2016.

\bibitem[Fur61]{furstenberg:1961}
H.~Furstenberg.
\newblock Strict ergodicity and transformation of the torus.
\newblock {\em Am.\ J.\ Math.}, 83:573--601, 1961.

\bibitem[GJ13]{GroegerJaeger2013SNADimensions}
M.~Gr{\"o}ger and T.~J{\"a}ger.
\newblock Dimensions of attractors in pinched skew products.
\newblock {\em Comm. Math. Phys.}, 320(1):101--119, 2013.

\bibitem[GJT13]{guttal2013robustness}
V.~Guttal, ~C.~Jayaprakash, ~O.P.~Tabbaa.
\newblock Robustness of early warning signals of regime shifts in time-delayed ecological models.
\newblock {\em Theoretical ecology}, 6(3):271--283, 2013.

\bibitem[Gle02]{glendinning:2002}
P.~Glendinning.
\newblock Global attractors of pinched skew products.
\newblock {\em Dyn.\ Syst.}, 17:287--294, 2002.

\bibitem[GOPY84]{grebogi/ott/pelikan/yorke:1984}
C.~Grebogi, E.~Ott, S.~Pelikan, and J.A. Yorke.
\newblock Strange attractors that are not chaotic.
\newblock {\em Physica D}, 13:261--268, 1984.

\bibitem[HP06]{haro/puig:2006}
A.~Haro and J.~Puig.
\newblock Strange non-chaotic attractors in {H}arper maps.
\newblock {\em Chaos}, 16, 2006.

\bibitem[HY09]{huang/yi:2007}
W.~Huang and Y.~Yi.
\newblock Almost periodically forced circle flows.
\newblock {\em J. Funct. Anal.}, 257(3):832--902, 2009.

\bibitem[J{\"a}g03]{Jaeger2003NegativeSchwarzian}
T.~J{\"a}ger.
\newblock Quasiperiodically forced interval maps with negative {Schwarzian}
  derivative.
\newblock {\em Nonlinearity}, 16(4):1239--1255, 2003.

\bibitem[J{\"a}g07]{Jaeger2007StructureOfSNA}
T.~J{\"a}ger.
\newblock On the structure of strange nonchaotic attractors in pinched skew
  products.
\newblock {\em Ergodic Theory Dynam.\ Systems}, 27(2):493--510, 2007.

\bibitem[J{\"a}g09]{Jaeger2009CreationOfSNA}
T.~J{\"a}ger.
\newblock The creation of strange non-chaotic attractors in non-smooth
  saddle-node bifurcations.
\newblock {\em Mem.\ Am.\ Math.\ Soc.}, 945:1--106, 2009.

\bibitem[JS06]{JaegerStark2006Classification}
T.~J{\"a}ger and J.~Stark.
\newblock Towards a classification for quasiperiodically forced circle
  homeomorphisms.
\newblock {\em J.\ Lond.\ Math.\ Soc.}, 73(3):727--744, 2006.

\bibitem[K11]{kuehn2011mathematical}
C.~Kuehn.
\newblock A mathematical framework for critical transitions: Bifurcations, fast--slow systems and stochastic dynamics.
\newblock {\em Physica D: Nonlinear Phenomena}, 240(12):1020--1035, 2011.

\bibitem[Kel96]{keller:1996}
G.~Keller.
\newblock A note on strange nonchaotic attractors.
\newblock {\em Fundam.\ Math.}, 151(2):139--148, 1996.

\bibitem[KGB{\etal}14]{kefi2014early}
S.~K{\'e}fi, ~V.~Guttal, ~W.A.~Brock, ~S.R. ~Carpenter, ~A.M.~Ellison, ~V.N.~Livina, ~D.A.~Seekell, ~M.~Scheffer, ~E.H.~Van Nes, ~V.~Dakos.
\newblock Early warning signals of ecological transitions: methods for spatial patterns.
\newblock {\em PloS one}, 9(3):e92097, 2014.

\bibitem[LKK{\etal}15]{Lindneretal2015StrangeNonchaoticStars}
J.F. Lindner, V.~Kohar, B.~Kia, M.~Hippke, J.G. Learned, and W.L. Ditto.
\newblock Strange nonchaotic stars.
\newblock {\em Phys.\ Rev.\ Letters}, 114(5):054101, 2015.

\bibitem[MCA15]{MitsuiCrucifixAihara2015BifurcationsAndSNSinClimateModels}
T.~Mitsui, M.~Crucifix, and K.~Aihara.
\newblock Bifurcations and strange nonchaotic attractors in a phase oscillator
  model of glacial--interglacial cycles.
\newblock {\em Physica D: Nonlinear Phenomena}, 306:25--33, 2015.

\bibitem[NO07]{nunez/obaya:2007}
C.~{N\'u\~{n}ez} and R.~Obaya.
\newblock A non-autonomous bifurcation theory for deterministic scalar
  differential equations.
\newblock {\em Discr. Contin. Dyn. Syst., Ser. B}, 9(3--4):701--730, 2007.

\bibitem[RBO{\etal}87]{romeiras/etal:1987}
F.J. Romeiras, A.~Bondeson, E.~Ott, T.M. Antonsen~Jr., and C.~Grebogi.
\newblock Quasiperiodically forced dynamical systems with strange nonchaotic
  attractors.
\newblock {\em Physica D}, 26:277--294, 1987.

\bibitem[RDB{\etal}16]{rikkert2016slowing} M.G.M.~Olde Rikkert, ~V.~Dakos,
  ~T.G ~Buchman, ~R.~de Boer, ~L.~Glass, ~A.O.J.~Cramer, ~S.~Levin, ~E.H.~Van
  Nes, ~G.~Sugihara, ~M.D.~Ferrari, ~E.A.~Tolner, ~I.A.~Van de Leemput,
  ~J.~Lagro, ~R.~Melis, ~M.~Scheffer.  \newblock Slowing Down of Recovery as
  Generic Risk Marker for Acute Severity Transitions in Chronic Diseases.
  \newblock {\em Critical care medicine}, 44(3):601--606, 2016.

\bibitem[RRM15]{RizwanaRaja2015WienBrideOscillators}
R.~Rizwana and I.~Raja~Mohamed.
\newblock Investigation of chaotic and strange nonchaotic phenomena in
  nonautonomous wien-bridge oscillator with diode nonlinearity.
\newblock {\em Nonlinear Dynam.}, 2015, 2015.

\bibitem[SBB{\etal}09]{Schefferetal2009EWSforCT}
M.~Scheffer, J.~Bascompte, W.A. Brock, V.~Brovkin, S.R. Carpenter, V.~Dakos,
  H.~Held, E.H Van~Nes, M.~Rietkerk, and G.~Sugihara.
\newblock Early-warning signals for critical transitions.
\newblock {\em Nature}, 461(7260):53--59, 2009.


\bibitem[Sch09]{Scheffer2009CriticalTransitions}
M.~Scheffer.
\newblock {\em Critical transitions in nature and society}.
\newblock Princeton University Press, 2009.

\bibitem[SCL{\etal}12]{scheffer2012anticipating} M.~Scheffer,
   ~S.R. ~Carpenter, ~T.M.~Lenton, ~J.~Bascompte, ~W.A.~Brock, ~V.~Dakos,
   ~J.~Van de Koppel, ~I.A.~Van de Leemput, ~S.A.~Levin, ~E.H.~Van Nes,
   ~M.~Pascual, ~J.~Vandermeer.  \newblock Anticipating critical transitions.
   \newblock {\em Science}, 338(6105):344--348, 2012.

\bibitem[SS00]{sturman/stark:2000}
J.~Stark and R.~Sturman.
\newblock Semi-uniform ergodic theorems and applications to forced systems.
\newblock {\em Nonlinearity}, 13(1):113--143, 2000.

\bibitem[Sta03]{stark:2003}
J.~Stark.
\newblock Transitive sets for quasi-periodically forced monotone maps.
\newblock {\em Dyn.\ Syst.}, 18(4):351--364, 2003.


\bibitem[VFD{\etal}12]{veraart2012recovery}
Annelies~J Veraart, Elisabeth~J Faassen, Vasilis Dakos, Egbert~H van Nes,
  Miquel L{\"u}rling, and Marten Scheffer.
\newblock Recovery rates reflect distance to a tipping point in a living
  system.
\newblock {\em Nature}, 481(7381):357, 2012.

\bibitem[VLPR00]{Venkatesanetal2000SNAinDuffing}
A.~Venkatesan, M.~Lakshmanan, A.~Prasad, and R.~Ramaswamy.
\newblock Intermittency transitions to strange nonchaotic attractors in a
  quasiperiodically driven duffing oscillator.
\newblock {\em Physical Review E}, 61(4):3641, 2000.

\bibitem[vLWC{\etal}14]{van2014critical} I.A.~van de Leemput, ~M.~Wichers,
  ~A.O.J.~Cramer, ~D.~Borsboom, ~F.~Tuerlinckx, ~P.~Kuppens, ~E.H.~Van Nes,
  ~W.~Viechtbauer, ~E.J.~Giltay, ~S.H.~Aggen, ~C.~Derom, ~N.~Jacobs,
  ~K.S.~Kendler, ~H.L.J.~Van der Maas, ~M.C.~Neale, ~F.~Peeters, ~E.~Thiery,
  ~P.~Zachar, ~M.~Scheffer.  \newblock Critical slowing down as early warning
  for the onset and termination of depression. \newblock {\em PNAS},
  111(1):87--92, 2014.

\bibitem[vNS12]{van2007slow}
E.H.~Van Nes and M.~Scheffer.
\newblock Slow recovery from perturbations as a generic indicator of a nearby catastrophic shift.
\newblock {\em The American Naturalist}, 169(6):738--747, 2007.


\bibitem[WFP97]{witt/feudel/pikovsky:1997}
A.~Witt, U.~Feudel, and A.~Pikovsky.
\newblock Birth of strange nonchaotic attractors due to interior crisis.
\newblock {\em Physica D}, 109:180--190, 1997.

\bibitem[Zha13]{Zhang2013WadaSNA}
Y.~Zhang.
\newblock Strange nonchaotic attractors with wada basins.
\newblock {\em Physica D}, 259:26--36, 2013.

\end{thebibliography}
